\documentclass[12pt,reqno]{amsart}
\usepackage{amsthm,amsfonts,amssymb,euscript}

\newcommand{\bea}{\begin{eqnarray}}
\newcommand{\eea}{\end{eqnarray}}
\def\beaa{\begin{eqnarray*}}
\def\eeaa{\end{eqnarray*}}
\def\ba{\begin{array}}
\def\ea{\end{array}}
\def\be#1{\begin{equation} \label{#1}}
\def \eeq{\end{equation}}

\newcommand{\nn}{\nonumber}

\def\a{{\alpha}}

\def\b{{\beta}}
\def\be{{\beta}}
\def\ga{\gamma}

\def\de{\delta}
\def\De{\Delta}
\def\ep{\epsilon}
\def\eps{\epsilon}
\def\ka{\kappa}
\def\la{\lambda}

\def\si{\sigma}
\def\Si{\Sigma}
\def\om{\omega}
\def\Om{\Omega}
\def\th{\theta}
\def\ze{\zeta}
\def\ka{\kappa}
\def\nab{\nabla}

\def\pr{{\partial}}
\def\al{\alpha}

\def\rh{{\rho}}

\def\AA{{\mathcal A}}

\def\BB{{\mathcal B}}

\def\MM{{\mathcal M}}
\def\NN{{\mathcal N}}
\def\II{{\mathcal I}}

\def\FF{{\mathcal F}}

\def\HH{{\mathcal H}}
\def\LL{{\mathcal L}}
\def\GG{{\mathcal G}}

\def\WW{{\mathcal W}}

\def\SS{{\mathcal S}}
\def\Ss{{\mathcal S}}
\def\UU{{\mathcal U}}
\def\JJ{{\mathcal J}}

\def\RR{{\mathcal R}}
\def\QQ{{\mathcal Q}}
\def\AA{{\mathcal A}}
\def\VV{{\mathcal V}}

\def\B{{\bf B}}
\def\D{{\bf D}}

\def\M{{\bf M}}
\def\N{{\bf N}}

\def\O{{\bf O}}

\def\T{{\bf T}}
\def\E{{\bf E}}

\def\g{{\bf g}}

\def\SSS{{\mathbb S}}
\def\RRR{{\mathbb R}}

\def\f12{{\frac 1 2}}

\def\dual{{\,\,^*}}

\def\lb{{\,\underline{l}}}

\def\trch{{\mbox tr}\, \chi}
\def\trchb{{\mbox tr}\, \chib}
\def\chih{{\hat \chi}}
\def\chib{{\underline \chi}}

\def\etab{{\underline \eta}}
\def\omb{{\underline{\om}}}
\def\bb{{\underline{\b}}}
\def\aa{{\underline{\a}}}

\def\ul{{\underline{l}}}
\def\xib{{\underline{\xi}}}
\def\th{\theta}
\def\thb{{\underline{\theta}}}
\def\va{\vartheta}
\def\vab{{\underline{\vartheta}}}
\def\Omm{\overline{m}}
\def\Mm{m}
\def\Ll{\underline{l}}

\def\rhod{{\,\dual\rho}}
\def\f{\widetilde{f}}

\def\up{{u_+}}
\def\um{{u_-}}
\def\Lp{{L_+}}
\def\Lm{{L_-}}

\newtheorem{theorem}{Theorem}[section]
\newtheorem{lemma}[theorem]{Lemma}
\newtheorem{proposition}[theorem]{Proposition}

\newtheorem{definition}[theorem]{Definition}
\newtheorem{remark}[theorem]{Remark}

\numberwithin{equation}{section}

\begin{document}

\title[On the  uniqueness  of smooth, stationary black holes in vacuum]
{On the  uniqueness  of smooth, stationary black holes in vacuum}
\author{Alexandru D. Ionescu}
\address{University of Wisconsin -- Madison}
\email{ionescu@math.wisc.edu}
\author{Sergiu Klainerman}
\address{Princeton University}
\email{seri@math.princeton.edu}
\thanks{The first author was supported in part by an NSF grant and a Packard Fellowship. The second author was partially supported by NSF grant DMS-0070696.}
\begin{abstract}
A fundamental conjecture in  General Relativity
asserts that the domain of outer communication
of a regular, stationary, four dimensional,  vacuum black hole solution is isometrically diffeomorphic to the domain of outer communication of a Kerr black hole.  So far the conjecture has been
resolved, by combining results of Hawking \cite{H-E},
  Carter \cite{Ca1} and Robinson \cite{Rob},   under the additional hypothesis of  non-degenerate
horizons and \textit{real analyticity}  of the space-time.
We develop a new strategy to bypass  analyticity  
based on a  tensorial  characterization of  the Kerr solutions, due to Mars
\cite{Ma1},  and new  geometric Carleman
 estimates. 
 We prove, under a technical  assumption (an identity  relating the Ernst potential  and the Killing scalar)  on the bifurcate sphere of the event horizon,   that the domain of outer communication  of a smooth, regular, stationary Einstein vacuum spacetime of dimension $4$ is locally isometric to the domain of outer communication of a Kerr spacetime.
  \end{abstract}
\maketitle
\tableofcontents

\section{Introduction} 
A fundamental conjecture in  General Relativity\footnote{See reviews 
by  B. Carter \cite{Ca-R} and   P. Chusciel   \cite{Chrusc-Rev}, \cite{Chrusc-Rev2}, 
 for a history  and  review of the current status of the conjecture.}
asserts that the domains of outer communication
of regular\footnote{The notion of regularity 
needed here  requires a careful discussions
 concerning the geometric hypothesis 
 on the space-time. }, stationary, four dimensional,  vacuum black hole solutions
are isometrically diffeomorphic to those of 
Kerr black holes.  One expects,  due to gravitational radiation,
 that general, asymptotically flat, dynamic,  solutions of the Einstein-vacuum
  equations settle down, asymptotically,  into a stationary regime. A similar scenario is supposed to hold true in the presence of matter. Thus
  the conjecture, if true, would characterize all possible asymptotic states 
  of the general evolution.

 So far the conjecture has  been
resolved, by combining results of Hawking \cite{H-E},
 Carter \cite{Ca1}, and Robinson \cite{Rob},   under the additional hypothesis of  non-degenerate
horizons and \textit{real analyticity}  of the space-time.  
The assumption of real analyticity  is both hard to justify 
and difficult to dispense of.   
 One  can show, using standard elliptic theory, 
that stationary  solutions are real analytic in regions 
where the corresponding  Killing vector-field  $\T$  is time-like, 
but there is no reason to expect the same result to hold true 
in the ergo-region (in Kerr,
  the Killing  vector-field $\T$, which is time-like in 
the asymptotic region, becomes space-like in the ergo-region).
In view of the main application of the conjectured result to the
 general problem of evolution,  mentioned above, 
 there is also no reason to expect that, by losing gravitational radiation,
 general  solutions become, somehow, analytic. Thus the assumption
  of analyticity is a  serious limitation of the  present uniqueness results.
  Unfortunately   one of  the  main step in the current
   proof,
  due to Hawking \cite{H-E}, depends heavily on  analyticity. As we 
  argue below, to extend Hawking's argument to a smooth setting
   requires solving an \textit{ill posed problem}.  Roughly speaking 
  Hawking's argument is based on the observation that, though a general
  stationary space  may seem  quite complicated, its behavior along the
   event horizon is remarkably  simple. Thus Hawking has shown that
   in addition to the original, stationary, Killing field, which has to be tangent
   to the event horizon, there must exist, infinitesimally along the horizon,  an additional Killing  vector-field.   To   extend this information, from the event
   horizon to the domain of outer communication, requires
   one to solve a boundary value problem, with data on the horizon,  for a linear 
   differential equation.  Such problems are  typically \textit{ill posed} (i.e.  solutions may  fail to exist in the smooth category.)   In the analytic category, however, the problem  can be solved by a
    straightforward  Cauchy-Kowalewsky type argument.
    Thus, by  assuming  analyticity for the stationary metric,  Hawking    bypasses this  fundamental difficulty,  and  thus is able  to extend this additional Killing field to the entire domain of outer communication.  As a consequence,  the   space-time under consideration is not just stationary
    but also axi-symmetric, situation     for which  Carter-Robinson's  uniqueness theorem \cite{Ca1}, \cite{Rob} applies.  It is interesting to remark that  this   final  step  does not require analyticity.
    
    Though ill posed problems do not, in general, admit solutions,
    one can, when a solution is known to exist,  often 
    prove uniqueness (we refer the reader to the  introduction in  \cite{Ion-K1} for
    a more thorough discussion of this issue).  This fact has led us to develop a different     strategy for proving  uniqueness   based on  a characterization of  the Kerr solution,  due to Mars \cite{Ma1},  and   geometric  Carleman estimates
     applied to covariant wave  equations
      on a general, stationary,  black hole background. We discuss 
    this strategy in more details in the following subsection, after we recall 
    a few basic definitions  and results  concerning stationary black holes.
    Our main result, stated in subsection \ref{maintheorem} below, 
    proves uniqueness of the Kerr family among all, smooth, appropriately 
    regular, stationary   solutions, with  a regular, bifurcate, event horizon, under  an additional  assumption which has to be satisfied along the bifurcate  sphere $S_0$
    of the event horizon.  More precisely  we assume  a pointwise   complex scalar identity relating the Ernst potential $\sigma$ and the Killing scalar $\FF^2$ on $S_0$.

     \subsection{Stationary, regular,  black holes}
        In this subsection we review some of the main definitions and results concerning stationary black holes (see also the discussion in the introduction to section 3. We will also give a more detailed discussion  of our  new  approach to  the problem of uniqueness.   Precise assumptions concerning our result will be made 
only         in the next subsection.

The main objects in the  theory of stationary, vacuum, black holes are    $3+1$ dimensional space-times  $(\M,\g)$ which  are smooth, strongly causal, time oriented,    solutions  of the  Einstein vacuum equations, see \cite{H-E} for  precise definitions,  and which are  also  \textit{stationary}, 
\textit{asymptotically flat}.
More precisely one  considers,  see  for example page 2 in  \cite{FrRaWa},  space-times
 $(\M,\g)$ endowed with  a 1-parameter group of isometries $\Phi_t$,
generated by a  Killing vector-field $\T$, and  which possess
a smooth space-like   slice $\Si_0$  with an asymptotically flat end $\Si_0^{(end)}\subset \Si_0$ on 
  which $g(\T,\T)<0$.  To ensure strong causality  we assume that 
  $\M$ is the maximal globally hyperbolic extension of   $\Si_0$.  This implies,
    in particular, that all  orbits of $\T$ are complete,
    see \cite{Chrusc0}, and must intersect $\Si_0$,
     see \cite{CW}.
     Define $\M^{(end)}=\cup_{t\in\RRR}\Phi_t(\Si_0^{end})$. Take $\B$ to be the
   complement of $\II^{-}(\M^{(end)})$,  ${\bf W}$ the complement of $\II^{+}(\M^{(end)})$, where
    $\II^{\pm}(S)$ denote the causal future and past sets of a set $S\subset \M$. In other words $\B$ (called the 
  \textit{black hole} region), respectively ${\bf W}$ (called the \textit{white hole} region), is the set of points in $\M$ for which no future directed, respectively past directed, 
  causal curve  meets $\M^{(end)}$.
  Also we take 
    $\E$ (called  domain of outer communication)  the complement of ${\bf W}\cup \B$, i.e. 
   $\E=\II^{-}(\M^{(end)})\cap\II^{+}(\M^{(end)})$. We further  define 
  the future event horizon   $\HH^+$ to be  the boundary of $\II^{-}(\M^{(end)})$ and
   the past event horizon  $\HH^{-}$ to be  the boundary of $\II^{+}(\M^{(end)})$,
   \beaa
   \HH^+=\de\B,\qquad \HH^-=\de {\bf W}.
   \eeaa
   By definition both  $\HH^+$ and $\HH^-$ are achronal (i.e. no two points on  $\HH^+ $,
   or $\HH^-$  can be connected by time-like curves)
   boundaries generated by null geodesic segments.
     According to the topological censorship
   theorem, see   \cite{CW2} or \cite{FSW},  the domain of outer communication $\E$ is
   simply connected. This implies that  all  connected components of
   event horizons must have the topology
   of  $\SSS^2\times \RRR$. In our work we  shall  assume    that the event horizon has only one component.    
   
   It follows immediately from the definitions
   above that the flow $\Phi_t$ must keep
   $\HH^+$ and $\HH^-$ invariant, therefore 
    the generating  vector-field $\T$ must be tangent 
   to $\HH$.  One further assumes
    that $\Phi_t$  has no fixed points on 
    $\HH$ with the possible exception of 
    $S_0=\HH^+\cap  \HH^-$.  Then either $\T$ is space-like or null  at all   points of  $\HH$.
   If  $\T$ is null on $\HH$, in which case   $\HH$  is said to be   a Killing horizon for $\T$, 
      Sudarski-Wald \cite{Su-Wa} have proved that
      the space-time must be static, i.e. $\T$ is hypersurface orthogonal.  Static solutions, on the other hand,  are known to be isomorphic 
      to    Schwarzschild metrics, see \cite{I}, \cite{Bu-M} and \cite{Chrusc}.
       In this paper we are interested  only  in the case  when    $\T$ is space-like at some points on the horizon.

   The existence of partial Cauchy hypersurface 
   $\Si_0$ implies, in particular,  the existence of a foliation $\Si_t$ on 
    $\E$,  which induces a foliation $S_t$ on
     the horizon $\HH$ with  a well defined area.  
     A key result of Hawking
     \cite{H-E} (see also   \cite{CGD} where  the  area theorem is proved under 
    very general  differentiability  assumptions),  shows that the area of $S_t$
     is a monotonous function of $t$.      
      Using this fact, together with the tangency of the Killing field $\T$,  one can show that  the null second fundamental forms of  both   $\HH^+$ and $\HH^-$ must vanish identically, see  \cite{H-E}. Specializing to the future event horizon  $\HH^+$, Hawking \cite{H-E} (see also \cite{IsMon}) has proved the existence
      of a  non-vanishing vector-field $K$, tangent to the null generators of $\HH^+$ which is Killing to any order along $\HH^+$. Moreover $\D_KK=\kappa K$ with $\kappa$,      constant  along $\HH^+$, called the surface gravity of $\HH^+$. If $\ka\neq 0$ we say that $\HH^+$
      is \textit{non-degenerate}.        
      In the non-degenerate case the work of  Racz and Wald
       \cite{Ra-Wa} supports the hypothesis, which we make  in  our work  (see  next subsection),   that $\HH^+$ and $\HH^-$ are smooth null
      hypersurfaces intersecting smoothly on a $2$ surface  $S_0$ with the topology of the standard sphere.   We say,
      in this case,  that the horizon $\HH$ is  a smooth  \textit{bifurcate horizon}. 
      
      Under the restrictive assumption of real analyticity 
      of the metric $g$  one can show, see \cite{H-E}
      and \cite{Chrusc2}, that the Hawking 
       vector-field $K$ can be extended to a neighborhood  of the entire domain       of communication\footnote{In \cite{FrRaWa} it is shown that   $K$ can be extended 
   in the complement of the domain of outer communication $\E$  without the  restrictive analyticity assumption. However their argument does not apply
    to the domain of outer communication  $\E$. }. One can then show that the spacetime
    $(\M, \g)$ is not just stationary but also
    axi-symmetric. One can then appeal to the   results of Carter \cite{Ca1}  and
    Robinson \cite{Rob} which  show   that the family of Kerr
    solutions with $0\le a<m$ exhaust the class of non-degenerate, stationary axi-symmetric, connected, four
    dimensional, vacuum black holes. This concludes the present
    proof of uniqueness, based on analyticity.

     Without analyticity any hope
    of extending $K$ outside $\HH$, in $\E$,  by a direct argument
    encounters a fundamental difficulty.  Indeed one needs 
    to extend  $K$ such that it satisfies the Killing equation,
    \bea
    \D_\mu K_\nu+\D_\mu K_\nu=0.\label{Killing}
    \eea
    Differentiating the Killing equation
    and using the Ricci flat condition $\mbox{Ric}(\g)=0$
    one derives the covariant wave equation $
    \square_\g K=0$. 
    The obstacle  we encounter  is that the boundary value problem
    $\square_\g K=0$ with $K$ prescribed on $\HH$
    is   \textit{ill posed}, which means that it is impossible
    to extend $K$ by solving $\square_\g K=0$, if the metric
    is smooth but fails to be  real analytic.    To understand the 
    ill posed character of the situation  it helps to  consider 
    the following simpler model problem in 
      the domain   $\E=\{(t,x)\in \RRR^{1+3}/ |x|> 1+|t|\}$ of  Minkowski    space $\RRR^{1+3}$.
     \bea
     \square \phi=F(\phi, \pr\phi),\qquad \phi|_{\de\E}=\phi_0.\label{model}
     \eea
     Here $\square $ is the usual D'Alembertian of 
     $\RRR^{1+3}$ and $F$ a smooth function of $\phi$ and  its partial  derivatives $\pr_\a\phi$, vanishing for 
     $\phi=\pr\phi=0$.  One can regard $\E$ as a model
     of the domain of outer communication and its boundary $\HH=\de\E$ as analogous to the bifurcate
     event horizon considered above. 
     The problem is  still ill posed; even in the case $F\equiv 0$ we cannot, in general,  find solutions for arbitrary  smooth boundary     data $\phi_0$. Yet, as 
     typical to many ill posed problems, even if existence fails we can still prove uniqueness. In other words
     if \eqref{model} has two solutions $\phi_1,\phi_2$
    which agree on $\HH=\de \E$ then they must coincide 
    everywhere in $\E$, see \cite{Ion-K1}. The result is based on Carleman estimates, i.e. on space-time  $L^2$ a-priori estimates   with carefully chosen weights. A more
    realistic  model problem is to consider smooth  space-time metrics $\g$ in $\RRR^{1+3}$ which verify the Einstein vacuum equations and agree, up to
     curvature,  with the  standard    Minkowski metric 
     on the boundary $\HH=\de\E$. Can we prove 
     that $\g$ must be flat also  in $\E$ ?
       It is easy to see,
     using the Einstein equations, that the 
     Riemann curvature tensor  $R$ of such  metrics 
     must verify a  covariant wave equation of the form
     $\square_\g R=R* R$, with $R* R$ denoting an
     appropriate quadratic  product of components of $R$.
     We are thus  led to a question  similar
      to the one above;  knowing that $R$ vanishes
      on the boundary of $\E$ can we deduce that
      it also vanishes  on $\E$ ?   Using methods 
      similar to those of \cite{Ion-K1} we can prove
      that $R$ must vanish in a neighborhood of
      $\HH$. We also expect that, under additional 
      global assumptions on the metric  $\g$, 
      one can show  that $R$ vanishes everywhere
      on $\E$ and therefore $\g$ is locally
       Minkowskian. 
       
       These   considerations  lead us to 
       look for a tensor-field $\SS$, associated  to our  stationary   metric $\g$,  which satisfies the following 
       properties.
       \begin{enumerate}
       \item If $\SS$ vanishes in $\E$ then the metric
       $\g$ is locally isometric to  a Kerr solution.
       \item $\SS$ verifies a covariant wave equation
       of the form,
        \bea
       \square_\g \SS=\AA*\SS+\BB*\D\SS, \label{wave-SS}
       \eea
       with $\AA$ and $\BB$ two arbitrary 
       smooth tensor-fields.
       \item $\SS$ vanishes identically  on the bifurcate 
       event  horizon $\HH$.
       
       \end{enumerate}
       An appropriate space-time  tensor verifying condition (1) has been proposed
       by M. Mars in \cite{Ma1}, based on some previous
       work of W.  Simon \cite{Sim};  we  refer to it as  the Mars-Simon tensor.   In this paper we shall
       show that $\SS$ verifies the desired wave equation
       in (2) and give a sufficient, simple condition on the bifurcate sphere $S_0$,
       which insures that $\SS$ vanishes 
       on  the event horizon $\HH$. We then
       prove, based  on a global unique continuation  argument, 
        that $\SS$ must vanish everywhere in 
       the domain of outer communication $\E$.            
       In view of Mars's result \cite{Ma1} we deduce
       that $\E$ is locally isometric with 
       a Kerr solution.  
       
         The unique continuation strategy  is based 
       on two Carleman estimates.
       The first one  establishes the vanishing of 
       solutions to covariant wave equations, 
       with zero boundary conditions on a neighborhood  of $S_0$ on the   event
       horizon, to a full space-time neighborhood 
       of $S_0$. The proof of this result can  be extended to the exterior of  a regular,  bifurcate 
       null hypersurface (i.e. with a regular  bifurcate sphere),   in a general, smooth,
       Lorentz manifold. Our second, conditional,  Carleman estimate  is significantly  deeper as it depends  heavily  on the specific properties of  stationary solutions        of the Einstein vacuum equations. We use  it,  together with  an appropriate  bootstrap argument, to extend   the region of 
        vanishing of the 
       Mars -Simon tensor from  a neighborhood 
       of $S_0$ to the entire domain of outer
       communication $\E$.        The proof of both
       Carleman estimates (see also discussion in the first subsection of section 3), but especially the  second,  rely on   calculations 
       based on null frames and complex
         null tetrads. We develop our own formalism,
         which is, we hope,  a  useful  compromise between 
         that  of   Newmann-Penrose \cite{N-P} and that used     in  \cite{CKl},    \cite{KlNi}.  Strictly speaking the formalism 
         used in      \cite{CKl}  does not apply in the situation studied here as it presupposes that 
         the horizontal distribution generated by 
    the null pair is integrable.  The horizontal distribution generated by the 
    principal null directions in Kerr do  not  verify this property.

\subsection{Precise assumptions and the Main Theorem}\label{maintheorem}
We state now our precise assumptions. We  assume that  $(\M,\g)$ is a smooth\footnote{$\M$ is assumed to be a connected, orientable, paracompact $C^\infty$ manifold without boundary.}, time oriented, vacuum Einstein spacetime of dimension $3+1$ and $\T\in\T(\M)$ is a smooth Killing vector-field on $\M$. In addition, we make the following assumptions and definitions.

{\bf{AF.}} (Asymptotic flatness) We assume that there is an open subset $\M^{(end)}$ of $\M$ which is diffeomorphic to $\mathbb{R}\times(\{x\in\mathbb{R}^3:|x|>R\})$ for some $R$ sufficiently large. In local coordinates $\{t,x^i\}$ defined by this diffeomorphism, we assume that, with $r=\sqrt{(x^1)^2+(x^2)^2+(x^3)^2}$,
\begin{equation}\label{As-Flat}
\g_{00}=-1+\frac{2M}{r}+O(r^{-2}),\quad \g_{ij}=\delta_{ij}+O(r^{-1}),\quad\g_{0i}=O(r^{-2}),
\end{equation}
for some $M>0$, and
\begin{equation*}
\T=\partial_t\text{ therefore }\partial_t\g_{\mu\nu}=0.
\end{equation*}
We define the domain of outer communication (exterior region)
\begin{equation*}
\E=\II^{-}(\M^{(end)})\cap\II^{+}(\M^{(end)}).
\end{equation*}
We assume that there is an imbedded space-like hypersurface $\Sigma_0\subseteq\M$ which is diffeomorphic to $\{x\in\mathbb{R}^3:|x|>1/2\}$ and, in $\M^{(end)}$, $\Sigma_0$ agrees with the hypersurface corresponding to $t=0$. Let $T_0$ denote the future directed unit vector orthogonal to $\Sigma_0$. We assume that every orbit of $\T$ in $\E$ is complete and intersects the hypersurface $\Sigma_0$, and
\begin{equation}\label{nontang}
|\g(\T,T_0)|>0\text{ on }\Sigma_0\cap\E.
\end{equation}
\medskip

{\bf{SBS.}} (Smooth bifurcate sphere) Let
\begin{equation*}
S_0=\delta(\II^{-}(\M^{(end)}))\cap\delta(\II^+(\M^{(end)})).
\end{equation*}
We assume that $S_0\subseteq\Sigma_0$ and $S_0$ is an imbedded $2$-sphere which agrees with the sphere of radius $1$ in $\mathbb{R}^3$ under the identification of $\Sigma_0$ with $\{x\in\mathbb{R}^3:|x|>1/2\}$. Furthermore, we assume that there is a neighborhood $\mathbf{O}$ of $S_0$ in $\mathbf{M}$ such that the sets
\begin{equation*}
\HH^+=\mathbf{O}\cap \delta(\II^{-}(\M^{(end)})\quad \text{ and }\quad \HH^-=\mathbf{O}\cap \delta(\II^{+}(\M^{(end)})
\end{equation*}
are smooth imbedded hypersurfaces diffeomorphic to $S_0\times(-1,1)$, We assume that these hypersurfaces are null, non-expanding\footnote{A null hypersurface is said to be non-expanding if the trace of its null second fundamental form vanishes identically.}, and intersect transversally in $S_0$. Finally, we assume that the vector-field $\T$ is tangent to both hypersurfaces $\HH^+=\mathbf{O}\cap \delta(\II^{-}(\M^{(end)}))$ and $\HH^-=\mathbf{O}\cap \delta(\II^{+}(\M^{(end)}))$, and does not vanish identically on $S_0$\footnote{In view of a well known result, see  \cite{FS}, any  non-vanishing Killing field  on $S_0$ can only vanish at a finite number of isolated points.}.
\medskip

{\bf{T.}} (Technical assumptions). Let $F_{\a\b}=\D_\a\T_\b$ denote the Killing form on $\mathbf{M}$, and $\FF_{\al\be}=F_{\a\b}+i\dual F_{\a\b}$, where $\dual F_{\a\b}=\frac 1 2 \in_{\a\b\ga\de}F^{\ga\de}$. Let $\FF^2=\FF_{\a\b}\FF^{\a\b}$. The Ernst 1-form associated to $\T$ is defined as $\si_\mu=2 \T^\a\FF_{\a\mu}$. It is easy to check, see  equation \eqref{Ernst1}, that  $\si_\mu$ is exact and, therefore, there exists a complex scalar $\si$ defined in an open neighborhood of $\Sigma_0$, called the Ernst potential, such that $\D_\mu \si=\si_\mu$. In view of the asymptotic flatness assumption {\bf{AF}}, we can choose $\si$ such that $\si\to 1$ at infinity along $\Sigma_0$. Our main technical assumptions are 

\begin{equation}\label{Main-Cond1}
-4M^2\FF^2=(1-\si)^4\quad\text{ on }S_0,
\end{equation}
and
\begin{equation}\label{Main-Cond2}
\qquad \qquad\qquad\Re\big((1-\si)^{-1}\big)>1/2\quad\,\text{  at some point on }S_0.
\end{equation}
\medskip
\begin{remark}
  As we have discussed  in the 
  previous subsection   some  of the assumptions 
  made above have been deduced from
  more primitive assumptions. For example,
  the completeness of orbits of $\E$ can be deduced by assuming that $\M$  is the maximal global
  hyperbolic extension of   $\Si_0$, see \cite{Chrusc0}.    Our  precise
  space-time   asymptotic flatness  conditions can be deduced  by making  asymptotic flatness  assumptions  only on $\Si_0$, see \cite{Be-Si},
   \cite{Be-Si2}.      The assumption \eqref{nontang} can be replaced, at the expense of some additional work in section \ref{lastsection}, by a suitable regularity assumption on the space of orbits of $\T$.
  The non-expanding condition in {\bf SBS} can
  be derived  using the area theorem, see  \cite{H-E},  \cite{CGD}. 
The regular  bifurcate structure of the horizon,
    assumed in   {\bf SBS},  is connected
   to the more primitive assumption of  non-degeneracy   of the horizon, see \cite{Ra-Wa}.
\end{remark} 
    
\begin{remark}    
   Assumption  \eqref{Main-Cond2}  is consistent with 
  the natural condition $0\le a<M$ satisfied by the two parameters of the Kerr family.   The key technical assumption in this paper is the identity \eqref{Main-Cond1}, which is assumed to hold on the bifurcate sphere $S_0$. This  assumption is made in 
  order to insure that the corresponding Mars-Simon
  tensor vanishes on $\HH^{-}\cup\HH^+$. We emphasize, however, that we do not make any technical assumptions in the open set $\E$ itself; the identity \eqref{Main-Cond1} is only assumed to hold on the bifurcate sphere $S_0$, which is a codimension $2$ set, while the inequality \eqref{Main-Cond2} is only assumed at one point of $S_0$. We hope to further relax these technical conditions and interpret them as part of the ``regularity'' assumptions on the black hole in future work. 
\end{remark}

\begin{remark}\label{Kerrexplicit} 
In Boyer-Lindquist coordinates the Kerr metric 
takes the form,
\bea
ds^2=-\frac{\rho^2\Delta}{\Sigma^2}(dt)^2+\frac{\Sigma^2(\sin\theta)^2}{\rho^2}\Big(d\phi-\frac{2aMr}{\Sigma^2}dt\Big)^2+\frac{\rh^2}{\Delta}(dr)^2+\rho^2(d\theta)^2,
\label{metric-Kerr}
\eea
where, 
\beaa
\rho^2&=&r^2+a^2\cos^2\th,\qquad
\De= r^2+a^2-2Mr,\qquad
\Sigma^2=(r^2+a^2)\rho^2+2Mra^2(\sin\theta)^2.
\eeaa
On the horizon we have $r=r_+:=M+\sqrt{M^2-a^2}$ and $\De=0$.   The domain of outer communication $\mathbf{E}$ is given by $r> r_+$. 
One can show that  the   complex Ernst potential $\si$ and the complex scalar
 $\FF^2$   are  given by
\bea
\si&=&1-\frac{2M}{r+ia\cos\th},\qquad \FF^2=-\frac{4M^2}{(r+ia\cos\th)^4}.
\eea
Thus,
\bea
-4M^2\FF^2&=&(1-\si)^4
\eea
everywhere in the exterior region. 
Writing $y+iz:=(1-\si)^{-1}$ we  observe that,
\beaa
y=\frac{r}{2M}\ge \frac{r_+}{2M}>\frac 12.
\eeaa
everywhere in the exterior region. 
\end{remark}

\noindent {\bf Main Theorem}.\quad \textit{Under the assumptions {\bf{AF}}, {\bf{SBS}}, and {\bf{T}} the domain of outer communication $\E$ of $\M$ is  locally isometric to the domain of outer communication of a Kerr space-time with mass $M$ and $0<a<M$.}

\medskip 

As  mentioned earlier, the basic idea of the  proof is to show that the Mars-Simon tensor is well-defined and vanishes in the entire domain of outer communication, by relying on Carleman type estimates. We provide below a more detailed outline of the proof.

In section \ref{generalCarleman}, we prove a sufficiently general geometric Carleman inequality, Proposition \ref{Cargen}, with weights that satisfy suitable conditional pseudo-convexity assumptions. This Carleman inequality is applied in section \ref{sectionS=01} to prove Proposition \ref{Lemmaa1} and section \ref{lastsection} to prove Proposition \ref{vanishingS3}.  

In section \ref{MarsSimon} we define, in a simply connected neighborhood $\widetilde{\mathbf{M}}$ of $\Si_0\cap \overline{\E}$, the Killing form $\mathcal{F}_{\al\be}$ and the Ernst potential $\sigma$. We then introduce the Mars-Simon tensor, see \cite{Ma1}, 
\begin{equation*}
\Ss_{\al\be\mu\nu}=\RR_{\al\be\mu\nu}+6(1-\sigma)^{-1}\big(\FF_{\al\be}\FF_{\mu\nu}-(1/3)\FF^2\II_{\al\be\mu\nu}\big)
\end{equation*}
as a self-dual Weyl tensor, which is well defined and smooth in the open set
\begin{equation*}
\mathbf{N}_0=\{x\in\widetilde{\mathbf{M}}:1-\sigma(x)\neq 0\}.
\end{equation*}
It is important to observe that $\mathbf{N}_0$ contains a neighborhood of the bifurcate sphere $S_0$, since $\Re\sigma=-\T^\al\T_\al$, which is nonpositive on $S_0$. In particular, the Mars-Simon tensor is well defined in a neighborhood of $S_0$. The  main  result of the section, stated in Theorem \ref{thm-MS-eqts}, is the identity 
\begin{equation}\label{ro10}
\D^\si\SS_{\si\a \mu\nu}=\JJ(\SS)_{\a\mu\nu}=-6(1-\sigma)^{-1}\T^\la\SS_{\la\rho\ga\de}\big(\FF_{\a}^{\,\,\,\rho}\de^\ga_\mu\de^\de_\nu-(2/3)\FF^{\ga\de}\II_{\a\,\,\mu\nu}^{\,\,\,\rho}\big),
\end{equation} 
which shows that $\SS$ verifies a divergence equation  with a source term $\JJ(\SS)$  proportional to $\SS$.
  It is then straightforward to deduce, see Theorem
  \ref{wave}, that $\SS$ verifies a covariant wave equation with a source proportional to $\SS$ and  first derivatives of $\SS$. 

In section \ref{horizon} we show that $\SS$ vanishes on the horizon $\delta(\II^{-}(\M^{(end)}))\cup\delta(\II^{+}(\M^{(end)}))$, in a neighborhood of the bifurcate sphere $S_0$. The proof depends on special properties of the horizon, such as the vanishing of the null second fundamental forms and certain null curvature components, and the divergence equation \eqref{ro10}. The proof also depends on the main technical assumption \eqref{Main-Cond1} to show that the component $\rho(\Ss)$ vanishes on $S_0$ (this is the only place where this technical assumption is used).   

In section \ref{sectionS=01} we show that $\SS$ vanishes in a full space-time neighborhood $\O_{r_1}\cap \E$ of $S_0$ in $\E$, see Proposition \ref{Lemmaa1}. For this we derive the Carleman inequality of Lemma \ref{Car}, as a consequence of the more general Proposition \ref{Cargen}. The weight function used in this Carleman inequality is constructed with the help of two optical functions $u_+$ and $u_-$, defined in a space-time  neighborhood of $S_0$. We then apply this Carleman inequality to the covariant wave equation verified by $\SS$, to prove Proposition \ref{Lemmaa1}. 
 
Once we have regions of space-time in which $\SS$ vanishes we can rely on some of the remarkable computations of Mars \cite{Ma1}. In section \ref{section2} we work in an open set $\mathbf{N}\subseteq\mathbf{N}_0$ (thus $1-\sigma\neq 0$ in $\mathbf{N}$), $S_0\subseteq \mathbf{N}$, with the property that $\Ss=0$ in $\mathbf{N}\cap\mathbf{E}$ and $\mathbf{N}\cap\mathbf{E}$ is connected. Such sets exist, in view of the main result of section \ref{sectionS=01}. Following Mars \cite{Ma1}, we define the real functions $y$ and $z$ in $\mathbf{N}$ by
\begin{equation*}
y+iz=(1-\si)^{-1},
\end{equation*}
see Remark \ref{Kerrexplicit} for explicit formulas in the Kerr spaces. The function $y$ satisfies the important identity \eqref{se30}, found by Mars,
\begin{equation}\label{ro11}
\D_\a y\D^\a y=\frac{y^2-y+B}{4M^2(y^2+z^2)}
\end{equation}
in $\mathbf{N}\cap\mathbf{E}$, where $B\in[0,\infty)$ is a constant which has the additional property that $z^2\leq B$ in $\mathbf{N}\cap\mathbf{E}$ (in the Kerr space $B=a^2/(4M^2)$). We then use this identity and the fact that $\Re(1-\sigma)=1+\g(\T,\T)$ to prove the key bound on the coordinate norm of the gradient
\begin{equation}\label{ro12}
|D^1y|\leq\widetilde{C}\text{ in }\mathbf{N}\cap\mathbf{E},
\end{equation}       
with a uniform constant $\widetilde{C}$ (see Proposition \ref{uniformy}). This bound, together with $z^2\leq B$, shows that the function $1-\sigma=(y+iz)^{-1}$ cannot vanish in a neighborhood of the closure of $\mathbf{N}\cap\mathbf{E}$, as  long as $\Ss=0$ in $\mathbf{N}\cap\mathbf{E}$ and $\mathbf{N}\cap\mathbf{E}$ is connected. This observation is important in section \ref{lastsection}, as part of the bootstrap argument, to show that $1-\sigma\neq 0$ in $\Sigma_0\cap\E$. Finally, in Lemma \ref{coefficients} we work in a canonical complex null tetrad and compute the Hessian $\D^2y$ in terms of the functions $y$ and $z$, and the connection coefficient $\ze$.
         
In section \ref{lastsection} we use a bootstrap argument to complete the proof of the Main Theorem. Our main goal is to show that $1-\sigma\neq 0$ and $\Ss=0$ in $\Sigma_0\cap\E$. We start by showing that $y=y_{S_0}$ is constant on the bifurcate sphere $S_0$, and use \eqref{ro11} to show that $y^2_{S_0}-y_{S_0}+B=0$; using \eqref{Main-Cond2} it follows that $B\in[0,1/4)$ and $y_{S_0}\in(1/2,1]$. We use then the wave equation
\begin{equation*}
\D^\al\D_\al y=\frac{2y-1}{4M^2(y^2+z^2)},
\end{equation*}
which is a consequence of $\Ss=0$, and the fact that $y_{S_0}>1/2$, to show that $y$ must increase in a small neighborhood  $\O_\ep\cap \E$. We can then start our bootstrap argument: for $R>y_{S_0}$ let $\mathcal{U}_R$ denote the unique connected component of the set $\{x\in\Sigma_0\cap\E:\sigma(x)\neq 1\text{ and }y(x)<R\}$ whose closure in $\Sigma_0$ contains $S_0$. We need to show, by induction over $R$, that $\Ss=0$ in $\mathcal{U}_R$ for any $R>y_{S_0}$; assuming this, it would follow from \eqref{ro12} that $\sigma\neq 1$ in $\Sigma_0\cap\E$ and $\cup_{R>y_{s_0}}\mathcal{U}_{R}=\Sigma_0\cap\E$, which would complete the proof of the Main Theorem. The key inductive step in proving that $\Ss=0$ in $\mathcal{U}_R$ is to show that if $x_0$ is a point on the boundary of $\mathcal{U}_R$ in $\Sigma_0\cap\E$, and if $\Ss=0$ in $\mathcal{U}_R$, then $\Ss=0$ in a neighborhood of $x_0$ (see Proposition \ref{vanishingS3}). For this we use a second Carleman inequality, Lemma \ref{Carl2}, with a weight that depends on the function $y$. To prove this second Carleman estimate we use the general Carleman estimate Proposition \ref{Cargen} and the remarkable pseudo-convexity properties of the Hessian of the function $y$ computed in Lemma \ref{coefficients}.

We would like to thank  P. Chrusciel, M. Dafermos, J. Isenberg, M. Mars and R. Wald for helpful conversations connected to our work. We would also like to thank the referees for very helpful comments, particularly on section \ref{generalCarleman}.

\section{Geometric preliminaries}\label{preliminaries}
\subsection{Optical functions}\label{subs:optical}
We    define  two optical functions $\um, \up$  in  a neighborhood of the bifurcate sphere  $S_0$, included in 
 the neighborhood $\O$ of hypothesis {\bf SBS}. Choose a smooth 
 future-past directed  null pair $(\Lp, \Lm)$
along $S_0$ (i.e. $\Lp$ is future oriented while $\Lm$ is past oriented), 
\begin{equation}\label{normalization}
\g(\Lm,\Lm)=\g(\Lp,\Lp)=0,\,\,\g(\Lp,T_0)=-1,\,\,\g(\Lp, \Lm)=1.\\
\end{equation}
We extend $\Lp$ (resp. $\Lm$)  along the null geodesic generators of $\HH^+$
 (resp. $\HH^-$) by parallel transport, i.e. $\D_\Lp\Lp=0$ (resp.  $\D_\Lm\Lm=0$). 
 We define the function $\um$ (resp. $\up$) along $\HH^+$ (resp. $\HH^-$) 
 by setting $\um=\up=0$ on the bifurcate sphere $S_0$ and solving $\Lp(\um)=1$ 
 (resp.  $\Lm(\up)=1$).  Let $S_{\um}$ (resp.  $S_{\um}$)  be the level
 surfaces of $\um$ (resp. $\up$)   along $\HH^+$ (resp. $\HH^-$). We define 
 $\Lm$ at every point of  $\HH^+$ (resp. $\Lp$ at every point of $\HH^-$)
 as the unique, past  directed (resp. future directed),  null vector-field orthogonal to the surface $S_{\um}$ (resp. $S_\up$) passing through that point and such that $\g(\Lp, \Lm)=1$.
 We now define the null hypersurface $\HH_{\um}$ to be the congruence
 of   null geodesics  initiating on $S_{\um}\subset\HH^+$ in the direction of
  $\Lm$.
 Similarly we define $\HH_{\up}$ to be the congruence
 of   null geodesics  initiating on $S_{\up}\subset\HH^-$ in the direction
 of  $\Lp$.
 Both congruences are well defined in a sufficiently small neighborhood $\O$  of $S_0$ in $\M$.  The null hypersurfaces $\HH_{\um}$ (resp.  $\HH_{\up}$) are the level sets
 of  a function $\um$  (resp $\up$)  vanishing on $\HH^-$  (resp. $\HH^+$). 
 Moreover we can arrange that both $\um, \up$ are positive in the domain
 of outer communication $\E$. By construction they are both null optical functions,
 i.e.
 \bea
 g^{\mu\nu}\pr_\mu u_+ \pr_\nu u_+= g^{\mu\nu}\pr_\mu u_- \pr_\nu u_-=      0.
 \eea
 We define
 \bea
 \Om&=&g^{\mu\nu}\pr_\mu u_+ \pr_\nu u_-.
 \eea
 In view of our construction we have,
  \bea
  \up|_{\HH^+}=\um|_{\HH^-}=0,\qquad 
  \Om|_{\HH^+\cup\HH^-}=1.
  \eea
  Let
  \bea
  \Lp=g^{\mu\nu}\pr_\mu \up\pr_\nu,\qquad  \Lm=g^{\mu\nu}\pr_\mu \um\pr_\nu.
  \eea
  We have,
  \beaa
  \g(\Lp,\Lp)=\g(\Lm, \Lm)=0,\quad \g(\Lp, \Lm)=\Om.
  \eeaa  
  Define the sets,  
  \beaa
\O_\ep=\{x\in \O:  |\um|< \ep, |\up|< \ep\}.
\eeaa
For   sufficiently small $\ep_0>0$ we have,
\bea
\Om>\frac 1 2\quad \text{in} \,\,\O_{\ep_0},\qquad 
\overline{\O_{\ep_0}}\subset\O.
\eea
 We  also have, for $  \ep\le \ep_0$, $\O_\ep\cap\overline{\E}=\{0\le\um<\ep, 
 0\le  \up<\ep\}  $.   If $\phi$ is a smooth function
 in $\O_\ep$, vanishing on $\HH^+\cap \O_\ep$, one can show
 that  there exists a smooth function $\phi'$ defined
 on $\O_\ep$ such that, 
\begin{equation}\label{boundedgeom2.2}
\phi=u_+\cdot\phi'\text{ on }\mathbf{O}_\ep.
\end{equation}
Similarly, if $\phi $  is a smooth function
  in $\O_\ep$, vanishing on $\HH^-\cap \O_\ep$,   then there exists  another  smooth function $\phi'$  defined on
  $\O_\ep$  such that,
\begin{equation}\label{boundedgeom2.3}
\phi=u_-\cdot\phi'\text{ on }\mathbf{O}_\ep.
\end{equation}

\subsection{Quantitative bounds}\label{coordinates}
Using the hypothesis \eqref{nontang} we may assume that for every $0< \ep<\ep_0$ there is a sufficiently large
constant $\widetilde{A}_\ep$ such that,
\bea
|\g(\T, T_0)|> \widetilde{A}_\ep^{-1},\qquad \forall x\in( \Si_0\cap \E)\setminus \O_{\ep}.\label{nontangq}
\eea
In view of the normalization \eqref{normalization} we may assume (after possibly decreasing the value of $\eps_0$) that, for some constant $A_0$,
\begin{equation}\label{boundedgeom3.1}
u_+/u_-+u_-/u_+\leq A_0\text{ on }\O_{\ep_0}\cap\E\cap\Sigma_0.
\end{equation}

We construct a system of coordinates 
which cover  a neighborhood of the space-like hypersurface $\Si_0$. For any $R\in(0,1]$ let $B_R=\{x\in\mathbb{R}^4:|x|<R\}$ denote the open ball of radius $R$ in $\mathbb{R}^4$. In view of the asymptotic flatness assumption $\mathbf{AF}$, there is a constant $A_0\in[\eps_0^{-1},\infty)$ such that \eqref{boundedgeom3.1} holds and, in addition, for any $x_0\in\Sigma_0\cap \overline{\mathbf{E}}$ there is an open set $B_1(x_0)\subseteq\mathbf{M}$ containing $x_0$ and a smooth coordinate chart $\Phi^{x_0}:B_1\to B_1(x_0)$, $\Phi^{x_0}(0)=x_0$, with the property that
\begin{equation}\label{boundgeom}
\begin{split}
&\sup_{x_0\in\Sigma_0\cap\overline{\mathbf{E}}}\,\sup_{x\in B_1(x_0)}\,\,\sum_{j=0}^6\,\,\sum_{\al_1,\ldots,\al_j,\be,\ga=1}^4\,\,\big( |\partial_{\al_1}\ldots\partial_{\al_j}\g_{\be\ga}(x)|+|\partial_{\al_1}\ldots\partial_{\al_j}\g^{\be\ga}(x)|\big)\leq A_0;\\
&\sup_{x_0\in\Sigma_0\cap\overline{\mathbf{E}}}\,\sup_{x\in B_1(x_0)}\,\,\sum_{j=0}^6\,\,\sum_{\al_1,\ldots,\al_j,\be=1}^4\,\,|\partial_{\al_1}\ldots\partial_{\al_j}\T^{\be}(x)|\leq A_0.
\end{split}
\end{equation}
We may assume that $B_1(x_0)\subseteq \mathbf{O}_{\eps_0}$ if $x_0\in S_0$. We define $\widetilde{\M}$ to be the union of the balls
$B_1(x_0)$ over all points $x_0\in \Si_0\cap \overline{\E}$. We can arrange such that $\widetilde{\M}$ is simply connected.

Since $S_0$ is compact, we may assume (after possibly increasing the value of $A_0$) that
\begin{equation}\label{boundedgeom2.1}
\sup_{x_0\in S_0}\,\sup_{x\in B_1(x_0)}\big[\sum_{j=0}^6\sum_{\al_1,\ldots,\al_j=1}^4|\partial_{\al_1}\ldots\partial_{\al_j}u_{\pm}(x)|+\big(\sum_{\al=1}^4|\partial_{\al}u_{\pm}(x)|\big)^{-1}\big]\leq A_0.
\end{equation}
Finally, we may also assume, in view of \eqref{Main-Cond2},  that
there is a point $x_0\in S_0$ such that,
\begin{equation}\label{boundedgeom4.1}
\Re\big((1-\si(x_0))^{-1}\big)>\frac 12 +A_0^{-1}.
\end{equation}
To summarize, we fixed constants $\eps_0$ and $A_0\geq\eps_0^{-1}$ such that \eqref{boundedgeom3.1}-\eqref{boundedgeom4.1} hold.

\section{Unique continuation and Carleman inequalities}\label{generalCarleman}
\subsection{General considerations}
As explained in  section 1 
our proof of the Main Theorem is based on  a global, unique continuation
strategy applied to  equation \eqref{wave-SS}.  We say
that a linear differential operator $L$, in a domain $\Omega\subset\RRR^d$,   satisfies the unique continuation property with respect  to a smooth, oriented, hypersurface $\Si\subset\Omega$,  if  any smooth solution  of $L\phi=0$ which vanishes on one side of $\Si$  must in fact vanish in a small neighborhood of $\Si$. Such a property depends,
of course, on the interplay between the properties of the operator $L$ and the hypersurface $\Sigma$. A classical result of H\"{o}rmander, see for example Chapter 28 in \cite{Ho}, provides sufficient conditions for a scalar linear equation which guarantee that the unique continuation property holds. In the particular case of
the scalar wave equation, 
$
\square_\g\phi=0,
$ 
and a non-characteristic surface $\Sigma$, defined by the equation $h=0$, $\nabla h\neq 0$, H\"{o}rmander's pseudo-convexity condition takes  the simple form,
\begin{equation}\label{HoCond}
\D^2h(X,X)<0\qquad\text{ if }\qquad \g(X,X)=\g(X,\D h)=0
\end{equation}
at all points on the surface $\Sigma$, where we assume that $\phi$ is known to vanish on the side of $\Sigma$ corresponding to $h<0$. 

In our situation, we plan  to apply the general philosophy of unique continuation to  the   covariant wave equation (see  Theorem \ref{wave}),
\begin{equation}\label{cartoon1}
\square_\g\Ss=\AA\ast\SS+\BB\ast\D\Ss,
\end{equation}
verified by the  Mars-Simon tensor $\Ss$, see Definition \ref{MStensor}. We prove in section \ref{horizon}, using the main technical assumption \eqref{Main-Cond1}, that $\Ss$ vanishes on the horizon $\mathcal{H}^+\cup\mathcal{H}^-$ and we would like to prove, by unique continuation, that $\Ss$ vanishes in the entire domain of outer communication.   In implementing such a strategy one encounters the following
difficulties:
\begin{enumerate}
\item   Equation  \eqref{cartoon1} is tensorial, rather than scalar.
\item The  horizon $\mathcal{H}^+\cup\mathcal{H}^-$ is characteristic and non
smooth in a neighborhood of the bifurcate sphere.
\item Though one can show that an appropriate variant of H\"{o}rmander's pseudo-convexity condition holds true  along  the horizon, in a neighborhood of the bifurcate
sphere,  we have no guarantee that such condition continue to be true  slightly away from
the horizon,  within the ergosphere  region of the stationary space-time  where $\T$
is space-like. 
\end{enumerate}
  Problem  (1) is not very serious; we can effectively reduce   \eqref{cartoon1} to a system  of scalar equations, diagonal with respect to the principal symbol.  Problem
   (2) can be dealt with   by an adaptation of  H\"ormander's
    pseudo-convexity condition.  We note however that such an adaptation is necessary since,
    given our simple vanishing condition of $\Ss$ along the horizon, we cannot 
    directly  apply  H\"ormander's result in   \cite{Ho}.
    Problem (3) is by far the most serious. Indeed, even in the   case
when  $\g$ is  a Kerr metric  \eqref{metric-Kerr},  one can show that there exist 
null geodesics  trapped within the ergosphere region    $m+\sqrt{m^2-a^2}\le r\le  m+\sqrt{m^2-a^2\cos^2\theta}$.  Indeed    
   surfaces  of the form $r\Delta=m(r^2-a^2)^{1/2}$,   which intersect the ergosphere
    for $a$ sufficiently close to $m$,  are known to contain such    null geodesics,
    see \cite{Ch}.  One can  show that the  presence of  trapped  null geodesics invalidates H\"ormander's pseudo-convexity condition. Thus, even in the case of the scalar wave equation  $\square_\g\phi=0$ in  such a Kerr metric, one cannot guarantee,  by a classical unique continuation argument (in the absence of additional conditions)  that  $\phi$ vanishes beyond a small neighborhood of the horizon. 
   
In order to overcome this difficulty  we exploit the geometric nature of our problem and make use of the invariance of $\Ss$ with respect to  $\T$, Thus  the tensor $\Ss$ satisfies, in addition to \eqref{cartoon1}, the identity
\begin{equation}\label{cartoon2}
\mathcal{L}_\T\Ss=0.
\end{equation}
Observe that \eqref{cartoon2} can, in principle, transform  \eqref{cartoon1}
into a much simpler elliptic problem, in  any domain which lies strictly outside the
ergosphere (where $\T$ is strictly  time-like). Unfortunately this   possible strategy  is not available to us since,
as we have remarked above,   we cannot hope  to  extend the vanishing  of $\Ss$,
by a simple  analogue of H\"ormander's pseudo-convexity condition,  
beyond the   first trapped null geodesics. 

Our solution is to extend H\"{o}rmander's classical  pseudo-convexity condition \eqref{HoCond} to one which takes into account both equations
\eqref{cartoon1} and \eqref{cartoon2}.  These considerations lead  to the following
 qualitative, $\T$-conditional, pseudo-convexity condition, 
\begin{equation}\label{HoCond2}
\begin{split}
&\T(h)=0;\\
&\D^2h(X,X)<0\qquad\text{ if }\qquad \g(X,X)=\g(X,\D h)=\g(\T,X)=0.
\end{split}
\end{equation}
 In a first approximation one can  show that this  condition can be verified  in  all Kerr spaces  $a\in[0,m)$, for the simple   function $h=r$ (see \cite{Ion-K1}), where $r$ is one  of the    Boyer--Lindquist coordinates. Thus \eqref{HoCond2}   is a good substitute for  the more general condition \eqref{HoCond}. The fact that the two geometric identities \eqref{cartoon1} and \eqref{cartoon2} cooperate exactly in the right way, via   \eqref{HoCond2},   thus  allowing us  to compensate for  both the failure of 
 condition \eqref{HoCond}   as well as  the failure of the vector field $\T$ to be time-like in the ergoregion, seems to us to be a very remarkable property of the Kerr spaces. In the next subsection we give a quantitative version    of the condition and derive a Carleman estimate of sufficient generality to cover    all 
our needs.

\subsection{A Carleman estimate of sufficient generality}
Unique continuation properties are often proved using Carleman inequalities. In this subsection we prove a sufficiently general Carleman inequality, Proposition \ref{Cargen}, under a quantitative conditional pseudo-convexity assumption. This general  Carleman inequality is used in section \ref{sectionS=01} to show that $\Ss$ vanishes in a small neighborhood of the bifurcate sphere $S_0$ in $\overline{\mathbf{E}}$, and then in section \ref{lastsection} to prove that $\Ss$ vanishes in the entire exterior domain. The two applications are genuinely different, since, in particular, the horizon is a bifurcate surface which is not smooth and the weights needed in this case have to be ``singular'' in an appropriate sense. In order to be able to cover both applications and prove unique continuation in a quantitative sense, which is important especially in section \ref{lastsection}, we work with a more  technical notion of conditional pseudo-convexity than  \eqref{HoCond2}, see Definition \ref{psconvex} below.      

Assume, as in the previous section, that $x_0\in\Sigma_0\cap\overline{\mathbf{E}}$ and $\Phi^{x_0}:B_1\to B_1(x_0)$ is the corresponding coordinate chart. For  simplicity  of notation, let $B_r=B_r(x_0)$. For any smooth function $\phi:B\to\mathbb{C}$, where $B\subseteq B_1$ is an open set, and $j=0,1,\ldots$ let
\begin{equation*}
|D^j\phi(x)|=\sum_{\al_1,\ldots,\al_j=1}^4|\partial_{\al_1}\ldots\partial_{\al_j}\phi(x)|.
\end{equation*}
Assume that $V=V^\alpha\partial_\alpha$ is a vector-field on $B_1$ with the property that
\begin{equation}\label{po1}
\sup_{x\in B_1}\sum_{j=0}^4\sum_{\beta=1}^4|D^jV^\beta|\leq A_0.
\end{equation}
In our applications, $V=0$ or $V=\mathbf{T}$.

\begin{definition}\label{psconvex}
A family of weights $h_\eps:B_{\eps^{10}}\to\mathbb{R}_+$, $\eps\in(0,\eps_1)$, $\eps_1\leq A_0^{-1}$, will be called $V$-conditional pseudo-convex if for any $\eps\in(0,\eps_1)$ \begin{equation}\label{po5}
\begin{split}
h_\eps(x_0)=\eps,\quad\sup_{x\in B_{\eps^{10}}}\sum_{j=1}^4\eps^j|D^jh_\eps(x)|\leq\eps/\eps_1,\quad |V(h_\eps)(x_0)|\leq\eps^{10},
\end{split}
\end{equation}
\begin{equation}\label{po3.2}
\D^\alpha h_\eps(x_0)\D^\be h_\eps(x_0)(\D_\al h_\eps\D_\be h_\eps-\eps\D_\al\D_\be h_\eps)(x_0)\geq\eps_1^2,
\end{equation}
and there is $\mu\in[-\eps_1^{-1},\eps_1^{-1}]$ such that for all vectors $X=X^\alpha\partial_\alpha\in\mathbf{T}_{x_0}(\mathbf{M})$
\begin{equation}\label{po3}
\begin{split}
&\eps_1^2[(X^1)^2+(X^2)^2+(X^3)^2+(X^4)^2]\\
&\leq X^\al X^\be(\mu\g_{\al\be}-\D_\al\D_\be h_\eps)(x_0)+\eps^{-2}(|X^\al V_\al(x_0)|^2+|X^\al\D_\al h_\eps(x_0)|^2).
\end{split}
\end{equation}
A function $e_\eps:B_{\eps^{10}}\to\mathbb{R}$ will be called a negligible perturbation if
\begin{equation}\label{smallweight}
\sup_{x\in B_{\eps^{10}}}|D^je_\eps(x)|\leq\eps^{10}\qquad\text{ for  }j=0,\ldots,4.
\end{equation}
\end{definition}

\begin{remark} One can see that the technical conditions \eqref{po5}, \eqref{po3.2}, and \eqref{po3} are related to the qualitative condition \eqref{HoCond2}, at least when $h_\eps=h+\eps$ for some smooth function $h$. The assumption $|V(h_\eps)(x_0)|\leq\eps^{10}$ is a quantitative version of $V(h)=0$. The assumption \eqref{po3.2} is a quantitative version of the non-characteristic condition $\D^\alpha h\D_\al h\neq 0$. The assumption \eqref{po3} is a quantitative version of the inequality in the second line of \eqref{HoCond2}, in view of the large factor $\eps^{-2}$ on the terms $|X^\al V_\al(x_0)|^2$ and $|X^\al\D_\al h_\eps(x_0)|^2$, and the freedom to choose $\mu$ in a large range. 

It is important that the Carleman estimates we prove are stable under small perturbations of the weight, in order to be able to use them to prove unique continuation. We quantify this stability in \eqref{smallweight}.  
\end{remark}

We observe that if $\{h_\eps\}_{\eps\in(0,\eps_1)}$ is a $V$-conditional pseudo-convex family, and $e_\eps$ is a negligible perturbation for any $\eps\in(0,\eps_1]$, then
\begin{equation*}
h_\eps+e_\eps\in[\eps/2,2\eps]\text{ in }B_{\eps^{10}}.
\end{equation*}
The pseudo-convexity conditions of Definition \ref{psconvex} are probably not as general as possible, but are suitable for our applications both in section \ref{sectionS=01}, with ``singular'' weights $h_\eps$, and section \ref{lastsection}, with ``smooth'' weights $h_\ep$. We  also  note that it is important to our goal to prove a global result (see  section \ref{lastsection}), to be able to track quantitatively the size of the support of the functions for which Carleman estimates can be applied; in our notation, this size depends only on the parameter $\eps_1$ in Definition \ref{psconvex}.

\begin{proposition}\label{Cargen}
Assume $x_0,V$ are as above, $\eps_1\leq A_0^{-1}$, $\{h_\eps\}_{\eps\in(0,\eps_1)}$ is a $V$-conditional pseudo-convex family, and $e_\eps$ is a negligible perturbation for any $\eps\in(0,\eps_1]$. Then there is $\eps\in (0,\eps_1)$ sufficiently small and $\widetilde{C}_\eps$ sufficiently large such that for any $\lambda\geq\widetilde{C}_\eps$ and any $\phi\in C^\infty_0(B_{\eps^{10}})$
\begin{equation}\label{Car1gen}
\lambda \|e^{-\lambda f_\eps}\phi\|_{L^2}+\|e^{-\lambda f_\eps}|D^1\phi|\,\|_{L^2}\leq \widetilde{C}_\eps\lambda^{-1/2}\|e^{-\lambda f_\eps}\,\square_{\g}\phi\|_{L^2}+\eps^{-6}\|e^{-\lambda f_\eps}V(\phi)\|_{L^2},
\end{equation}
where $f_\ep=\ln (h_\eps+e_\eps)$.
\end{proposition}

\begin{proof}[Proof of Proposition \ref{Cargen}] As mentioned earlier, many Carleman estimates such as \eqref{Car1gen} are known,  for the particular case  when $V=0$, in more general settings. The  optimal proof,  see  chapter  28  of  \cite{Ho}, is based on 
the Fefferman--Phong inequality. Here we provide a self-contained, elementary, proof
which, though not optimal,  it is perfectly adequate  to our needs.

We will use the notation $\widetilde{C}$ to denote various constants in $[1,\infty)$ that may depend only on the constant $\eps_1$. We will use the notation $\widetilde{C}_\eps$ to denote various constants in $[1,\infty)$ that may depend only on $\eps$. We emphasize that these constants do not depend on the (very large) parameter $\lambda$ or the function $\phi$ in \eqref{Car1}. The value of $\eps$ will be fixed at the end of the proof and depends only on $\eps_1$. We divide the proof into several steps.  

{\bf Step 1.} Clearly, we may assume that $\phi$ is real-valued. Let $\psi=e^{-\lambda f_\eps}\phi\in C^\infty_0(B_{\eps^{10}})$. In terms of $\psi$, inequality \eqref{Car1gen}  takes the form,
\begin{equation}\label{Car3}
\lambda \|\psi \|_{L^2}+\|e^{-\lambda f_\eps}|D^1(e^{\lambda f_\eps}\psi)|\,\|_{L^2}\leq \widetilde{C}_\eps\lambda^{-1/2}\|e^{-\lambda f_\eps}\,\square_{\g}(e^{\lambda f_\eps}\psi)\|_{L^2}+\eps^{-6}\|e^{-\lambda f_\eps}V(e^{\lambda f_\eps}\psi)\|_{L^2}.
\end{equation}
We reduce the proof of \eqref{Car3}  by a sequence of steps. We claim first that for \eqref{Car3} to hold true,  it suffices to prove that there exist $\eps\ll 1$ and $\widetilde{C}_\eps\gg 1$ such that
\begin{equation}\label{Car4}
\lambda \|\psi \|_{L^2}+\|\,|D^1\psi|\,\|_{L^2}\leq \widetilde{C}_\eps\lambda^{-1/2}\|e^{-\lambda f_\eps}\,\square_{\g}(e^{\lambda f_\eps}\psi)\|_{L^2}+8\eps^{-4}\|V(\psi)\|_{L^2},
\end{equation}
for any $\lambda\geq\widetilde{C}_\eps$ and any $\psi \in C^\infty_0(B_{\eps^{10}})$. Indeed,  using \eqref{po5} and \eqref{smallweight} (thus $|V(h_\eps+e_\eps)(x)|\leq \widetilde{C}\eps^8$ for $x\in B_{\eps^{10}}$), the observation $h_\eps+e_\eps\in[\eps/2,2\eps]$ in $B_{\eps^{10}}$, and the definition $f_\eps=\ln (h_\eps+e_\eps)$, we have
\begin{equation*}
\begin{split}
&e^{-\lambda f_\eps}|D^1(e^{\lambda f_\eps}\psi)|\le |D^1\psi|+\widetilde{C}\eps^{-1}\lambda |\psi|;\\
&|e^{-\lambda f_\eps}V(e^{\lambda f_\eps}\psi)-V(\psi)|\leq\widetilde{C}\eps^7\lambda|\psi|.
\end{split}
\end{equation*}
Thus,  assuming  \eqref{Car4},  we  deduce,
\begin{equation*}
\begin{split}
&\lambda \|\psi \|_{L^2}+\|e^{-\lambda f_\eps}|\,D^1(e^{\lambda f_\eps}\psi)|\,\|_{L^2}\leq \lambda \|\psi \|_{L^2}+\||D^1\psi|\,\|_{L^2}+\widetilde{C}\eps^{-1}\la\|\psi\|_{L^2}\\
&\leq(1+\widetilde{C}\eps^{-1})  (\widetilde{C}_\eps\lambda^{-1/2}\|e^{-\lambda f_\eps}\,\square_{\g}(e^{\lambda f_\eps}\psi)\|_{L^2}+8\eps^{-4}\|V(\psi)\|_{L^2})\\
&\leq (1+\widetilde{C}\eps^{-1})  [\widetilde{C}_\eps\lambda^{-1/2}\|e^{-\lambda f_\eps}\,\square_{\g}(e^{\lambda f_\eps}\psi)\|_{L^2}+8\eps^{-4}\|e^{-\lambda f_\eps}V(e^{\lambda f_\eps}\psi)\|_{L^2}+8\widetilde{C}\eps^3\lambda\|\psi\|_{L^2}]\\
&\leq  \widetilde{C}_\eps\lambda^{-1/2}\|e^{-\lambda f_\eps}\,\square_{\g}(e^{\lambda f_\eps}\psi)\|_{L^2}+\widetilde{C}\eps^{-5}\|e^{-\lambda f_\eps}V(e^{\lambda f_\eps}\psi)\|_{L^2}+\widetilde{C}\eps^2\lambda\|\psi\|_{L^2},
\end{split}
\end{equation*}
and the inequality \eqref{Car3} follows for $\eps\ll\widetilde{C}^{-1}$.

{\bf Step 2.} We write 
\begin{equation}\label{Car5}
\begin{split}
e^{-\lambda f_\eps}\,\square_{\g}(e^{\lambda f_\eps}\psi)&=\square_{\g}\psi+2\lambda\D^\al(f_\eps)\D_\al\psi+\lambda^2\D_\al(f_\eps)\D^\al(f_\eps)\cdot \psi+\lambda \square_{\g}(f_\eps)\cdot \psi,\\
&=L_\ep \psi+\lambda \square_{\g}(f_\eps)\cdot \psi,
\end{split}
\end{equation}
with $L_\ep := \square_{\g}+2\lambda\D^\al(f_\eps)\D_\al+\lambda^2\D_\al(f_\eps)\D^\al(f_\eps)$,
and show that  \eqref{Car4} follows from,
\begin{equation}\label{Car6}
\begin{split}
\lambda \|\psi \|_{L^2}+\|\,|D^1\psi|\,\|_{L^2} & \leq \widetilde{C}_\eps\lambda^{-1/2}\|L_\ep\psi\|_{L^2}+4\eps^{-4}\|V(\psi)\|_{L^2}
\end{split}
\end{equation}
for any $\lambda\geq\widetilde{C}_\eps$ and any $\psi \in C^\infty_0(B_{\eps^{10}})$. Indeed,
 \beaa
  \|e^{-\lambda f_\eps}\,\square_{\g}(e^{\lambda f_\eps}\psi)\|_{L^2}\ge \|L_\ep\psi\|_{L^2}-\la \|\square_{\g}(f_\eps)\psi\|_{L^2}
 \eeaa
Observe that, according to \eqref{po5},  we have
 $|\square_{\g}(f_\eps)|\leq\widetilde{C}_\eps$ on $B_{\eps^{10}}$.
 Thus, if \eqref{Car6} holds,
 \begin{equation*}
 \begin{split}
 \lambda \|\psi \|_{L^2}+\|\,|D^1\psi|\,\|_{L^2}&\leq\widetilde{C}_\eps\lambda^{-1/2}\big(    \|e^{-\lambda f_\eps}\,\square_{\g}(e^{\lambda f_\eps}\psi)\|_{L^2}+
 \la \|\square_{\g}(f_\eps)\psi\|_{L^2}\big)+4\eps^{-4}\|V(\psi)\|_{L^2}\\
 &\leq\widetilde{C}_\eps\lambda^{-1/2}   \|e^{-\lambda f_\eps}\,\square_{\g}(e^{\lambda f_\eps}\psi)\|_{L^2}+ \widetilde{C}_\eps^2\lambda^{1/2} \|\psi\|_{L^2}+4\eps^{-4}\|V(\psi)\|_{L^2}
 \end{split}
 \end{equation*}
 or,
 \beaa
 ( \lambda   - \widetilde{C}_\eps^2\lambda^{1/2})  \|\psi \|_{L^2}+\|\,|D^1\psi|\,\|_{L^2}\leq \widetilde{C}_\eps\lambda^{-1/2}    \|e^{-\lambda f_\eps}\,\square_{\g}(e^{\lambda f_\eps}\psi)\|_{L^2}+4\eps^{-4}\|V(\psi)\|_{L^2}
 \eeaa
 from which we easily derive \eqref{Car4},    by  redefining the constant  $ \widetilde{C}_\eps$ and taking $\la$ sufficiently large relative to  $ \widetilde{C}_\eps$.

{\bf Step 3.} We  write $L_\ep$  in the form, 
 \bea
 L_\ep &=&\square_\g+2\lambda W  +\la^2 G\nn\\
W&=&\D^\al(f_\eps)\D_\al,\qquad G=\D_\al(f_\eps)\D^\al(f_\eps). \label{Car8}
 \eea
We observe that   inequality \eqref{Car6}  follows  as a consequence of  the following statement: there exist  $\eps\ll 1$, $\mu_1\in[-\eps^{-3/2},\eps^{-3/2}]$, and $\widetilde{C}_\eps\gg1$ such that
\begin{equation}\label{Car7}
\begin{split}
2\lambda&\eps^{-8}\|V(\psi)\|_{L^2}^2+\int_{B_{\eps^{10}}}L_\ep \psi\cdot(2\lambda W(\psi)-2\lambda w\psi)\,d\mu\\
&\geq\widetilde{C}_\eps^{-1}\|\lambda W(\psi)-\lambda w\psi\|^2_{L^2}+\lambda^3\|\psi\|_{L^2}^2+\lambda \|\,|D^1\psi|\,\|^2_{L^2},
\end{split}
\end{equation}
for any $\lambda\geq\widetilde{C}_\eps$ and any $\psi \in C^\infty_0(B_{\eps^{10}})$, where
\begin{equation}\label{wdefin}
w=\mu_1-(1/2)\square_\g f_\eps.
\end{equation}
The reason for choosing  $w$ of this form will become clear in Step 6. Assuming that \eqref{Car7} holds true and  denoting by $\mbox{RHS}$ the right-hand side of  that inequality, we have
\begin{equation*}
\begin{split}
 \mbox{RHS}&\leq\int_{B_{\eps^{10}}}\widetilde{C}_\eps^{1/2} L_\ep \psi  \cdot \widetilde{C}_\eps^{-1/2}(2\lambda W(\psi)-2\lambda w\psi)\,d\mu+2\lambda\eps^{-8}\|V(\psi)\|_{L^2}^2\\
 &\le\widetilde{C}_\eps^{-1}\|\lambda W(\psi)-\lambda w\psi\|^2_{L^2}+\widetilde{C}_\eps \|L_\eps\psi\|^2_{L^2}+2\lambda\eps^{-8}\|V(\psi)\|_{L^2}^2.
\end{split}
\end{equation*}
Hence
\begin{equation*}
\lambda^3\|\psi\|_{L^2}^2+\lambda \|\,|D^1\psi|\,\|^2_{L^2}\leq\widetilde{C}_\eps\|L_\eps\psi\|^2_{L^2}+2\lambda\eps^{-8}\|V(\psi)\|_{L^2}^2
\end{equation*}
from which \eqref{Car6} follows easily.

{\bf Step 4.} We claim now that inequality  \eqref{Car7} is a consequence of the inequality
\begin{equation}\label{Car9}
\begin{split}
2\lambda&\eps^{-8}\|V(\psi)\|_{L^2}^2+\int_{B_{\eps^{10}}}\big(\square_{\g}\psi+\lambda^2G\psi\big)\cdot (2\lambda W(\psi)-2\lambda w\psi)\,d\mu+2\lambda^2\|W(\psi)\|_{L^2}^2\\
&\geq 2\lambda^3\|\psi\|_{L^2}^2+2\lambda \|\,|D^1\psi|\,\|^2_{L^2}.
\end{split}
\end{equation}
To prove that \eqref{Car9} implies \eqref{Car7} we write
\beaa
L_\ep \psi=\square_\g\psi+\lambda^2G\cdot \psi +(\la W(\psi)- \la w\psi)+(\lambda W(\psi)+\lambda w\psi),
\eeaa
Thus, assuming \eqref{Car9},
\begin{equation*}
\begin{split}
&2\lambda\eps^{-8}\|V(\psi)\|_{L^2}^2+\int_{B_{\eps^{10}}}L_\ep \psi\cdot(2\lambda W(\psi)-2\lambda w\psi)\,d\mu\\
&=2\lambda\eps^{-8}\|V(\psi)\|_{L^2}^2+\int_{B_{\eps^{10}}}\big(\square_{\g}\psi+\lambda^2G\psi\big)\cdot (2\lambda W(\psi)-2\lambda w\psi)\,d\mu\\
&+2\|\lambda W(\psi)-\lambda w\psi\|^2_{L^2}+2\lambda^2(\|W(\psi)\|_{L^2}^2-\|w\psi\|_{L^2}^2)\\
&\geq 2\lambda^3\|\psi\|_{L^2}^2+2\lambda \|\,|D^1\psi|\,\|^2_{L^2}+2\|\lambda W(\psi)-\lambda w\psi\|^2_{L^2}-2\lambda^2\|w\psi\|_{L^2}^2\\
&\geq 2\|\lambda W(\psi)-\lambda w\psi\|^2_{L^2}+\lambda^3\|\psi\|_{L^2}^2+2\lambda \|\,|D^1\psi|\,\|^2_{L^2},
\end{split}
\end{equation*}
if $\widetilde{C}_\eps$ is sufficiently large and $\lambda\geq \widetilde{C}_\eps$, which gives \eqref{Car7}. In the last inequality we use the bound $|w|\leq\widetilde{C}\eps^{-2}$ (see \eqref{wdefin}) thus $2\lambda^3\|\psi\|_{L^2}^2-2\lambda^2\|w\psi\|_{L^2}^2\geq \lambda^3\|\psi\|_{L^2}^2$ for $\lambda$ sufficiently large.

{\bf Step 5.} Let $Q_{\al\be}$ denote the enery-momentum tensor of $\square_{\g}$, i.e.
\begin{equation*}
Q_{\al\be}=\D_\al\psi\D_\be\psi-\frac{1}{2}g_{\al\be}(\D^\mu\psi\D_\mu \psi).
\end{equation*}
Direct computations show that
\begin{equation}\label{iden1}
\begin{split}
\square_{\g}\psi\cdot(2W(\psi)-2 w\psi)&=\D^\al(2W^\be Q_{\al\be}-2w\psi\cdot\D_\al\psi+\D_\al w\cdot \psi^2)\\
&-2\D^\al W^\be\cdot Q_{\al\be}+2w\D^\al\psi\cdot \D_\al\psi-\square_{\g}w\cdot \psi^2,
\end{split}
\end{equation}
and
\begin{equation}\label{iden2}
G\psi\cdot(2W(\psi)-2 w\psi)=\D^\al(\psi^2G\cdot W_\al)-\psi^2(2wG+W(G)+G\cdot \D^\al
W_\al).
\end{equation}
Since $\psi\in C^\infty_0(B_{\eps^{10}})$ we integrate by parts to conclude that\begin{equation*}
\begin{split}
&\int_{B_{\eps^{10}}}\big(\square_{\g}\psi+\lambda^2G\cdot \psi\big)\cdot(2W(\psi)-2 w\psi)\,d\mu=\int_{B_{\eps^{10}}}2w\D^\al\psi\cdot \D_\al\psi-2\D^\al W^\be\cdot Q_{\al\be}\,d\mu\\
&+\lambda ^2\int_{B_{\eps^{10}}}\psi^2(-2wG-W(G)-G\cdot \D^\al W_\al-\lambda^{-2}\square_\g w)d\mu.
\end{split}
\end{equation*}
Thus, after dividing by $\lambda$, for \eqref{Car9} it suffices to prove that the pointwise bounds
\begin{equation}\label{Car20}
|D^1\psi|^2\leq \eps^{-8}|V(\psi)|^2+\lambda|W(\psi)|^2+(w\D^\al\psi\cdot \D_\al\psi-\D^\al W^\be\cdot Q_{\al\be}),
\end{equation}
and
\begin{equation}\label{Car21}
2\leq-2wG-W(G)-G\cdot \D^\al W_\al-\lambda^{-2}\square_\g w,
\end{equation}
hold on $B_{\eps^{10}}$. 

{\bf Step 6.} Recall that $w=\mu_1-(1/2)\square_\g f_\eps$, $W^\al=\D^\al(f_\eps)$ and $G=\D_\al(f_\eps)\D^\al(f_\eps)$. Observe that
\begin{equation*}
w\D^\al\psi\cdot \D_\al\psi-\D^\al W^\be\cdot Q_{\al\be}=(\D^\al\psi\cdot \D^\be\psi)[(w+(1/2)\square_\g f_\eps)\g_{\al\be}-\D_\al\D_\be f_\eps]
\end{equation*}
and
\begin{equation*}
-2wG-W(G)-G\cdot \D^\al W_\al=-G(2w+\square_\g f_\eps)-2\D^\al f_\eps\D^\be f_\eps \cdot \D_\al \D_\be f_\eps.
\end{equation*}
Thus \eqref{Car20} and \eqref{Car21} are equivalent to the pointwise inequalities
\begin{equation}\label{Car22}
|D^1\psi|^2\leq\eps^{-8}|V(\psi)|^2+\lambda|\D_\alpha f_\eps\cdot\D^\alpha\psi|^2+ (\D^\al\psi\cdot \D^\be\psi)(\mu_1\g_{\al\be}-\D_\al\D_\be f_\eps),
\end{equation}
and
\begin{equation}\label{Car23}
1\leq-\mu_1 G-\D^\al f_\eps\D^\be f_\eps \cdot \D_\al \D_\be f_\eps+(1/4)\lambda^{-2}\square_\g^2(f_\eps)
\end{equation}
on $B_{\eps^{10}}$, for some $\eps\ll 1$ and $\lambda$ sufficiently large.

Let $\widetilde{h}_\eps=h_\eps+e_\eps$ and $\widetilde{H}_\eps=\D^\alpha \widetilde{h}_\eps\D_\alpha \widetilde{h}_\eps$. We use now the definition $f_\eps=\ln \widetilde{h}_\eps$. Since $\widetilde{h}_\eps\in[\eps/2,2\eps]$, for \eqref{Car22} and \eqref{Car23} it suffices to prove that there are constants $\eps\ll 1$ and $\mu_1\in[-\eps^{-3/2},\eps^{-3/2}]$ such that the pointwise bounds
\begin{equation}\label{Car24}
|D^1\psi|^2\leq\eps^{-8}|V(\psi)|^2+\eps^{-8}|\D_\alpha \widetilde{h}_\eps\cdot\D^\alpha\psi|^2+ (\D^\al\psi\cdot \D^\be\psi)(\mu_1\g_{\al\be}-\widetilde{h}_\eps^{-1}\D_\al\D_\be \widetilde{h}_\eps),
\end{equation}
and
\begin{equation}\label{Car25}
2\leq \widetilde{h}_\eps^{-4}\widetilde{H}_\eps^2-\widetilde{h}_\eps^{-3}\D^\al \widetilde{h}_\eps\D^\be \widetilde{h}_\eps\D_\al\D_\be \widetilde{h}_\eps-\widetilde{h}_\eps^{-2}\mu_1 \widetilde{H}_\eps
\end{equation} 
hold on $B_{\eps^{10}}$ for any $\psi\in C^\infty_0(B_{\eps^{10}})$. Indeed, the bound \eqref{Car22} follows from \eqref{Car24} if $\lambda\geq 2\eps^{-7}$. The bound \eqref{Car23} follows from \eqref{Car25} if $|\lambda^{-2}\square_\g^2(f_\eps)|\leq 1$, which holds true if $\lambda\geq \widetilde{C}\eps^{-2}$.
\medskip

{\bf Step  7.} We prove now that the bound \eqref{Car25} holds for any $\mu_1\in[-\eps^{-3/2},\eps^{-3/2}]$. We start from the assumption \eqref{po3.2}
\begin{equation*}
\D^\alpha h_\eps(x_0)\D^\be h_\eps(x_0)(\D_\al h_\eps\D_\be h_\eps-\eps\D_\al\D_\be h_\eps)(x_0)\geq\eps_1^2.
\end{equation*}
For $x\in B_{\eps^{10}}$ let $$K(x)=\D^\alpha h_\eps(x)\D^\be h_\eps(x)(\D_\al h_\eps\D_\be h_\eps-h_\eps\cdot\D_\al\D_\be h_\eps)(x).$$ It follows from the second  bound in \eqref{po5} that $|D^1K(x)|\leq\widetilde{C}\eps^{-1}$, thus, since $\eps=h_\eps(x_0)$, $K(x)\geq \eps_1^2/2$ for any $x\in B_{\eps^{10}}$ if $\eps$ is sufficiently small. 

Let $$\widetilde{K}(x)=\D^\alpha \widetilde{h}_\eps(x)\D^\be \widetilde{h}_\eps(x)(\D_\al \widetilde{h}_\eps\D_\be \widetilde{h}_\eps-\widetilde{h}_\eps\cdot\D_\al\D_\be \widetilde{h}_\eps)(x).$$ It follows from the assumption \eqref{smallweight} on $e_\eps$ and the assumption \eqref{po5} that $|\widetilde{K}(x)-K(x)|\leq \widetilde{C}\eps$, thus $\widetilde{K}(x)\geq \eps_1^2/4$ on $B_{\eps^{10}}$, provided that $\eps$ is sufficiently small.  By multiplying with $\widetilde{h}_\eps^{-4}$ we have
\begin{equation*}
\widetilde{h}_\eps^{-4}\eps_1^2/4\leq \widetilde{h}_\eps^{-4}\widetilde{K}(x)= \widetilde{h}_\eps^{-4}\widetilde{H}_\eps^2-\widetilde{h}_\eps^{-3}\D^\alpha \widetilde{h}_\eps\D^\be \widetilde{h}_\eps\cdot\D_\al\D_\be \widetilde{h}_\eps
\end{equation*}
on $B_{\eps^{10}}$.
The bound  \eqref{Car25} follows for $\eps$ small enough since $\widetilde{h}_\eps(x)\in[\eps/2,2\eps]$ on  $B_{\eps^{10}}$ and $|\widetilde{h}_\eps^{-2}\mu_1 \widetilde{H}_\eps|\leq\widetilde{C}|\mu_1|\eps^{-2}\leq \widetilde{C}\eps^{-7/2}$.

{\bf{Step 8.}} We prove now the bound \eqref{Car24}. We start from the assumption \eqref{po3}
\begin{equation}\label{po49}
\begin{split}
&\eps_1^2[(X^1)^2+(X^2)^2+(X^3)^2+(X^4)^2]\\
&\leq X^\al X^\be(\mu\g_{\al\be}-\D_\al\D_\be h_\eps)(x_0)+\eps^{-2}(|X^\al V_\al(x_0)|^2+|X^\al\D_\al h_\eps(x_0)|^2),
\end{split}
\end{equation}
for some $\mu\in[-\eps_1^{-1},\eps_1^{-1}]$ and all vectors $X=X^\alpha\partial_\alpha\in\mathbf{T}_{x_0}(\mathbf{M})$. Let
\begin{equation*}
K_{\al\be}=\mu\eps^{-1}h_\eps\g_{\al\be}-\D_\al\D_\be h_\eps+\eps^{-2}V_\al V_\be+\eps^{-2}\D_\al h_\eps \D_\be h_\eps.
\end{equation*}
We work in the local frame $\partial_1,\partial_2,\partial_3,\partial_4$. In view of \eqref{po5},
\begin{equation*}
|D^1K_{\al\be}(x)|\leq \widetilde{C}\eps^{-3}
\end{equation*} 
for any $\al,\be=1,2,3,4$ and $x\in B_{\eps^{10}}$. It follows from \eqref{po49} and $\eps^{-1}h_\eps(x_0)=1$ that
\begin{equation}\label{po50}
\sum_{\al,\be=1}^4X^\al X^\be K_{\al\be}(x)\geq (\eps_1^2/2)[(X^1)^2+(X^2)^2+(X^3)^2+(X^4)^2]
\end{equation}
for any $x\in B_{\eps^{10}}$ and $(X^1,X^2,X^3,X^4)\in\mathbb{R}^4$, provided that $\eps$ is sufficiently small. Let
\begin{equation*}
\widetilde{K}_{\al\be}=\mu\eps^{-1}\widetilde{h}_\eps\g_{\al\be}-\D_\al\D_\be \widetilde{h}_\eps+\eps^{-2}V_\al V_\be+\eps^{-2}\D_\al \widetilde{h}_\eps \D_\be \widetilde{h}_\eps,
\end{equation*}
and observe that, in view of \eqref{smallweight} and \eqref{po5}, $|\widetilde{K}_{\al\be}(x)-K_{\al\be}(x)|\leq \widetilde{C}\eps^5$ for any $\al,\be=1,2,3,4$ and $x\in B_{\eps^{10}}$. Thus, using \eqref{po50}, if $\eps$ is sufficiently small then
\begin{equation*}
\sum_{\al,\be=1}^4X^\al X^\be \widetilde{K}_{\al\be}(x)\geq (\eps_1^2/4)[(X^1)^2+(X^2)^2+(X^3)^2+(X^4)^2]
\end{equation*}
for any $x\in B_{\eps^{10}}$ and $(X^1,X^2,X^3,X^4)\in\mathbb{R}^4$. We multiply this by $\widetilde{h}_\eps^{-1}\in[\eps^{-1}/2,2\eps^{-1}]$ and use the definition of $\widetilde{K}_{\al\be}$ to conclude that
\begin{equation*}
\begin{split}
\sum_{\al,\be=1}^4X^\al X^\be(\mu\eps^{-1}\g_{\al\be}-\widetilde{h}_\eps^{-1}\D_\al\D_\be \widetilde{h}_\eps)+2\eps^{-3}\big|\sum_{\al=1}^4X^\alpha V_\al\big|^2+2\eps^{-3}\big|\sum_{\al=1}^4X^\alpha \D_\al h_\eps\big|^2\\
\geq \widetilde{h}_\eps^{-1}(\eps_1^2/4)[(X^1)^2+(X^2)^2+(X^3)^2+(X^4)^2].
\end{split}
\end{equation*}
The bound \eqref{Car24} follows for $\eps$ sufficiently small, with $\mu_1=\mu\eps^{-1}\in[-(\eps\eps_1)^{-1},(\eps\eps_1)^{-1}]$. This completes the proof of the proposition.
\end{proof}

\section{The Mars-Simon tensor $\Ss$}\label{MarsSimon}

\subsection{Preliminaries}
Assume $(\mathbf{N},\mathbf{g})$ is a smooth vacuum Einstein spacetime of dimension $4$. Given an antisymmetric  2-form, real or complex valued,   $G_{\a\b}=-G_{\b\a}$ we define its Hodge dual,
\beaa
\dual G_{\a\b}=\frac 12 \in_{\a\b}^{\quad \mu\nu} G_{\mu\nu}.
\eeaa
Observe that $ \dual(\dual G)=-G$. This follows easily from the identity,
\begin{equation*}
\in_{\al\be\rh\si}\in^{\mu\nu\rh\si}=-2\delta_\al^\mu\,\wedge\de_\be^\nu=-2 (\delta_\al^\mu\,\de_\be^\nu-\delta_\al^\nu\,\delta_\be^\mu),
\end{equation*}
Given $2$ such forms $F,G$ we have the identity
\bea
F_{\mu\si}{G_{\nu}}^{\sigma}-(\dual F)_{\nu\si}{(\dual G)_{\mu}}^{\sigma}&=&\frac{1}{2}\g_{\mu\nu}F_{\al\be}G^{\al\be}
\label{b8.2}
\eea
which follows easily from the identity
\begin{equation*}
\begin{split}
\in^{\al_2\al_3\al_4\al_1}&\in_{\be_2\be_3\be_4\al_1}=-
\delta_{\be_2}^{\al_2}\wedge\delta_{\be_3}^{\al_3}\wedge\delta_{\beta_4}^{\al_4}\\
&=\delta_{\be_4}^{\al_2}\delta_{\be_3}^{\al_3}\delta_{\beta_2}^{\al_4}+\delta_{\be_2}^{\al_2}\delta_{\be_4}^{\al_3}\delta_{\beta_3}^{\al_4}+\delta_{\be_3}^{\al_2}\delta_{\be_2}^{\al_3}\delta_{\beta_4}^{\al_4}-\delta_{\be_2}^{\al_2}\delta_{\be_3}^{\al_3}\delta_{\beta_4}^{\al_4}-\delta_{\be_3}^{\al_2}\delta_{\be_4}^{\al_3}\delta_{\beta_2}^{\al_4}-\delta_{\be_4}^{\al_2}\delta_{\be_2}^{\al_3}\delta_{\beta_3}^{\al_4}.
\end{split}
\end{equation*}
An antisymmetric 2-form $\FF$ is called self-dual
if,
\beaa
\dual \FF=-i\FF.
\eeaa
It follows easily form \eqref{b8.2} that if $\FF, \GG$ are two self-dual $2$-forms then
\bea
\FF_{\mu\si}{\GG_{\nu}}^{\sigma} +\FF_{\nu\si}{\GG_{\mu}}^{\sigma}=\frac{1}{2}\g_{\mu\nu}\FF_{\al\be}\GG^{\al\be}.
\label{b8.2'} 
\eea
We also have, for any self-dual $\FF$,
\bea
\FF_{\mu\si}( \Re \FF)_{\nu}^{\,\,\,\,\si}=\FF_{\nu\si}( \Re \FF)_{\mu}^{\,\,\,\,\si}\label{ReF-FF}
\eea
where $\Re\FF$ denotes the real part of $\FF$.

A tensor $W\in\mathbf{T}_4^0(\mathbf{N})$ will be called partially antisymmetric if
\begin{equation}
W_{\al\be\mu\nu}=-W_{\be\al\mu\nu}=-W_{\al\be\nu\mu}. \label{eq:Weyl1}
\end{equation}
Given such a tensor-field we define its Hodge dual
\beaa
{\dual W}_{\al\be\ga\de}&=&\frac{1}{2}{\in_{\ga\de}}^{\rh\si}W_{\al\be\rh\si}.
\eeaa  
As before, ${}^\ast(\dual W)=-W$ for any partially antisymmetric tensor $W$. A complex partially antisymmetric tensor $\UU$ of rank 4 is called self-dual if $\dual\UU=(-i)\UU$. The following
extension of 
identity \eqref{b8.2'}  holds for such tensors, 
 \begin{equation} \label{le:useful}
 \FF_{\mu}^{\,\,\,\si}\UU_{\a\b\nu\sigma} +\FF_{\nu}^{\,\,\,\si} \UU_{\a\b\mu\sigma}=
 \frac{1}{2}\g_{\mu\nu}\FF^{\ga\de}\UU_{\a\b\ga\de}.
 \end{equation}

A partially antisymmetric tensor of rank 4 is called a Weyl field if
\begin{equation}\label{Weyl12}
\begin{cases}
&W_{\al\be\mu\nu}=-W_{\be\al\mu\nu}=-W_{\al\be\nu\mu}=W_{\mu\nu\al\be};\\
&W_{\a[\b\mu\nu]}=W_{\a\b\mu\nu}+W_{\a\mu\nu\be}+W_{\a\nu\be\mu}=0;\\
&\g^{\be\nu}W_{\al\be\mu\nu}=0.
\end{cases}
\end{equation}
It is well-known that if $W$ is a Weyl field then ${\dual W}$ is also a Weyl field. In particular
\begin{equation}\label{Weylsym}
{\dual W}_{\al\be\mu\nu}={\dual W}_{\mu\nu\al\be}=\frac{1}{2}\in_{\al\be}^{\quad\rh\si}W_{\mu\nu\rh\si}.
\end{equation}
The Riemann curvature tensor $R$ of an Einstein vacuum spacetime provides
an example of a Weyl field. Moreover  $R$ verifies the
Bianchi identities,
\beaa
\D_{[\si} R_{\ga\de]\a\b}=0
\eeaa
In this paper  we will have to consider Weyl fields $W$ which verify equations of the form
\bea
\D^{\a}\, W_{\a\b\ga\de}=J_{\b\ga\de}\label{eq:Bianchi1}
\eea
for some Weyl current $J\in \mathbf{T}_3^0(\mathbf{N})$. It follows from \eqref{eq:Bianchi1} that
\begin{equation}\label{eq:Bianchi2}
\D^{\a}\dual W_{\a\b\ga\de}=\dual J_{\b\ga\de}=\frac{1}{2}{\in_{\ga\de}}^{\rho\si}J_{\be\rho\si}.
\end{equation}
The following proposition follows immediately from definitions and \eqref{Weylsym}. 

\begin{proposition}\label{prop:Bianchi-dual}
If $W$ is a Weyl field and \eqref{eq:Bianchi1} is satisfied then
\bea
\D_{[\si}W_{\ga\de]\a\b}=\in_{\mu\si\ga\de}{{\dual J}^\mu}_{\a\b}.
\eea
\end{proposition}

\subsection{Killing vector-fields and the Ernst potential}

We assume now that $\mathbf{T}$ is a Killing vector-field  on $\mathbf{N}$, i.e.
\bea
\D_\al\T_\be+\D_\be\T_\al=0\label{xi-killing}
\eea
We define the 2-form,
\beaa
F_{\al\be}=\D_\al\T_\be
\eeaa
and recall that $F$ verifies the Ricci identity
\bea
\D_\mu F_{\al\be}=\T^\nu R_{\nu\mu\al\be},\label{eq:Ricci}
\eea
with $R$ the curvature tensor of the spacetime.
In view of the  first Bianchi identity for $R$ we infer that,
\bea
\D_{[\mu}F_{\a\b]}=\D_\mu F_{\al\be}+\D_\a F_{\be\mu}+\D_\b F_{\mu\a}=0.\label{eq:Max1}
\eea
Also, since  we are in an Einstein vacuum spacetime,
\bea
\D^\b F_{\a\b}=0.\label{eq:Max2}
\eea
 
We now  define the complex valued 2-form,
\begin{equation}\label{s1}
\FF_{\al\be}=F_{\al\be}+i {\dual F}_{\al\be}.
\end{equation}
Clearly, $\FF$ is self-dual solution of the Maxwell equations,  i.e. $ \FF\dual=(-i) \FF$ 
and
\bea
\D_{[\mu} \FF_{\al\be]}=0\label{eq:Max1'},\quad \D^\b \FF_{\a\b}=0.\label{eq:Max2'}
\eea
We define also the Ernst $1$-form associated to the Killing vector-field  $\T$,
\bea
\si_\mu&=&2\T^\a\FF_{\a\mu}=\D_\mu(-\T^\al \T_\al)- i\in_{\mu\b\ga\de}\T^\b \D^\ga\T^\de.
\eea
It is easy to check (see, for example, \cite[section 3]{Ma2}) that 
\begin{equation}\label{Ernst1}
\begin{cases}
&\D_\mu\si_\nu-\D_\nu\si_\mu=0;\\
&\D^\mu\si_\mu=-\FF^2;\\
&\si_\mu\si^\mu=g(\T,\T) \FF^2.
\end{cases}
\end{equation}
Since $d(\si_\mu dx^\mu)=0$ and the set $\widetilde{\mathbf{M}}$ is simply connected we infer that there exists a function $\si:\widetilde{\mathbf{M}}\to\mathbb{C}$, called the Ernst potential, such that $ \si_\mu =\D_\mu\si$, $\si\to 1$ at infinity along $\Sigma_0$, and $\Re\si=-\T^\al\T_\al$.

\subsection{The Mars-Simon tensor}
In the rest of this section we assume that $\mathbf{N}\subseteq\widetilde{\mathbf{M}}$ is an open set with the property that
\begin{equation}\label{s0}
1-\sigma\neq 0\text{ in }\mathbf{N}.
\end{equation}
We define 
 the complex-valued self-dual Weyl tensor
\begin{equation}\label{s2}
\RR_{\al\be\mu\nu}=R_{\al\be\mu\nu}+\frac{i}{2}{\in_{\mu\nu}}^{\rh\si}R_{\al\be\rh\si}=R_{\al\be\mu\nu}+i{\dual R}_{\al\be\mu\nu}.
\end{equation}
We define the tensor $\II\in\mathbf{T}_4^0(\mathbf{N})$,
\begin{equation}\label{s4}
\II_{\al\be\mu\nu}=(\g_{\al\mu}\g_{\be\nu}-\g_{\al\nu}\g_{\be\mu}+i\in_{\al\be\mu\nu})/4.
\end{equation}
Clearly,
\bea
\II_{\al\be\mu\nu}=-\II_{\b\a\mu\nu}=-\II_{\al\be\nu\mu}=\II_{\mu\nu\al\be}.\label{eq:dualIII}
\eea
On the other hand,
\begin{equation}\label{s55}
\II_{\a[\b\ga\de]}=\II_{\a\b\ga\de}+\II_{\a\ga\de\b}+\II_{\a\de\b\ga}=\frac{3i}{4}\in_{\a\b\ga\de}.
\end{equation}
Using the definition \eqref{s4} we derive
\begin{equation}
\dual\II_{\a\b\mu\nu}=\frac{1}{2}{\in_{\mu\nu}}^{\rh\si}\II_{\al\be\rh\si}=(-i)\II_{\al\be \mu\nu}.\label{eq:dualI}
\end{equation}
Thus $\II$ is a self-dual partially antisymmetric tensor. We can therefore apply \eqref{le:useful} and \eqref{eq:dualIII} to derive
 \bea
 \FF_{\mu}^{\,\,\,\si}\II_{\nu\sigma \a\b} +\FF_{\nu}^{\,\,\,\si} \II_{\mu\sigma\a\b}=
 \frac{1}{2}\g_{\mu\nu}\FF^{\ga\de}\II_{\ga\de\a\b}.\label{b8.2''} 
\eea
We observe also that
\bea
\FF^{\mu\nu}\II_{\a\b\mu\nu}=\FF_{\a\b}.\label{b34}
\eea
Following \cite{Ma1}, we define the tensor-field  $\QQ\in\mathbf{T}_4^0(\mathbf{N})$,
\begin{equation}\label{s5}
\QQ_{\al\be\mu\nu}=(1-\sigma)^{-1}\big(\FF_{\al\be}\FF_{\mu\nu}-\frac{1}{3}\FF^2\II_{\al\be\mu\nu}\big).
\end{equation}

We show now that $\QQ$ is a self-dual Weyl field on $\mathbf{N}$.

\begin{proposition}\label{tensorQ}
The tensor-field $\QQ$ is a self-dual Weyl field, i.e.
\begin{equation*}
\begin{cases}
&\QQ_{\al\be\mu\nu}=-\QQ_{\be\al\mu\nu}=-\QQ_{\al\be\nu\mu}=\QQ_{\mu\nu\al\be};\\
&\QQ_{\a\b\mu\nu}+\QQ_{\a\mu\nu\be}+\QQ_{\a\nu\be\mu}=0;\\
&\g^{\be\nu}\QQ_{\al\be\mu\nu}=0,
\end{cases}
\end{equation*}
and
\begin{equation*}
\frac{1}{2}\in_{\mu\nu\rh\si}{\QQ_{\al\be}}^{\rh\si}=(-i)\QQ_{\al\be\mu\nu}.
\end{equation*}
\end{proposition}

\begin{proof}[Proof of Proposition \ref{tensorQ}] The identities 
\beaa
\QQ_{\al\be\mu\nu}=-\QQ_{\be\al\mu\nu}=-\QQ_{\al\be\nu\mu}=\QQ_{\mu\nu\al\be}
\eeaa
follow immediately from the definition. To prove
\beaa
\QQ_{\al[\be\mu\nu]}=\QQ_{\al\be\mu\nu}+\QQ_{\al\mu\nu\be}+\QQ_{\al\nu\be\mu}=0
\eeaa
it suffices to check, in view of  the identity  \eqref{s55},
\begin{equation}\label{b38}
\FF_{\al\be}\FF_{\mu\nu}+\FF_{\al\mu}\FF_{\nu\be}+\FF_{\al\nu}\FF_{\be\mu}=\frac{i}{4}\in_{\al\be\mu\nu}\cdot \FF^2.
\end{equation}
Since $\FF$ is a $2$-form, the left-hand side of \eqref{b38} is a $4$-form on $\mathbf{N}$ (which has dimension $4$). Thus, for \eqref{b38} it suffices to check
\beaa
\in^{\al\be\mu\nu}\big(\FF_{\al\be}\FF_{\mu\nu}+\FF_{\al\mu}\FF_{\nu\be}+\FF_{\al\nu}\FF_{\be\mu}\big)=-6i \FF^2.
\eeaa
This follows since the left-hand side of the above equation is equal to 
$6\FF_{\al\be}{\dual\FF}^{\al\be}=-6i\FF^2$.

We compute
\begin{equation*}
\begin{split}
\g^{\be\nu}\QQ_{\al\be\mu\nu}&=(1-\sigma)^{-1}\Big(\FF_{\al\be}{\FF_{\mu}}^\be-\frac{1}{3}\FF^2\cdot \g^{\be\nu}\II_{\al\be\mu\nu}\Big)=0. 
\end{split}
\end{equation*}
Also
\begin{equation*}
\frac{1}{2}\in_{\mu\nu\rh\si}{\QQ_{\al\be}}^{\rh\si}=(1-\sigma)^{-1}\big(\FF_{\al\be}{\dual\FF}_{\mu\nu}-\frac{1}{3}\FF^2{\dual\II}_{\al\be\mu\nu}\big)=(-i)\QQ_{\al\be\mu\nu}.
\end{equation*}
This completes the proof of the proposition.
\end{proof}

We define now the Mars-Simon tensor.
\begin{definition}\label{MStensor}
We define the self-dual Weyl field $\SS$,
\begin{equation}\label{s6}
\SS=\RR+6\QQ.
\end{equation}
\end{definition}

\begin{remark}\label{Swelldef}
Since $\Re\sigma=-\T^\al\T_\al\leq 0$ on $S_0$, it follows from the definition of the constant $A_0$ in section \ref{preliminaries} that $\Re(1-\sigma)\geq 1/2$ in a neighborhood $\O_{\ep_2}\subseteq\widetilde{\mathbf{M}}$ of $S_0$, for some $\ep_2\leq \ep_0$ that depends only on $A_0$. In particular, the tensor $\Ss$ is well defined in $\O_{\ep_2}$.
\end{remark}

\subsection{A covariant wave equation for $\SS$}
Our main goal now is to show that $\SS$ verifies a 
covariant wave equation. We first 
  calculate  its spacetime divergence
 $\D^\a\SS_{\a\b\mu\nu}$. Clearly, it suffices to calculate
 $\D^\a\QQ_{\a\b\mu\nu}$.
Recalling the definition of 
 the $1$-form
$
\si_\al=2\T^\nu\FF_{\nu\al}
$
and using the definition \eqref{s5} we compute
\begin{equation}\label{s10}
\begin{split}
\D_\rh\mathcal{Q}_{\al\be\mu\nu}&=(1-\sigma)^{-1}\D_\rh\mathcal{F}_{\al\be}\cdot \mathcal{F}_{\mu\nu}+(1-\sigma)^{-1}\mathcal{F}_{\al\be}\cdot \D_\rh\mathcal{F}_{\mu\nu}\\
&-\frac{1}{3}(1-\sigma)^{-1}\D_\rho\FF^2\cdot\II_{\al\be\mu\nu}+(1-\sigma)^{-2}\si_\rho\big(\FF_{\al\be}\FF_{\mu\nu}-\frac{1}{3}\FF^2\II_{\al\be\mu\nu}\big).
\end{split}
\end{equation}
Using \eqref{eq:Ricci}, \eqref{b34}, and  $\RR=\mathcal{S}-6\mathcal{Q}$, we have
\begin{equation}\label{f1}
\begin{split}
\D_\rho\FF_{\ga\de}&=\T^\nu\RR_{\nu\rho \ga\de}=\T^\nu \SS_{\nu\rho \ga\de}-6\cdot \T^\nu\QQ_{\nu\rho \ga\de}\\
&=-3(1-\sigma)^{-1}\si_\rho\FF_{\ga\de}+2(1-\sigma)^{-1}\mathcal{F}^2\cdot\T^\nu\II_{\nu\rho \ga\de}+\T^\nu \SS_{\nu\rho \ga\de}.
\end{split}
\end{equation}
Thus,
\begin{equation}\label{s11.1}
\begin{split}
(1-\sigma)^{-1}\mathcal{F}_{\al\be}\cdot \D_\rh\mathcal{F}_{\mu\nu}
&=-3(1-\sigma)^{-2}\cdot \si_\rho\mathcal{F}_{\al\be}\mathcal{F}_{\mu\nu}\\
&+2(1-\sigma)^{-2}\FF^2\FF_{\al\be}\T^\la\II_{\la\rh\mu\nu}+\JJ_1(\SS)_{\rho\a\b\mu\nu},
\end{split}
\end{equation}
where
\beaa
 \JJ_1(\SS)_{\rho\a\b\mu\nu}=(1-\sigma)^{-1}\cdot\FF_{\a\b}\T^\la S_{\la\rho\mu\nu}.
\eeaa
Observe that, in view of \eqref{f1} and \eqref{b34}
\bea
\D_\rho\,\mathcal{F}^2&=&2\D_\rho\mathcal{F}_{\ga\de}\cdot \mathcal{F}^{\ga\de}=-4(1-\sigma)^{-1}\FF^2\si_\rho+2\T^\nu \SS_{\nu\rho\ga\de}\FF^{\ga\de}.\label{Ernst-important}
\eea
Thus
\begin{equation}\label{s11.3}
\begin{split}
-\frac{1}{3}(1-\sigma)^{-1}\D_\rho\FF^2\cdot \mathcal{I}_{\al\be\mu\nu}=\frac{4}{3}(1-\sigma)^{-2}\FF^2\cdot \si_\rho\mathcal{I}_{\al\be\mu\nu}+\JJ_2(\SS)_{\rh\al\be\mu\nu},
\end{split}
\end{equation}
where,
\beaa
\JJ_2(\SS)_{\rh\al\be\mu\nu}&=&-\frac{2}{3}(1-\sigma)^{-1}\cdot \T^\la\SS_{\la\rho\ga\de}\FF^{\ga\de}\II_{\a\b\mu\nu}.
\eeaa
We combine \eqref{s10}, \eqref{s11.1}, and \eqref{s11.3} to write
\bea
\D_\rh\mathcal{Q}_{\al\be\mu\nu}&=&(1-\sigma)^{-1}\D_\rh\mathcal{F}_{\al\be}\cdot \mathcal{F}_{\mu\nu}-2(1-\sigma)^{-2}\si_\rho\mathcal{F}_{\al\be}\mathcal{F}_{\mu\nu}\nn\\
&+&2(1-\sigma)^{-2}\FF^2\FF_{\al\be}\T^\la\cdot \II_{\la\rh\mu\nu}+(1-\sigma)^{-2}\FF^2\si_\rho\mathcal{I}_{\al\be\mu\nu}\label{s12}\\
&+&\JJ_1(\SS)_{\rh\al\be\mu\nu}+\JJ_2(\SS)_{\rh\al\be\mu\nu}\nn.
\eea

We  are now ready to compute the divergence
 $\D^\si\mathcal{Q}_{\be\si \mu\nu}$. Using \eqref{s12} and 
 the Maxwell equations \eqref{eq:Max1'} we derive
\beaa
\D^\b\mathcal{Q}_{\a\b \mu\nu}&=&\JJ''(\SS)_{\a\mu\nu}-2(1-\sigma)^{-2}\si_\rho{\mathcal{F}_{\a}}^\rho\mathcal{F}_{\mu\nu}\\
&+&2(1-\sigma)^{-2}\FF^2{\FF_\a}^\rh\T^\la\II_{\la\rh\mu\nu}+(1-\sigma)^{-2}\FF^2\si^\b\mathcal{I}_{\a\b \mu\nu};\\
\JJ''(\SS)_{\a\mu\nu}&=&\g^{\rho\b}(\JJ_1(\SS)_{\rh\al\be\mu\nu}+\JJ_2(\SS)_{\rh\al\be\mu\nu}).\\
\eeaa
Using \eqref{b8.2'} and the definition of $\si_\rho$ we derive,
\begin{equation}\label{s26}
\begin{split}
-2(1-\sigma)^{-2}\cdot \si_\rho{\mathcal{F}_{\a}}^\rho\mathcal{F}_{\mu\nu}&=-2(1-\sigma)^{-2}\cdot 2\T^\la\FF_{\la\rho}{\mathcal{F}_{\a}}^\rho\mathcal{F}_{\mu\nu}\\
&=-2(1-\sigma)^{-2}\cdot 2\T^\la\cdot \frac{1}{4}\g_{\la\a}\FF^2\cdot \mathcal{F}_{\mu\nu}\\
&=-(1-\sigma)^{-2}\FF^2\T_\a\FF_{\mu\nu}.
\end{split}
\end{equation}
Using  \eqref{b8.2''}, \eqref{b34}, and the definitions,
\begin{equation}\label{s27}
\begin{split}
2(1-\sigma)^{-2}\FF^2\cdot&{\FF_\a}^\rh\T^\la\II_{\la\rh\mu\nu}+(1-\sigma)^{-2}\FF^2\cdot \si^\be\mathcal{I}_{\a\be \mu\nu}\\
&=2(1-\sigma)^{-2}\FF^2({\FF_\a}^\rh \T^\la\II_{\la\rh\mu \nu}+\T^\la{\FF_\la}^\rh \II_{\a\rh\mu\nu})\\
&=2(1-\sigma)^{-2}\FF^2\T^\la({\FF_\a}^\rh \II_{\la\rh\mu \nu}+{\FF_\la}^\rh \II_{\a\rh\mu\nu})\\
&=2(1-\sigma)^{-2}\FF^2\T^\la\cdot \frac{1}{2}\g_{\la\a}\FF^{\rh\si}\II_{\rh\si\mu\nu}\\
&=(1-\sigma)^{-2}\FF^2\T_\a\FF_{\mu\nu}.
\end{split}
\end{equation}
Thus, using  \eqref{s26}, and \eqref{s27}, 
we derive,
\beaa
\D^\si\mathcal{Q}_{\a\si \mu\nu}=\JJ''(\SS)_{\a\mu\nu},
\eeaa
with
\begin{equation*}
\begin{split}
\JJ''(\SS)_{\a\mu\nu}=(1-\sigma)^{-1}\T^\la\SS_{\la\rho\ga\de}\big(\FF_{\a}^{\,\,\,\rho}\de^\ga_\mu\de^\de_\nu-(2/3)\FF^{\ga\de}\II_{\a\,\,\mu\nu}^{\,\,\,\rho}\big).
\end{split}
\end{equation*}
Since, according to the  Bianchi identities, and the Einstein equations,  
we have $\D^\si\RR_{\be\si \mu\nu}=0$ we deduce the following.

\begin{theorem} 
\label{thm-MS-eqts}The Mars-Simon tensor $\SS$ verifies,
\bea
\D^\si\SS_{\si\a \mu\nu}=\JJ(\SS)_{\a\mu\nu}.\label{Div-QQ}
\eea
where,
\beaa
\JJ(\SS)_{\a\mu\nu}=-6(1-\sigma)^{-1}\T^\la\SS_{\la\rho\ga\de}\big(\FF_{\a}^{\,\,\,\rho}\de^\ga_\mu\de^\de_\nu-(2/3)\FF^{\ga\de}\II_{\a\,\,\mu\nu}^{\,\,\,\rho}\big).
\eeaa
\end{theorem}

As a consequence of the theorem we  deduce from Proposition \ref{prop:Bianchi-dual}
 and the self-duality of $\SS$ and $\JJ$,
\bea
\D_{[\si} \SS_{\mu\nu]\a\b}=-i \in_{\rho\si\mu\nu}\JJ^\rho_{\,\,\,\,\a\b}(\SS).\label{Bianchi-SS}
\eea
In the following calculations the precise form of $\JJ(\SS)$ is not important, 
we only need to keep track of the fact that it  is a multiple of $\SS$.
\begin{definition}\label{DEF2}
We denote by $\mathcal{M}(\SS)$ any $k$-tensor  with the property that there is a 
smooth tensor-field  $\mathcal{A}$  such that 
\begin{equation}\label{s9}
\mathcal{M}(\SS)_{\alpha_1\ldots\alpha_k}=\SS_{\be_1\ldots \be_4}{\mathcal{A}^{\be_1\ldots\be_4}}_{\alpha_1\ldots\alpha_k}.
\end{equation}
Similarly we denote by  $\mathcal{M}(\SS,\D \SS)$  any $k$-tensor with the property that there 
exist smooth tensor-fields $\AA$ and $\BB$ such that
\begin{equation}\label{s9.2}
\mathcal{M}(\SS, \D \SS)_{\alpha_1\ldots\alpha_k}=\SS_{\be_1\ldots \be_4}{\AA^{\be_1\ldots\be_4}}_{\alpha_1\ldots\alpha_k}+\D_{\be_5}\SS_{\be_1\ldots \be_4}{\BB^{\be_1\ldots\be_5}}_{\alpha_1\ldots\alpha_k}.
\end{equation}
\end{definition}
We state the main result of this section.

\begin{theorem}\label{wave}
We have
\begin{equation}\label{todo}
\square_\g \SS= \mathcal{M}(\SS,\D\mathcal{S}).
\end{equation}
\end{theorem}

\begin{proof}[Proof of Theorem \ref{wave}]
The result follows easily from the equations \eqref{Div-QQ} and
\eqref{Bianchi-SS}
\beaa
&\D^\a\SS_{\a\b\ga\de}=\JJ(\SS)_{\b\ga\de}\\
&\D_{[\si} \SS_{\a\b]\ga\de}=-i \in_{\rho\si\a\b}\JJ^\rho_{\,\,\,\,\ga\de}(\SS).
\eeaa
Indeed,
differentiating once more the second equation we derive,
\beaa
\D^\si (\D_\si\SS_{\a\b\ga\de}+\D_\a\SS_{\b\si\ga\de}+\D_\b\SS_{\si\a\ga\de})=
-i \in_{\rho\si\a\b}\D^\si\JJ^\rho_{\,\,\,\,\ga\de}(\SS).
\eeaa
Thus, after commuting covariant derivatives and using the first equation
 we derive,
 \beaa
 \square_\g \SS_{\a\b\ga\de}&=&\MM(\SS, \D\SS)_{\a\b\ga\de}
 \eeaa
 as desired.
\end{proof}

\section{Vanishing of $\SS$ on the horizon}\label{horizon}
In this  section we prove that the Mars--Simon tensor $\SS$ vanishes on $\HH^+\cup\HH^-$.

\begin{proposition}\label{shorizon}
The Mars--Simon tensor $\SS$ vanishes along
the horizon $\HH^+\cup\HH^-$.
\end{proposition}

The rest of the section is concerned with the proof of Proposition \ref{shorizon}. Recall, see Remark \ref{Swelldef}, that the tensor $\Ss$ is well defined on $S_0$. We will use the notation in the appendix. Assume $\NN$ is a null hypersurface (in our case $\NN=\HH^+$ or $\NN=\HH^-$) and let $l\in\mathbf{T}(\NN)$ denote a null vector-field orthogonal to $\NN$. The Lie bracket $[X,Y]$  of any two 
vector-fields  $X,Y$  tangent to $\NN$  is again 
tangent to $\NN$ and therefore
\beaa
&&\g(\D_X l, Y)-\g(\D_Yl,X)=-\g(l, [X,Y])=0\text { and }\g(\D_l l, X)=-\g(l, \D_l X)=0.
\eeaa
In particular we infer that along $\NN$ $\,^{(h)}\xi$ vanishes
identically and $\,^{(h)}\chi$ is symmetric.  
\begin{definition}
Given a null hypersurface $\NN$ and $l$ a  fixed non-vanishing
null vector-field on it we define $\chi(X, Y)=\g(\D_Xl, Y)$, $X,Y\in \mathbf{T}(\NN)$,  the null second fundamental form of $\NN$.  We denote by $\trch$ the
 trace\footnote{The trace is well defined since $\chi(X,l)=\ga(X,l)=0$ for 
 all $X\in \mathbf{T}(\NN)$.} of $\chi$
with respect  to the induced  metric and by $\chih$
the traceless part of $\chi$, i.e. $\chih=\chi-\frac 1 2 
\ga\trch$, with $\ga$ the degenerate metric on $\NN$ induced by $g$.
\end{definition}

In view of the definitions  \eqref{fo10}, writing 
$m=(e_1+i e_2)/\sqrt{2}$, with $e_1, e_2$ an arbitrary horizontal orthonormal frame,  we 
deduce that,
\beaa
\th&=&(\chi_{11}+\chi_{22})/2=\trch/2\\
\vartheta&=&(\chi_{11}-\chi_{22})/2+i\chi_{12}
\eeaa

We now restrict  our considerations to that
of  a non-expanding null hypersurface. In other words we 
assume that $\th=\trch/2$ vanishes identically along $\NN$.  In view of the null structure  equation \eqref{fo22} and the vanishing of $\xi={}^{(h)}\xi(m)$, we deduce
that $|\va|^2=0$ along $\NN$  therefore $\va\equiv 0$. Therefore the full null second fundamental form of
$\NN$ vanishes identically. We now consider  the null structure equation \eqref{fo21}. Since $\xi, \th,
\va$ vanish we deduce that $\Psi_{(2)}(R)$ 
must vanish along $\NN$.  Similarly we  deduce that $\Psi_{(1)}(R) $ vanishes  along $\NN$ from 
equation \eqref{fo26}. 
 Finally, we consider the Bianchi equations with zero source $J$. From
  \eqref{Bi5} we   deduce that $D\Psi_{(0)}$ vanishes
  identically along $\NN$.  Observe also
that   $\Psi_{(0)}(R)$ is invariant under general  changes  of the null pair $(l, \lb)$ which keep  $l$ orthogonal to $\NN$. Indeed $\Psi_{(0)}(R)$ is
always invariant under the scale transformations
$l'=fl, \lb'=f^{-1}\lb$. On the other hand if we keep
$l$ fixed and perform the general transformations 
$\lb'=\lb+Al +Bm+\overline{Bm}$ we 
easily find that $\Psi_{(0)}'(R)$ differs from  $\Psi_{(0)}(R)$ by a linear combination of 
 $\Psi_{(2)}(R)$ and  $\Psi_{(1)}(R)$.

We have thus proved the
following.
\begin{proposition}\label{NE-hor}
Let  $(l, \lb)$ be a null pair in an open set $\N$ with
$l$ orthogonal to a  non-expanding   null hypersurface in $\NN\subset \N$. 
Then $\,^{(h)}\xi$ and $\,^{(h)}\chi$ vanish
identically on $\NN$. Moreover the
curvature components  $\Psi_{(2)}(R)$  and
 $\Psi_{(1)}(R)$ (or equivalently,
 $\a(R), \b(R)$)  vanish along $\NN$ and
 the invariant  $\Psi_{(0)}(R)$ (or equivalently $ \rho(R+i\dual R)$) is constant along the null generators.
\end{proposition}

We apply this proposition to the surfaces $\HH^+$ and $\HH^-$ to establish the following facts. Recall that $\RR=R+i\dual R$. 

\begin{enumerate}
\item The null second fundamental form $\chi$,
 respectively $\chib$,
vanishes identically along $\HH^+$, respectively  $\HH^-$. 
\item The null curvature  components $\a=\a(\RR)$ and $\b=\b(\RR)$ (respectively $\aa(\RR)$, $\bb(\RR)$), vanish identically along $\HH^+$  (respectively $\HH^-$).
\item The null curvature component $\rho(\RR)$ is invariant and constant
along the  null generators of both  $\HH^+$ and  $\HH^-$. 
\item All null curvature components, except $\rho(\RR)$, 
vanish along the bifurcate sphere  $S_0$. We also have $\chi=\chib=0$ on $S_0$.
\end{enumerate}
Consider an adapted null frame $e_1, e_2,  e_3=\lb, e_4=l$ in $\mathbf{O}$ with $l$ tangent to the null generators of $\HH^+$
 and $\lb$ tangent to the null generators of $\HH^-$. Thus,
\begin{equation*}
\begin{split}
&\g(l,l)=\g(\lb, \lb)=0, \quad \g(l,\lb)=-1,\quad\g(l, e_a)=\g(\lb, e_a)=0,\\
&\g(e_a, e_b)=\de_{ab}, \quad a,b=1,2,\quad\in_{12}=\in(e_1,e_2,e_3,e_4)=1.
\end{split}
\end{equation*}
We introduce the notation,
\bea
\a_a(\FF)=\FF(e_a, l), \quad \aa_a(\FF)=\FF(e_a,\lb),\quad
\rho(\FF)=\FF(\lb, l).\label{null-dec-F}
\eea
Observe that the null components  $\a_a(\FF), \aa_a(\FF), \rho(\FF)$ completely determine the antisymmetric, self-dual tensor
$\FF$. Indeed,
$
-i\FF_{34}=(\dual\FF)_{34}=\frac 1 2 \in_{34 ab}\FF^{ab}=\frac 12 \in_{ab} \FF^{ab}
$. 
Hence,
\bea
\FF_{ab}=-i\in_{ab}\rho(\FF).\label{alexadd1}
\eea
We claim  that $\a(\FF)$ vanishes on $\HH^+$ while $\aa(\FF)$
vanishes on $\HH^-$,
\bea
\a(\FF)=0\quad \mbox{on}\,\, \HH^+, \qquad \aa(\FF)=0\quad \mbox{on}\,\, \HH^-.\label{avanish}
\eea
Indeed since $g(\T,\, l)=0$ on $\HH^+$ (see the assumption {\bf{SBF}}) and the null second fundamental
form $\chi$ vanishes identically on $\HH^+$,
\beaa
F(e_a, l)=-g(\T, \D_{e_a}l)=-\chi(\T, e_a)=0\text{ on }\HH^+.
\eeaa 
 On the other hand, 
  \beaa
  \dual F_{a4}=\frac 12 \in_{a4\mu\nu}F^{\mu\nu}=\in_{a4b3}F^{b3}=-\in_{a4b3}F_{b4}=0.
  \eeaa
   Hence $ \a_a(\FF)=\FF_{a4}=F_{a4}+i\dual F_{a4}=0$ on $\HH^+$.
The proof of vanishing of $\aa(\FF)$ on $\HH^{-}$ is similar. 
We infer that both $\a(\FF)$ and $\aa(\FF)$ have to vanish
along the bifurcate sphere $S_0$.
We also observe,
\beaa
\FF^2&=&\FF_{\mu\nu}\FF^{\mu\nu}=
2\FF^{34}\FF_{34}+\FF^{ab}\FF_{ab}
=-4\FF_{34}^2=-4\rho(\FF)^2\text{ on }S_0.
\eeaa
Since $\FF^2$ does not vanish on $S_0$ we infer that $\rho(\FF)$ cannot vanish on $S_0$.

Consider now the Mars-Simon tensor \eqref{s6}. To show
that the Weyl tensor  $\SS$ vanishes along the $\HH^+\cup\HH^-$ it suffices to
show that all its  null  components (see Appendix)
$\a(\SS),\b(\SS),\rho(\SS), \aa(\SS),\bb(\SS)$,
relative to an arbitrary, adapted, null frame $(e_1, e_2, \lb, l)$, vanish along $\HH^+\cup\HH^-$. We first show that
\bea
\a(\SS)=\b(\SS)=0\quad \mbox{on} \,\, \HH^+,\qquad 
\aa(\SS)=\bb(\SS)=0 \quad \mbox{on} \,\, \HH^-.
\eea
Indeed, 
\beaa
\II(l,e_a,l,e_b)&=&0,\qquad
\II(e_a,l,\lb, l)=0,\quad
\II(l,\lb,l,\lb)=-1/4.
\eeaa
Therefore along $\HH^{+}$, where  $\a(\FF),  \, \a(\RR),\, \b(\RR)$
vanish,
\beaa
\a(\SS)_{ab}=\b(\SS)_a=0,
\eeaa
using the formula $\SS=\RR+6\QQ$. Similarly we infer that  $\aa(\SS)=\bb(\SS)=0$ along $\HH^-$.

We show now that $\rho(\SS)$ vanishes on $S_0$. This is where we need the main technical assumption \eqref{Main-Cond1} along  $S_0$,
\beaa
(1-\si)^4=-4M^2\FF^2.
\eeaa
Differentiating it along $S_0$ we find,
\beaa
0=\D_a (\FF^2(1-\sigma)^{-4})=(1-\sigma)^{-4}(\D_a\FF^2+4(1-\sigma)^{-1}\sigma_a).
\eeaa
On the other hand, recalling formula \eqref{Ernst-important}
\beaa
D_\a\FF^2+4(1-\sigma)^{-1}\FF^2\sigma_\a=2\T^\la \SS_{\la\a\ga\de} \FF^{\ga\de},
\eeaa
We deduce that
\bea
\T^\la \SS_{\la a\ga\de} \FF^{\ga\de}=0\text{ on }S_0.\label{vanishing1}
\eea
Recall  $\T$ is tangent on $S_0$  and can only vanish
at a discrete set of points (see assumption {\bf{SBF}} in subsection \ref{maintheorem}). Therefore, at a point where $g(\T,\T)^{1/2} $,  does not vanish
we  can   introduce an orthonormal frame $e_1, e_2$  with  $\T=\g(\T,\T)^{1/2}e_1$. 

We now expand the left  hand side of  \eqref{vanishing1} using \eqref{alexadd1} while setting the index $a=2$,
\beaa
0&=&\T^\la \SS_{\la 2 \ga\de} \FF^{\ga\de}=2\T^\la \SS_{\la 2 34} \FF^{34}+\T^\la \SS_{\la 2 cd} \FF^{cd}
=-2\T^\la \SS_{\la 2 34} \FF_{34}-i \T^\la \SS_{\la 2  cd}\in^{cd}\rho(\FF)\\
&=&-\g(\T,\T)^{1/2}\rho(\FF)\big( 2\SS_{1 2 34}+i\SS_{12 cd}\in ^{cd})
=4i\g(\T,\T)^{1/2}\rho(\FF)\rho(\SS).
\eeaa
The last equality follows from ( see  \eqref{dualWW-form})
  \beaa
  \SS_{1234}=-i\rho(\SS),\quad
  \SS_{12 cd}=-\in_{cd}\rho(\SS).
  \eeaa
  Therefore, at all points of $S_0$ where $\T$ does not vanish
  we infer that $\rho(\SS)=0$ (since $\rho(\FF)$ cannot vanish on $S_0$, due to \eqref{Main-Cond1} and \eqref{Main-Cond2}). Since the set of such points
  is  dense in $S_0$ we conclude that $\rho(\SS)$ vanishes
  identically on the bifurcate sphere $S_0$. 
   We have thus proved the following.
   \begin{proposition}
   The components $\a(\SS)$, $ \b(\SS)$  vanish along $\HH^+$
   while $\aa(\SS)$, $ \bb(\SS)$  vanish along $\HH^-$.  In addition, if \eqref{Main-Cond1} holds then $\rho(\SS)$ also vanishes  on $S_0$.
   \end{proposition}

   To show that  $\rho(\SS), \bb(\SS),
   \aa(\SS)$  vanish on $\HH^+$ we need to use the Bianchi
   equations  (see Theorem \ref{thm-MS-eqts}),  
\bea
\D^\si\SS_{\si\a \mu\nu}&=&\JJ(\SS)_{\a\mu\nu}=-6(1-\sigma)^{-1}\T^\la\SS_{\la\rho\ga\de}\bigg(\FF_{\a}^{\,\,\,\rho}\de^\ga_\mu\de^\de_\nu-\frac{2}{3}\FF^{\ga\de}\II_{\a\,\,\mu\nu}^{\,\,\,\rho}\bigg)\label{BianSS2}.
\eea
Assume, without loss of generality, that the null generating
vector-field $l$ is geodesic along $\HH^+$, i.e.
 $\D_l l=0$.   Since both $\b(\SS)=\a(\SS)=0$  along $\HH$
     we deduce\footnote{Alternatively
   we can use  the null  Bianchi identities of the Appendix.}
 directly  that $\rho(\SS)$ must verify the   equation,
 \bea
 \nab_l\rho(\SS)=- \JJ(\SS)_{434}.
 \eea
 To deduce that $\rho(\SS)$ vanishes identically on $\HH^+$
 it only remains to verify that 
 that $\JJ(\SS)_{434}$ vanishes on  $\HH^+$. Clearly
 \beaa
 \JJ(\SS)_{434}=-6(1-\sigma)^{-1}\T^\la\SS_{\la \rho\ga\de}\big(\FF_{4}^{\,\,\,\rho}\de^\ga_3\de^\de_4-\frac{2}{3}\FF^{\ga\de}\II_{4\,\,34}^{\,\,\,\rho}\big).
 \eeaa
Observe that the only choice of the index $\rho$ for which the expression inside brackets does not vanish is $\rho=4$. Thus
 \begin{equation*}
\begin{split}
  \JJ(\SS)_{434}&=-6(1-\sigma)^{-1}\T^\la\SS_{\la 4\ga\de}\big(\FF_{34}\de^\ga_3\de^\de_4+\frac{1}{6}\FF^{\ga\de}\big)\\
&=-6(1-\sigma)^{-1}\T^\la\SS_{\la 434}\FF_{34}-(1-\sigma)^{-1}\T^\la\SS_{\la 4\ga\de}\FF^{\ga\de}.
\end{split}
 \end{equation*}
 Since $\a(\SS), \b(\SS)$ vanish the only pair of indices $\ga\de$ for which
 $\T^\la\SS_{\la 4\ga\de}$ does not vanish is when either of the two indices is a $3$
 and the other is $a\in\{1,2\}$. Since $\al(\FF)=0$, it follows that $\JJ(\SS)_{434}$ vanishes identically  as stated. Thus 
$\rho(\SS)$ is constant along generators and vanishes on $S_0$.
We conclude that $\rho(\SS)$ vanishes identically on $\HH^+$.

To show that $\bb(\SS)$  also vanishes we derive
a transport equation for it along the generators of
$\HH^+$. In view of the vanishing of $\a(\SS), \b(\SS), \rho(\SS)$
we can directly deduce\footnote{We  also refer the reader to
 the Appendix for the definition of the horizontal covariant derivative $\nab_l$.}  (see also Appendix)  it from \eqref{BianSS2},
\beaa
\nab_L\bb(\SS)_a=\JJ(\SS)_{4a3}
\eeaa
Thus, since $\bb(\SS)$ vanishes on $S_0$,
to deduce that it vanishes everywhere on $\HH^+$
we only need to verify that  $\JJ(\SS)_{4a3}$ vanishes
identically on $\HH^+$. Now,
\begin{equation*}
\begin{split}
 \JJ(\SS)_{4a3}&=-6(1-\sigma)^{-1}\T^\la\SS_{\la\rho\ga\de}\big(\FF_{4}^{\,\,\,\rho}\de^\ga_a\de^\de_3-\frac{2}{3}\FF^{\ga\de}\II_{4\,\,a3}^{\,\,\,\rho}\big)\\
&=-6(1-\sigma)^{-1}\T^\la\SS_{\la4 a3}\FF_{43}+8(1-\sigma)^{-1}\T^\la\SS_{\la b34}\FF_{43}\II_{4\,\,a3}^{\,\,\,b}\\
&+4(1-\sigma)^{-1}\T^\la\SS_{\la bcd}\FF^{cd}\II_{4\,\,a3}^{\,\,\,b}+8(1-\sigma)^{-1}\T^\la\SS_{\la b4c}\FF^{4c}\II_{4\,\,a3}^{\,\,\,b}.
\end{split}
\end{equation*}
Since $\al(\SS)$, $\be(\SS)$ and $\rho(\SS)$ vanish, it follows that $\SS_{b4a3}=\SS_{ab34}=\SS_{abcd}=\SS_{4bcd}=\SS_{4b4c}=0$, which gives $\JJ(\SS)_{4a3}=0$.

 To show that $\aa(\SS)$ also vanishes on $\HH^+$
 we derive another transport equation for it. Since all
 other  components of $\SS$ have already been shown to vanish
 we easily derive, from \eqref{BianSS2},
 \bea
 \nab_L\aa(\SS)_{ab}=-\JJ(\SS)_{a3b}.
 \eea
 Since $\aa(\SS)$ vanishes on $S_0$ 
 it only remains to check that $\JJ(\SS)_{a3b}$
 vanishes identically.
    This  can be  checked as before    taking 
    advantage of  the cancellations 
 of    all the other null components of $\SS$. 
 Therefore $\SS$ vanishes 
 along the entire event horizon.

\section{Vanishing of $\SS$ in a neighborhood of the bifurcate sphere}\label{sectionS=01}

Let $\O_\ep=\{x\in \O:|u_-|< \ep, |\up|< \ep\}$ as in section \ref{preliminaries}. In this section we show that the tensor $\Ss$ vanishes in a neighborhood of the bifurcate sphere $S_0$ in $\mathbf{E}$.

\begin{proposition}\label{Lemmaa1}
There is $r_1=r_1(A_0)>0$ such that 
\begin{equation*}
\Ss\equiv 0\text{ in }\mathbf{O}_{r_1}\cap\mathbf{E}.
\end{equation*}
\end{proposition}

The rest of this section is concerned with the proof of Proposition \ref{Lemmaa1}. Recall, see Remark \ref{Swelldef}, that the tensor $\Ss$ is well defined and smooth on $\mathbf{O}_{\eps_2}$ for some $\eps_2=\eps_2(A_0)\in(0,\eps_0)$.
Recall that we have
\beaa
\g(L_{\pm},L_\pm)=0,\quad \g(L_+,L_-)=\Om>\frac 1 2\text{ in }\O_{\ep_0}.
\eeaa
Moreover both $L_+, L_-$ are orthogonal to the $2$-surfaces $S_{\um,\up}=\HH_{\um}\cap\HH_{\up}$.
We  choose, locally at any point $p\in S_{\um,\up}$,
an  orthonormal frame $(L_a)_{a=1,2}$ tangent to
$S_{\um,\up}$. Thus, relative to the null frame $L_1, L_2, L_3=L_-,
L_4=L_+$ the metric $\g$ takes the form,
\begin{equation}\label{Car30}
\begin{cases}
&\g_{ab}=\de_{ab}, \quad \g_{a3}=\g_{a4}=0,\quad   a,b=1,2\\
&\g_{33}=\g_{44}=0,  \quad \g_{34}= \Omega.
\end{cases}
\end{equation}
Also, for the inverse metric,
\begin{equation}\label{Car52}
\begin{cases}
&\g^{ab}=\de^{ab}, \quad \g^{a3}=\g^{a4}=0,\quad   a,b=1,2\\
&  \g^{33}=g^{44}=0,  \quad \g^{34}= \Omega^{-1}.
\end{cases}
\end{equation}

We denote by $O(1)$ any quantity with absolute value
uniformly bounded by a positive constant which depends 
only on $A_0$ (in particular $L_\al(\Omega)=O(1)$, $\al=1,2,3,4$). In view of the definitions of $u_{\pm}$ and $L_{\pm}$ 
we have,
\begin{equation}\label{lubounds}
L_1(u_\pm)=L_2(u_\pm)=L_-(u_-)=L_+(u_+)=0,\quad L_-(u_+)=L_+(u_-)=\Om.
\end{equation}
For $\ep\in(0,\ep_0]$ we define the weight  function in $\O_{\ep^2}$,
\begin{equation}\label{Car2}
h_\ep=\eps^{-1}(u_++\ep)(u_-+\ep).
\end{equation}
Observe that,
\bea
  L_4(h_\ep)=\eps^{-1}(u_++\ep)\Om,
\quad L_3(h_\ep)=\eps^{-1}(u_-+\ep)\Om,\quad L_a(h_\ep)=0,\,\,a=1,2.\label{eq:Car1.1}
\eea
Also, using \eqref{lubounds} and \eqref{eq:Car1.1}
\bea\label{car:D2f}
\begin{cases}
&(\D^2h_\ep)_{33}=O(1),\quad(\D^2h_\ep)_{44}=O(1),\\
&(\D^2h_\ep)_{34}=\eps^{-1}\Omega^2+O(1),\quad(\D^2h_\ep)_{ab}=O(1),\qquad a,b=1,2,\\
&(\D^2h_\ep)_{3a}=O(1),\quad(\D^2h_\ep)_{4a}=O(1),\qquad a=1,2.
\end{cases}
\eea

Assume $x_0\in S_0$ is a fixed point and define, using the coordinate chart $\Phi^{x_0}:B_1\to B_1(x_0)$, $N^{x_0}:B_1(x_0)\to[0,\infty)$,
\begin{equation}\label{normlocal}
N^{x_0}(x)=|(\Phi^{x_0})^{-1}(x)|^2.
\end{equation}
We state now the main Carleman estimate needed in the proof of Proposition \ref{Lemmaa1}.

\begin{lemma}\label{Car}
There is $\eps\in (0,\eps_2)$ sufficiently small and $\widetilde{C}_\eps$ sufficiently large such that for any $x_0\in S_0$, any $\lambda\geq\widetilde{C}_\eps$, and any $\phi\in C^\infty_0(B_{\eps^{10}}(x_0))$
\begin{equation}\label{Car1}
\lambda \|e^{-\lambda f_\eps}\phi\|_{L^2}+\|e^{-\lambda f_\eps}|D^1\phi|\,\|_{L^2}\leq \widetilde{C}_\eps\lambda^{-1/2}\|e^{-\lambda f_\eps}\,\square_{\g}\phi\|_{L^2},
\end{equation}
where $f_\ep=\ln(h_\eps+\eps^{12}N^{x_0})$, see definitions \eqref{Car2} and \eqref{normlocal}.
\end{lemma}

\begin{proof}[Proof of Lemma \ref{Car}] It is clear that $B^{\eps^{10}}(x_0)\subseteq \O_{\ep^2}$ for $\eps$ sufficiently small (depending only on the constant $A_0$), thus the weight $f_\eps$ is well defined in $B_{\eps^{10}}(x_0)$. We apply Proposition \ref{Cargen} with $V=0$. It is clear that $\eps^{12}N^{x_0}$ is a negligible perturbation, in the sense of \eqref{smallweight}, for $\eps$ sufficiently small. It remains to prove that there is $\eps_1=\eps_1(A_0)>0$ such that the family of weights $\{h_\eps\}_{\eps\in (0,\eps_1)}$ satisfies conditions \eqref{po5}, \eqref{po3.2} and \eqref{po3}.

Let $\widetilde{C}$ denote constants that may depend only on $A_0$. The definition \eqref{Car2} easily gives $h_\eps(x_0)=\eps$, $|D^1h_\eps|\leq\widetilde{C}$ on $B_{\eps^{10}}(x_0)$, and $|D^jh_\eps|\leq\widetilde{C}\eps^{-1}$ on $B_{\eps^{10}}(x_0)$ for $j=2,3,4$. Thus condition \eqref{po5} is satisfied provided $\eps_1\leq \widetilde{C}^{-1}$.

Using \eqref{Car52}, \eqref{eq:Car1.1}, \eqref{car:D2f}, and $\Omega(x_0)=1$ we compute in the frame $L_1,L_2,L_3,L_4$
\begin{equation*}
\D^\al h_\eps(x_0)\D^\be h_\eps(x_0)(\D_\al h_\eps\D_\be h_\eps-\eps\D_\al\D_\be h_\eps)(x_0)=2+\eps O(1)\geq 1
\end{equation*}
if $\eps_1$ is sufficiently small. Thus condition \eqref{po3.2} is satisfied provided $\eps_1\leq \widetilde{C}^{-1}$.

Assume now $Y=Y^\al L_\al$ is a vector in $\mathbf{T}_{x_0}(\mathbf{M})$. We fix $\mu=\eps_1^{-1/2}$ and compute, using \eqref{eq:Car1.1}, \eqref{car:D2f}, and $\Omega(x_0)=1$,
\begin{equation*}
\begin{split}
&Y^\al Y^\be(\mu\g_{\al\be}-\D_\al\D_\be h_\eps)(x_0)+\eps^{-2}|Y^\al\D_\al h_\eps|^2\\
&=\mu((Y^1)^2+(Y^2)^2+2Y^3Y^4)-2\eps^{-1}Y^3Y^4+\eps^{-2}(Y^3+Y^4)^2+O(1)\sum_{\al=1}^4(Y^\al)^2\\
&\geq (\mu/2)[(Y^1)^2+(Y^2)^2]+(\eps^{-1}/2)[(Y^3)^2+(Y^4)^2]\\
&\geq (Y^1)^2+(Y^2)^2+(Y^3)^2+(Y^4)^2
\end{split}
\end{equation*}
if $\eps_1$ is sufficiently small. We notice now that we can write $Y=X^\al\partial_\al$ in the coordinate frame $\partial_1,\partial_2,\partial_3,\partial_4$, and $|X^\alpha|\leq \widetilde{C}(|Y^1|+|Y^2|+|Y^3|+|Y^4|)$ for $\al=1,2,3,4$. Thus condition \eqref{po3} is satisfied provided $\eps_1\leq \widetilde{C}^{-1}$, which completes the proof of the lemma.
\end{proof}

We prove now Proposition \ref{Lemmaa1}.

\begin{proof}[Proof of Proposition \ref{Lemmaa1}] In view of Lemma \ref{Car}, there are constants $\eps=\eps(A_0)\in(0,\eps_0)$ and $\widetilde{C}_\eps\geq 1$ such that, for any $x_0\in S_0$, $\lambda\geq\widetilde{C}_\eps$ and any $\phi\in C^\infty_0(B_{\eps^{10}}(x_0))$
\begin{equation}\label{Bar1}
\lambda \|e^{-\lambda f_\eps}\phi\|_{L^2}+\|e^{-\lambda f_\eps}|D^1\phi|\,\|_{L^2}\leq \widetilde{C}_\eps\lambda^{-1/2}\|e^{-\lambda f_\eps}\,\square_{\g}\phi\|_{L^2},
\end{equation}
where
\begin{equation}\label{Bar2}
f_\eps=\ln(\eps^{-1}(u_++\eps)(u_-+\eps)+\eps^{12}N^{x_0}).
\end{equation}
The constant $\eps$ will remain fixed in this proof. For simplicity of notation, we replace the constants $\widetilde{C}_\eps$ in \eqref{Bar1} with $\widetilde{C}$; since $\eps$ is fixed, these constants may depend only on the constant $A_0$. We will show that $\Ss\equiv 0$ in $B_{\ep^{40}}(x_0)\cap\mathbf{E}$ for any $x_0\in S_0$. This suffices to prove the proposition.

We fix $x_0\in S_0$ and, for $(j_1,\ldots,j_4)\in\{1,2,3,4\}^4$, we define using the vector-fields $\partial_\al$ induced by the coordinate chart $\Phi^{x_0}$
\begin{equation}\label{va5}
\phi_{(j_1\ldots j_4)}=\mathcal{S}(\partial_{j_1},\ldots ,\partial_{j_4}).
\end{equation}
The functions $\phi_{(j_1\ldots j_4)}:B_{\ep^{10}}(x_0)\to\mathbb{C}$ are smooth. Let $\eta:\mathbb{R}\to[0,1]$ denote a smooth function supported
in $[1/2,\infty)$ and equal to $1$ in $[3/4,\infty)$. For $\delta \in(0,1]$ we define,
\begin{equation}\label{pr2}
\begin{split}
\phi^{\delta,\eps}_{(j_1\ldots j_4)}&=\phi_{(j_1\ldots j_4)}\cdot \mathbf{1}_{\mathbf{E}}\cdot \eta(u_+u_-/ \delta)\cdot \big(1-\eta(N^{x_0}/\eps^{20})\big)\\
&=\phi_{(j_1\ldots j_4)}\cdot \widetilde{\eta}_{\delta,\eps}.
\end{split}
\end{equation}
Clearly, $\phi^{\delta,\eps}_{(j_1\ldots j_4)}\in C^\infty _0(B_{\eps^{10}}(x_0)\cap\mathbf{E})$. We would like to apply the inequality \eqref{Bar1} to the functions $\phi^{\delta,\eps}_{(j_1\ldots j_4)}$, and then let $\delta \to 0$ and $\lambda \to\infty$ (in this order).

Using the definition \eqref{pr2}, we have
\begin{equation*}
\square_\g\phi^{\delta,\eps}_{(j_1\ldots j_4)}=\widetilde{\eta}_{\delta,\eps}\cdot\square_\g\phi_{(j_1\ldots j_4)}+2\D_\al\phi_{(j_1\ldots j_4)}\cdot \D^\al \widetilde{\eta}_{\delta,\eps}+\phi_{(j_1\ldots j_4)}\cdot \square_\g\widetilde{\eta}_{\delta,\eps}.
\end{equation*}
Using the Carleman inequality \eqref{Bar1}, for any $(j_1,\ldots j_4)\in\{1,2,3,4\}^4$ we have
\begin{equation}\label{va10}
\begin{split}
&\lambda\cdot \|e^{-\lambda f_{\eps}}\cdot \widetilde{\eta}_{\delta,\eps}\phi_{(j_1\ldots j_4)} \|_{L^2}+\|e^{-\lambda f_{\eps}}\cdot \widetilde{\eta}_{\delta,\eps}|D^1\phi_{(j_1\ldots j_4)}|\, \|_{L^2}\\
&\leq \widetilde{C}\lambda ^{-1/2}\cdot \|e^{-\lambda f_{\eps}}\cdot \widetilde{\eta}_{\delta,\eps}\square_\g\phi_{(j_1\ldots j_4)}\|_{L^2}\\
&+\widetilde{C}\Big[\|e^{-\lambda f_{\eps}}\cdot \D_\al\phi_{(j_1\ldots j_4)}\D^\al \widetilde{\eta}_{\delta,\eps} \|_{L^2}+\|e^{-\lambda f_{\eps}}\cdot \phi_{(j_1\ldots j_4)} ( |\square_\g\widetilde{\eta}_{\delta,\eps}|+|D^1\widetilde{\eta}_{\delta,\eps}| )\|_{L^2}\Big],
\end{split}
\end{equation}
for any $\lambda\geq\widetilde{C}$. We estimate now $|\square_\g\phi_{(j_1\ldots j_4)}|$. Using Theorem \ref{wave} and the definition \eqref{s9.2}, in $B_{\eps^{10}}(x_0)$ we estimate pointwise
\begin{equation}\label{va11}
|\square_\g\phi_{(j_1\ldots j_4)}|\leq M\sum_{l_1,\ldots,l_4}\big(|D^1\phi_{(l_1\ldots l_4)}|+|\phi_{(l_1\ldots l_4)}|\big),
\end{equation}
for some large constant $M$. We add inequalities \eqref{va10} over $(j_1,\ldots,j_4)\in\{1,2,3,4\}^4$. The key observation is that, in view of \eqref{va11}, the first term in the right-hand side of \eqref{va10} can be absorbed into the left-hand side for $\lambda$ sufficiently large. Thus, for any $\lambda$ sufficiently large and $\delta \in(0,1]$,
\begin{equation}\label{va12}
\begin{split}
&\lambda\sum_{j_1,\ldots ,j_4}\|e^{-\lambda f_{\eps}}\cdot \widetilde{\eta}_{\delta,\eps}\phi_{(j_1\ldots j_4)} \|_{L^2}\\
&\leq \widetilde{C}\sum_{j_1,\ldots ,j_4}\Big[\|e^{-\lambda f_{\eps}}\cdot \D_\al\phi_{(j_1\ldots j_4)}\D^\al \widetilde{\eta}_{\delta,\eps} \|_{L^2}+\|e^{-\lambda f_{\eps}}\cdot \phi_{(j_1\ldots j_4)} ( |\square_\g\widetilde{\eta}_{\delta,\eps}|+|D^1\widetilde{\eta}_{\delta,\eps}| )\|_{L^2}\Big].
\end{split}
\end{equation}

We would like to let $\delta\to 0$ in \eqref{va12}. For this, we observe first that the functions $\D_\al\phi_{(j_1\ldots j_4)}\D^\al \widetilde{\eta}_{\delta,\eps}$ and $( |\square_\g\widetilde{\eta}_{\delta,\eps}|+|D^1\widetilde{\eta}_{\delta,\eps}| )$ vanish outside the set $\mathbf{A}_\delta\cup\widetilde{\mathbf{B}}_\eps$, where
\begin{equation*}
\begin{split}
&\mathbf{A}_\delta=\{x\in B_{\eps^{10}}(x_0)\cap\mathbf{E}:u_+(x)u_-(x)\in(\delta/2,\delta )\};\\
&\mathbf{B}_\eps=\{x\in B_{\eps^{10}}(x_0)\cap\mathbf{E}:N^{x_0}\in(\eps^{20}/2,\eps^{20})\}.
\end{split}
\end{equation*}
In addition, since $\phi_{(j_1\ldots j_4)}=0$ on $\mathbf{O}_{\ep_2}\cap[\delta(\II^{-}(\M^{(end)}))\cup \delta(\II^{+}(\M^{(end)}))]$ (see section \ref{horizon}), it follows from \eqref{boundedgeom2.2} and \eqref{boundedgeom2.3} that there are smooth functions ${\phi'}_{(j_1\ldots j_4)}:\mathbf{O}_{\ep_2}\to\mathbb{C}$ such that
\begin{equation}\label{va20}
\phi_{(j_1\ldots j_4)}=u_+u_-\cdot {\phi'}_{(j_1\ldots j_4)}\text{ in }\mathbf{O}_{\ep_2}.
\end{equation}

We show now that
\begin{equation}\label{va23}
\begin{split}
|\square_\g\widetilde{\eta}_{\delta,\eps}|+|D^1\widetilde{\eta}_{\delta,\eps}|&\leq \widetilde{C}(\mathbf{1}_{{\mathbf{B}}_\eps}+(1/ \delta)\mathbf{1}_{\mathbf{A}_\delta}).
\end{split}
\end{equation}
The  inequality for $|D^1\widetilde{\eta}_{\delta,\eps}|$ follows directly from the definition \eqref{pr2}. Also, using again the definition,
\begin{equation*}
|\D^\al\D_\al\widetilde{\eta}_{\delta,\eps}|\leq |\D^\al\D_\al(\mathbf{1}_{\mathbf{E}}\cdot \eta(u_+u_-/ \delta))|\cdot \big(1-\eta(N^{x_0}/\eps^{20})\big)+\widetilde{C}(\mathbf{1}_{{\mathbf{B}}_\eps}+(1/ \delta)\mathbf{1}_{\mathbf{A}_\delta}).
\end{equation*}
Thus, for \eqref{va23}, it suffices to prove that
\begin{equation}\label{va23.7}
\mathbf{1}_{\mathbf{E}\cap B_{\eps^{10}}(x_0)}\cdot |\D^\al\D_\al(\eta(u_+u_-/ \delta))|\leq \widetilde{C}/ \delta\cdot \mathbf{1}_{\mathbf{A}_\delta}.
\end{equation}
Since $u_+,u_-,\eta$ are smooth functions, for \eqref{va23.7} it suffices to prove that
\begin{equation}\label{va89}
\delta^{-2}|\D^\al(u_+u_-)\D_\al(u_+u_-)|\leq\widetilde{C}/ \delta\text{ in }\mathbf{A}_\delta, 
\end{equation}
which follows from \eqref{lubounds}.
 
We show now that 
\begin{equation}\label{va21}
|\D_\al\phi_{(j_1\ldots j_4)}\D^\al \widetilde{\eta}_{\delta,\eps}|\leq  \widetilde{C}_{\phi'}(\mathbf{1}_{{\mathbf{B}}_\eps}+\mathbf{1}_{\mathbf{A}_\delta}),
\end{equation}
where the constant $\widetilde{C}_{\phi'}$ depends on the smooth functions ${\phi'}_{(j_1\ldots j_4)}$ defined in \eqref{va20}. Using the formula \eqref{va20}, this follows easily from \eqref{va89}.

It follows from \eqref{va20}, \eqref{va23}, and \eqref{va21} that
\begin{equation*}
|\D_\al\phi_{(j_1\ldots j_4)}\D^\al \widetilde{\eta}_{\delta,\eps}|+|\phi_{j_1\ldots j_4}|( |\square_\g\widetilde{\eta}_{\delta,\eps}|+|D^1\widetilde{\eta}_{\delta,\eps}| )\leq\widetilde{C}_{\phi'}(\mathbf{1}_{{\mathbf{B}}_\eps}+\mathbf{1}_{\mathbf{A}_\delta}). 
\end{equation*}
Since $\lim_{\delta\to 0}\|\mathbf{1}_{\mathbf{A}_\delta}\|_{L^2}=0$, we can let $\delta\to
0$ in \eqref{va12} to conclude that
\begin{equation}\label{va13}
\lambda\sum_{j_1,\ldots ,j_4}\|e^{-\lambda f_{\eps}}\cdot \mathbf{1}_{B_{\eps^{10}/2}(x_0)\cap\mathbf{E}}\cdot \phi_{(j_1\ldots j_4)} \|_{L^2}\leq \widetilde{C}_{\phi'}\|e^{-\lambda f_{\eps}}\cdot \mathbf{1}_{{\mathbf{B}}_\eps}\|_{L^2}
\end{equation}
for any $\lambda$ sufficiently large. Finally, using the definition \eqref{Bar2}, we observe that
\begin{equation*}
\inf_{B_{\eps^{40}}(x_0)\cap\mathbf{E}}\,e^{-\lambda f_{\eps}}\geq e^{-\lambda \ln[\eps+\eps^{32}/2]}\geq \sup_{{\mathbf{B}}_\eps}\,e^{-\lambda f_{\eps}}.
\end{equation*}
It follows from \eqref{va13} that
\begin{equation*}
\lambda\sum_{j_1,\ldots ,j_4}\|\mathbf{1}_{B_{\eps^{40}}(x_0)\cap\E}\cdot \phi_{(j_1\ldots j_4)} \|_{L^2}\leq \widetilde{C}_{\phi'}\|\mathbf{1}_{{\mathbf{B}}_\eps}\|_{L^2}
\end{equation*}
for any $\lambda$ sufficiently large. We let $\lambda\to\infty$ to conclude that $\phi_{(j_1\ldots j_4)}=0$ in $B_{\eps^{40}}(x_0)\cap\E$, which completes the proof of the proposition.
\end{proof}

\section{Consequences of the vanishing of $\SS$}\label{section2}

We assume in this section that $\mathbf{N}\subseteq \widetilde{\mathbf{M}}$ is a open set, $S_0\subseteq\mathbf{N}$, $\mathbf{N}\cap\mathbf{E}$ is connected, and
\begin{equation}\label{big0}
\begin{cases}
&1-\sigma\neq 0\text{ in }\mathbf{N};\\
&\Ss_{\al\be\mu\nu}=\RR_{\al\be\mu\nu}+6(1-\sigma)^{-1}\big(\FF_{\al\be}\FF_{\mu\nu}-\frac{1}{3}\FF^2\II_{\al\be\mu\nu}\big)=0\text{ in }\mathbf{N}\cap\mathbf{E}.
\end{cases}
\end{equation}
It follows from the assumption \eqref{Main-Cond1}, and the identities \eqref{Ernst-important} and \eqref{big0} (which give $\D_\rho(\FF^2(1-\sigma)^{-4})=0)$ in $\mathbf{N}\cap\mathbf{E}$) that
\begin{equation}\label{big1}
-4M^2\FF^2=(1-\sigma)^4\text{ in }\mathbf{N}\cap\mathbf{E}.
\end{equation}
We define the smooth function $P=y+iz:\mathbf{N}\to\mathbb{C}$,
\begin{equation}\label{se90}
P=y+iz=(1-\sigma)^{-1}.
\end{equation}

Since $-\FF^2/4=(4MP^2)^{-2}\neq 0$ (see \eqref{big1}), there are null vector-fields $l,\ul$, locally around every point in $\mathbf{N}$, such that
\begin{equation}\label{big2}
\FF_{\al\be}l^\be=(4MP^2)^{-1}l_\a,\,\,\FF_{\al\be}\ul^\b=-(-4MP^2)^{-1}\ul_\a,\,\,\text{ and }l^\a\ul_\a=-1\text{ in }\mathbf{N}\cap\mathbf{E}.
\end{equation}
We fix a complex-valued null vector-field $m$ on $\mathbf{N}$ such that $(m,\overline{m},\ul,l)=(e_1,e_2,e_3,e_4)$ is a complex null tetrad, see the definitions in subsection \ref{complexnull}. We may also assume that $(m,\overline{m},\ul,l)$ has positive orientation, i.e.
\begin{equation*}
\in_{\al\be\mu\nu}m^\al\overline{m}^\be\ul^\mu l^\nu=i.
\end{equation*}

We prove now some identities. Most of these identities, with the exception of Proposition \ref{uniformy} and the computation of the Hessian of $y$ in Lemma \ref{coefficients}, were derived by Mars \cite{Ma1}; for the sake of completeness we rederive them in our notation. 
It follows from \eqref{big2} and \eqref{big1} that, in $\mathbf{N}\cap\mathbf{E}$,
\begin{equation}\label{se5}
\FF_{\a\b}=\frac{1}{4MP^2}\big(-l_\a\ul_\b+l_\b\ul_\a-i\in_{\al\be\mu\nu}l^\mu\ul^\nu\big).
\end{equation}

Using \eqref{se5}, we compute easily
\begin{equation}\label{se6'}
\FF_{41}=\FF_{42}=\FF_{31}=\FF_{32}=0\text{ and }\FF_{43}=\FF_{21}=1/(4MP^2).
\end{equation}
Using \eqref{big0} and \eqref{se6'} we compute
\begin{equation}\label{se6}
\begin{split}
&\RR_{4141}=\RR_{4242}=0\quad\text{ thus }\Psi_{(2)}(R)=0;\\
&\RR_{3131}=\RR_{3232}=0\quad\text{ thus }\underline{\Psi}_{(2)}(R)=0;\\
&\RR_{1434}=\RR_{2434}=0\quad\text{ thus }\Psi_{(1)}(R)=0;\\
&\RR_{1343}=\RR_{2343}=0\quad\text{ thus }\underline{\Psi}_{(1)}(R)=0;\\
&\RR_{2314}=\frac{1}{4M^2P^3},\,\RR_{1324}=0\quad\text{ thus }\Psi_{(0)}(R)=\frac{1}{8M^2P^3}.\\
\end{split}
\end{equation}
We use now the first $4$ Bianchi identities \eqref{Bi1}--\eqref{Bi4} to conclude that
\begin{equation}\label{se7}
\xi=\xib=\va=\vab=0\text{ in }\mathbf{N}\cap\mathbf{E}.
\end{equation}
The remaining $4$ Bianchi identities, \eqref{Bi5}--\eqref{Bi8} give
\begin{equation}\label{se8}
DP=\theta P,\,\,\,\underline{D} P=\overline{\thb} P,\,\,\,\delta P=\eta P,\,\,\,\overline{\delta}P=\overline{\etab} P.
\end{equation}

We analyze now the functions $y$ and $z$. By contracting \eqref{se5} with $2\T^\a$ and using $2\T^\a\FF_{\a\b}=\sigma_\be=\D_\be\sigma$ we derive
\begin{equation}\label{se9}
\D_\be y=\frac{1}{2M}\big[-(\T^\al l_\al)\ul_\b+(\T^\al \ul_\al)l_\be\big]\text{ and }\D_\b z=\frac{-1}{2M}\in_{\al\be\mu\nu}\T^\a l^\mu\ul^\nu.
\end{equation}
In particular,
\begin{equation}\label{se9.6}
\delta y=\overline{\delta}y=Dz=\underline{D} z=0.
\end{equation}
Using \eqref{se8} it follows that
\begin{equation}\label{se9.7}
Dy=\theta P,\,\,\underline{D} y=\overline{\thb} P,\,\,\delta z=-i\eta P,\,\,\overline{\delta}z=-i\overline{\etab} P.
\end{equation}
In particular $\theta P=\overline{\theta P}$, $\overline{\thb} P=\thb\overline{P}$, $-\eta P=\etab\overline{P}$, and, using again \eqref{se9},
\begin{equation}\label{se9.3}
\T^\al l_\al=2M\theta P,\,\,\T^\al\Ll_\al=-2M\overline{\thb} P.
\end{equation}
Using \eqref{se9.6} and \eqref{se9.7} we rewrite \eqref{se9} in the form
\begin{equation}\label{se9.9}
\D_\b y=-\overline{\thb} Pl_\b-\theta P\ul_\b,\,\,\,\,\D_\b z=-i\overline{\etab} P m_\b-i\eta P\overline{m}_\b.
\end{equation}

A direct computation using the definition of $P$ shows that
\begin{equation*}
\D_\a P\D^\a P=\frac{\D_\a\sigma \D^\a\sigma}{(1-\sigma)^4}=-\frac{\T^\al\T_\al}{4M^2}.
\end{equation*}
The real part of this identity and $-\T^\al\T_\al=\Re\sigma$ give
\begin{equation}\label{se10}
\D_\a y\D^\a y-\D_\a z\D^\a z=\frac{-\T^\al\T_\al}{4M^2}=\frac{1}{4M^2}\Big(1-\frac{y}{y^2+z^2}\Big).
\end{equation}
Using \eqref{se9.9} this gives
\begin{equation}\label{se10.1}
8M^2(\eta\overline{\etab}-\th\overline{\thb})P^3\overline{P}=y^2-y+z^2.
\end{equation}

\begin{lemma}\label{MarsMain}
There is a constant $B\in[0,\infty)$ such that
\begin{equation}\label{se11}
\D_\a z\D^\a z=\frac{B-z^2}{4M^2(y^2+z^2)}\text{ in }\mathbf{N}\cap\mathbf{E}.
\end{equation}
In addition $z^2\leq  B$ in $\mathbf{N}\cap\mathbf{E}$. 
\end{lemma}

\begin{proof}[Proof of Lemma \ref{MarsMain}]
For \eqref{se11} it sufficces to prove that
\begin{equation}\label{se21}
4M^2P\overline{P}\cdot \D_\a z\D^\a z+z^2=B.
\end{equation}
Let $Z=4M^2P\overline{P}\cdot \D_\a z\D^\a z$. To show $\underline{D}(Z+z^2)=0$ we use the formula $Z=8M^2P^2\overline{P}^2\etab\,\overline{\etab}$ (which follows from \eqref{se9.9} and $-\eta P=\etab\overline{P}$), the identities \eqref{se8}, $\overline{\thb} P=\thb\overline{P}$, and the Ricci equation (see \eqref{fo23'}, \eqref{se6}, and \eqref{se7})
\begin{equation*}
\underline{D}\etab=\thb(\eta-\etab)-\Gamma_{123}\etab.
\end{equation*}
Indeed,
\begin{equation*}
\begin{split}
\underline{D}(Z+z^2)&=8M^2P^2\overline{P}^2\etab\,\overline{\etab}\Big( \frac{2\underline{D}P}{P}+\frac{2\underline{D}\overline{P}}{\overline{P}}+\frac{\underline{D}\etab}{\etab}+\frac{\underline{D}\overline{\etab}}{\overline{\etab}}\Big)\\
&=8M^2P^2\overline{P}^2\etab\,\overline{\etab}[2\overline{\thb}+2\thb-\thb(\overline{P}/P+1)-\Gamma_{123}-\overline{\thb}(P/\overline{P}+1)-\Gamma_{213}]\\
&=0.
\end{split}
\end{equation*}

To show $D(Z+z^2)=0$ we use the formula $Z=8M^2P^2\overline{P}^2\eta \overline{\eta}$ (which follows from \eqref{se9.9} and $-\eta P=\etab\overline{P}$), the identities \eqref{se8}, $\theta P=\overline{\theta}\overline{P}$, and the Ricci equation (see \eqref{fo23}, \eqref{se6}, and \eqref{se7})
\begin{equation*}
D\eta=\th(\etab-\eta)-\Gamma_{124}\eta.
\end{equation*}
Indeed,
\begin{equation*}
\begin{split}
D(Z+z^2)&=8M^2P^2\overline{P}^2\eta\,\overline{\eta}\Big( \frac{2DP}{P}+\frac{2D\overline{P}}{\overline{P}}+\frac{D\eta}{\eta}+\frac{D\overline{\eta}}{\overline{\eta}}\Big)\\
&=8M^2P^2\overline{P}^2\eta\,\overline{\eta}[2\th+2\overline{\th}-\th(P/\overline{P}+1)-\Gamma_{124}-\overline{\th}(\overline{P}/P+1)-\Gamma_{214}]\\
&=0.
\end{split}
\end{equation*}

Finally, to show that $\delta(Z+z^2)=0$ we use the formula
\begin{equation*}
Z+z^2=-8M^2 P^2\overline{P}^2\th\thb-y^2+y,
\end{equation*}
which follows from \eqref{se10.1} and $\overline{\thb} P=\thb\overline{P}$, the identities \eqref{se8}, $\th P=\overline{\th }\overline{P}$, and the Ricci equations (see \eqref{fo26}, \eqref{fo26'}, \eqref{se6}, and \eqref{se7})
\begin{equation*}
\begin{cases}
&\delta\th=-\ze\th-\eta(\theta-\overline{\th});\\
&\delta\thb=\ze\thb-\etab(\thb-\overline{\thb}).
\end{cases}
\end{equation*}
Indeed,
\begin{equation*}
\begin{split}
\delta(Z+z^2)&=-8M^2P^2\overline{P}^2\th\thb \Big(\frac{\delta\th}{\th}+\frac{\delta\thb}{\thb}+\frac{2\delta P}{P}+\frac{2\delta \overline{P}}{\overline{P}}\Big)\\
&=-8M^2P^2\overline{P}^2\th\thb[-\ze-\eta(1-P/\overline{P})+\ze-\etab(1-\overline{P}/P)+2\eta+2\etab]\\
&=0.
\end{split}
\end{equation*}
This completes the proof of \eqref{se21}.
\end{proof} 

It follows from \eqref{se11} and \eqref{se10} that
\begin{equation}\label{se30}
\D_\a y\D^\a y=\frac{y^2-y+B}{4M^2(y^2+z^2)}.
\end{equation}
Using \eqref{se9.3} and \eqref{se9.9}, it follows that
\begin{equation}\label{se31}
-\th\overline{\thb} P^2=\frac{y^2-y+B}{8M^2(y^2+z^2)}=\frac{(\T^\al l_\al)\cdot(\T^\al\ul_\al)}{4M^2}.
\end{equation}
We express also the vector $\T$ in the complex null tetrad $(\Mm,\Omm,\ul,l)$. Using \eqref{se5}, and \eqref{se9},
\begin{equation}\label{se70}
\T_\a=(\FF^2/4)^{-1}{\FF_\a}^\mu\T^\b\FF_{\b\mu}=-(\T^\b\ul_\b)l_\a-(\T^\be l_\be)\ul_\a-2M\in_{\al\be\mu\nu}\D^\be zl^\mu \ul^\nu. 
\end{equation}

We prove now a uniform bound on the gradient of the function $y$.

\begin{proposition}\label{uniformy}
There is a constant $\widetilde{C}=\widetilde{C}(A_0)$ that depends only on $A_0$ such that
\begin{equation}\label{unifybound}
|D^1y|\leq\widetilde{C}\text{ in }\mathbf{N}.
\end{equation}
\end{proposition}

\begin{proof}[Proof of Proposition \ref{uniformy}] For $p\in\Phi^{x_0}(B_1)$, $x_0\in\Sigma_0$, the gradient $|D^1y|$ is defined using the coordinate chart $\Phi^{x_0}$, i.e.
\begin{equation*}
|D^1y|(p)=\sum_{j=1}^{4}|\partial_j(y)(p)|.
\end{equation*}
In view of the definition $y=\Re[(1-\sigma)^{-1}]$ and the smoothness of $\sigma$, the bound \eqref{unifybound} is clear if $|1-\sigma(p)|\geq 1/4$. Assume that $|1-\sigma(p)|<1/4$. Since
\begin{equation*}
\Re(1-\sigma)=1+\g(\T,\T),
\end{equation*}
it follows that $\g(\T,\T)(p)<-3/4$. In particular, $p\in \mathbf{N}\cap\mathbf{E}$. We define the vector-field,
\begin{equation}
Y=\g^{\al\be}\partial_\al y\partial_\be. \label{vect-Y}
\end{equation}
In view of \eqref{se30} and $T(y)=0$, we have
\begin{equation*}
|\g(Y,Y)|=\Big|\frac{y^2-y+B}{4M^2(y^2+z^2)}\Big|\leq\widetilde{C}\text{ and }\g(\T,Y)=\T(y)=0\text{ at } p.
\end{equation*}
Since $\g(\T,\T)<-3/4$ it follows that $Y_p$ is a space-like vector with norm (as induced by the coordinate chart $\Phi^{x_0}$) dominated by $\widetilde{C}$. The bound \eqref{unifybound} follows since $\partial_jy=\g(Y,\partial_j)$.
\end{proof}

\subsection{The connection coefficients and the Hessian of $y$}

Assume now that $\mathbf{N}'$ is a subset of $\mathbf{N}\cap\mathbf{E}$ with the property that
\begin{equation*}
y^2-y+B>0\text{ in }\mathbf{N}'.
\end{equation*}
Using \eqref{se31}, we can normalize the vector $l$ such that
\begin{equation}\label{stuff2}
\T^\al l_\al=2M\quad \text{ in }\,\,\mathbf{N}'.
\end{equation}
Thus, using \eqref{se9.3} and \eqref{se31}, we compute
\begin{equation}\label{stuff3}
\theta=1/P\text{ and }\thb=-W/\overline{P}=-\frac{1}{\overline{P}}\cdot \frac{y^2-y+B}{8M^2(y^2+z^2)}\quad \text{ in }\,\, \mathbf{N}'.
\end{equation}
Using the null structure equation \eqref{fo22} (see also \eqref{se7})
\begin{equation*}
D\th=-\th^2-\om\th,
\end{equation*}
together with \eqref{se8} and \eqref{stuff3}, we compute
\begin{equation}\label{stuff4}
\om=0\quad \text{ in }\,\,\mathbf{N}'.
\end{equation}
Using the null structure equation \eqref{fo22'} (see also \eqref{se7})
\begin{equation*}
\underline{D}\,\thb=-\thb^2-\omb\,\thb,
\end{equation*}
together with \eqref{se9.6}, \eqref{se9.7}, and \eqref{stuff3}, we compute
\begin{equation}\label{stuff5}
\omb=\frac{y^2-z^2-2y(B-z^2)}{8M^2(y^2+z^2)^2}\quad \text{ in }\,\,\mathbf{N}'.
\end{equation}
We can express $\omb$ in the form,
\bea
\omb&=& HW,\qquad H=\frac{y^2-z^2-2y(B-z^2)}{(y^2-y+B)(y^2+z^2)}.\label{stuff50}
\eea
Using the null structure equation \eqref{fo26} (see also \eqref{se7} and \eqref{se6})
\begin{equation*}
\delta\theta=-\ze\theta-\eta(\th-\overline{\th}),
\end{equation*}
together with \eqref{se8} and \eqref{stuff3}, we compute
\begin{equation}\label{stuff6}
\ze=\frac{\eta P}{\overline{P}}=-\etab\quad \text{ in }\,\,\mathbf{N}'.
\end{equation}
Using \eqref{se10.1} and \eqref{stuff3},
\begin{equation}\label{stuff7}
|\ze|^2=\etab\,\overline{\etab}=\frac{B-z^2}{8M^2(y^2+z^2)^2}\text{ and }\eta=-\frac{\etab\overline{P}}{P}\quad \text{ in }\,\,\mathbf{N}'.
\end{equation}
Finally, using \eqref{se9.9}, we rewrite \eqref{se70} in the form
\begin{equation}\label{stuff8}
\T=-2M(Wl+\ul-\overline{\ze}Pm-\ze\overline{P}\overline{m})\quad \text{ in }\,\, \mathbf{N}'.
\end{equation}
We summarize these computations in the first part of the  following lemma.

\begin{lemma}\label{coefficients} Let $\N$ be the set 
defined by \eqref{big0}  and $\mathbf{N}'$ the subset of 
$\N\cap\E$ for which $y^2-y+B>0$, with $B$ the constant
 of lemma \ref{MarsMain}.
In $\mathbf{N}'$ we have, with $P=y+iz=(1-\si)^{-1}$,
\begin{equation*}
\begin{split}
&\qquad \qquad\qquad\qquad  \xi=\xib=\va=\vab=\om=0,\quad\\
&\omb=HW=\frac{y^2-z^2-2y(B-z^2)}{8M^2(y^2+z^2)^2},\quad 
 H=\frac{y^2-z^2-2y(B-z^2)}{(y^2-y+B)(y^2+z^2)},\\
&W=\frac{y^2-y+B}{8M^2(y^2+z^2)}>0,\qquad\qquad\qquad  |z|^2\le B, \\
&\th=1/P,\quad\thb=-W/ \overline{P},\quad|\ze|^2=\frac{B-z^2}{8M^2(y^2+z^2)^2},\quad\eta=\ze\frac{\overline{P}}{P},\quad\etab =-\ze,\\
&\delta y=\overline{\delta}y=Dz=\underline{D}z=0,\quad Dy=1,\quad \underline{D}y=-W,\quad \D_\al y\D^\al y=2W.
\end{split}
\end{equation*}
We also have, for the Hessian of  $y$, 
\bea
\label{hessian-y}
\begin{cases}
&(\D^2 y)_{44}=(\D^2 y)_{33}=0, \qquad\qquad (\D^2 y)_{43}=(\D^2 y)_{34}=-WH\\
&(\D^2 y)_{41}=(\D^2 y)_{14}=\ze,\qquad \qquad (\D^2 y)_{42}=(\D^2 y)_{24}=\overline{\ze} \\
&(\D^2 y)_{31}=(\D^2 y)_{13}=\eta W, \qquad\quad  (\D^2 y)_{32}= (\D^2 y)_{23}=\overline{\eta}  W\\
&(\D^2 y)_{12}=(\D^2 y)_{21} = W\frac{2y}{y^2+z^2},\quad\,\, (\D^2 y)_{11}=(\D^2 y)_{22}=0.
\end{cases}
\eea
\end{lemma}
\begin{proof}
It only remains to prove formulas in \eqref{hessian-y}. These formulas follow easily using $(\D^2 y)_{\al\be}=e_\al(e_\be y)-{\Gamma^\mu}_{\be\al}e_\mu(y)$, the first part of the lemma, and the table \eqref{def1}.
\end{proof}

\section{The main bootstrap argument}\label{lastsection}

In this section we show that
\begin{equation}\label{mainclaim}
1-\sigma\neq 0 \,\text{ and }\, \SS=0\,\,\, \text{ on }\,\Sigma_0\cap\mathbf{E}.\end{equation}
In view of our  assumption {\bf{AF}}, this suffices to show that $\SS=0$ in $\mathbf{E}$. Our  Main Theorem is  then consequence of the main result  of Mars  in \cite{Ma1}.

We show first that the function $y$ is constant on $\HH^+\cup\HH^-$ and increases in $\mathbf{E}$.

\begin{lemma}\label{yhorizon}
There is a constant $y_{S_0}\in(1/2,1]$ such that
\begin{equation}\label{yho1}
y=y_{S_0}\text{ on }\HH^+\cup\HH^-.
\end{equation}
In addition $B\in[0,1/4)$, where $B$ is the constant in Lemma \ref{MarsMain}. Finally, for sufficiently small $\ep=\ep(A_0)>0$,
\begin{equation}\label{yho2}
y>y_{S_0}+\widetilde{C}^{-1}u_+u_-\text{ on }\O_\ep\cap\mathbf{E},
\end{equation}
where $\mathbf{O}_\ep$ are the open sets defined in section \ref{preliminaries}, and $\widetilde{C}=\widetilde{C}(A_0)>0$.
\end{lemma}

\begin{proof}[Proof of Lemma \ref{yhorizon}] Let $\mathbf{N}=\mathbf{O}_{r_1}$ denote the set constructed in Proposition \ref{Lemmaa1}. Since $\SS=0$ in $\mathbf{N}$, we can apply the computations of the previous section. It follows from \eqref{avanish} that if $\widetilde{l}$ is tangent to the null generators of $\HH^+$ then $\FF_{\al\be}\widetilde{l}^\be=C\widetilde{l}_\al$ for some scalar $C$. Thus $\widetilde{l}$ is parallel to either $l$ or $\ul$ on $\HH^+$. Similarly, the null generator of $\HH^-$ is also parallel to either $l$ or $\ul$ on $\HH^-$. Thus the vector $m$ is tangent to the bifurcate sphere $S_0$. Using $\delta y=0$, see \eqref{se9.6}, it follows that $y$ is constant on $S_0$. Using \eqref{se9.7} and Proposition \ref{NE-hor} it follows that $y$ is constant on $\HH^+\cup\HH^-$, which gives \eqref{yho1}. Also, using \eqref{se31} on $S_0$ and the fact that $\T$ is tangent to $S_0$, it follows that
\begin{equation*}
y_{S_0}^2-y_{S_0}+B=0.
\end{equation*}
Since $B\in[0,\infty)$ and $y_{S_0}>1/2$ (using  assumption \eqref{Main-Cond2}), it follows that $B\in[0,1/4)$ and
\begin{equation*}
y_{S_0}=\frac{1+\sqrt{1-4B}}{2}\in(1/2,1].
\end{equation*}

To prove \eqref{yho2} we consider  the open sets $\O_\ep$ and the functions $u_{\pm}:\mathbf{O}_\ep\to\mathbb{R}$ defined in section \ref{preliminaries}.  It follows from \eqref{yho1} combined with \eqref{boundedgeom2.2}, \eqref{boundedgeom2.3} that
\begin{equation}\label{yho4}
y=y_{S_0}+u_+u_-\cdot y',
\end{equation} 
for some smooth function $y':\mathbf{O}_\ep\to\mathbb{R}$, with $|D^1y'|\leq\widetilde{C}$. The identities $P=(1-\sigma)^{-1}$, $\D^\mu\D_\mu\sigma=-\FF^2$, $\D_\mu\sigma\D^\mu\sigma=\T^\al\T_\al\cdot \FF^2=-\Re\sigma\cdot \FF^2$ (see \eqref{Ernst1}), and $\FF^2=-(1-\sigma)^4/(4M^2)$ (see \eqref{big1}) show that,
\beaa
\D^\mu\D_\mu P&=&(1-\si)^{-2}\D^\mu\D_\mu\si+2(1-\si)^{-3}\D^\mu\si \D_\mu \si\\
&=&\frac{1}{4M^2}(1-\si)(1+\bar\si)=\frac{2\overline{P}-1}{4M^2 P\overline{P}}.
\eeaa
Therefore,
\begin{equation}\label{yho5}
\D^\mu\D_\mu y=\frac{2y-1}{4M^2(y^2+z^2)}.
\end{equation}
We substitute $y=y_{S_0}+u_+u_-\cdot y'$ (see \eqref{yho4}) and evaluate on $S_0$
\begin{equation*}
\frac{2y_{S_0}-1}{4M^2(y_{S_0}^2+z^2)}=\D^\mu\D_\mu(y_{S_0}+u_+u_-\cdot y')=2 \D^\mu(u_+)\D_\mu(u_-)\cdot y'=4y'.
\end{equation*}
Since $y_{S_0}>1/2+\widetilde{C}^{-1}$ it follows that $y'>\widetilde{C}^{-1}$ on $S_0$. Thus, for $\ep\in(0,r_1)$ sufficiently small,
\begin{equation*}
y>y_{S_0}+\widetilde{C}^{-1}\, u_+ u_-\quad \text{ in }\O_\ep\cap\mathbf{E},
\end{equation*}
as desired.
\end{proof}

We define the set
\begin{equation*}
\Sigma'_0=\{x\in\Sigma_0\cap\E:\sigma(x)\neq 1\}.
\end{equation*}
Clearly, $\Sigma'_0$ is an open subset of $\Sigma_0\cap\E$ which contains a neighborhood of $S_0$ in $\Sigma_0\cap\E$. We define the function (which agrees with the function $y$ defined earlier on open sets)
\begin{equation*}
y:\Sigma'_0\to\mathbb{R},\quad y(x)=\Re[(1-\sigma)^{-1}].
\end{equation*}
For any $R>y_{S_0}$ let $\VV_R=\{x\in\Sigma'_0:y(x)<R\}$ and $\UU_R$ the unique connected component of $\VV_R$ whose closure in $\Sigma_0$ contains $S_0$ (this unique connected component exists since $y(x)=y_{S_0}<R$ on $S_0$). We prove now the first step in our bootstrap argument.

\begin{proposition}\label{bootstrap1}
There is a real number $R_1\geq y_{S_0}+\widetilde{C}^{-1}$, for some constant $\widetilde{C}=\widetilde{C}(A_0)>0$, such that $\SS=0$ in $\UU_{R_1}$.
\end{proposition}

\begin{proof}[Proof of Proposition \ref{bootstrap1}] 
With $\ep$ as in Lemma \ref{yhorizon}, it follows from Proposition \ref{Lemmaa1} that $\SS=0$ in $\mathbf{O}_\ep\cap\mathbf{E}$. Also, since $u_+/u_-+u_-/u_+\leq A_0$ in $\Sigma_0\cap\E\cap\O_\ep$, it follows from \eqref{yho2} that
\begin{equation*}
y-y_{S_0}\in[\widetilde{C}^{-1}(u_+^2+u_-^2),\widetilde{C}(u_+^2+u_-^2)]\text{ in }\Sigma_0\cap\E\cap\O_\ep.
\end{equation*}
Thus, for $R_1$ sufficiently close to $y_{S_0}$, the set $\UU_{R_1}$ is included in $\mathbf{O}_\ep$, and the proposition follows.
\end{proof}

With $R_1$ as in Proposition \ref{bootstrap1}, the main result in this section is the following:

\begin{proposition}\label{bootstrap2}
For any $R_2\geq R_1$ we have $\SS=0$ in $\UU_{R_2}$.
\end{proposition}

The proof of Proposition \ref{bootstrap2}, which  will be completed in subsection \ref{subsectionlast}, is done by induction.
In view of Proposition \ref{bootstrap1}, we may assume that the claims in Proposition \ref{bootstrap2} hold for some value $R_2\geq R_1$.
We therefore make the following induction hypothesis:

\medskip
{\bf Induction  Hypothesis}.\quad  For a fixed $R_2\ge R_1$  the tensor  $\Ss$ vanishes on the set 
$\UU_{R_2}$, which is the unique connected component of 
the set $\VV_{R_2}=\{x\in \Si_0\cap\E: y(x)<R_2,\,\, \si(x)\neq 1\}$
 whose closure in $\Si_0$ contains the bifurcate sphere  $S_0$.
\medskip

To complete the proof of the proposition we have to advance  these claims for $R'_2=R_2+r'$, where $r'>0$ depends only on the constants $A_0$, $\widetilde{A}_{\widetilde{C}^{-1}}$  (here  $\widetilde{A}_{\widetilde{C}^{-1}}=\widetilde{A}_\ep$ with $\ep=\widetilde{C}^{-1}$,  see \eqref{nontangq} for the definition of  $\widetilde{A}_\ep$), and $R_2$ (as before, the constants $\widetilde{C}$ may depend only on $A_0$). In the rest of this section we let $\widetilde{C}_{R_2}$ denote various constants in $[1,\infty)$ that may depend only on $A_0$, $\widetilde{A}_{\widetilde{C}^{-1}}$, and $R_2$. It is important that such constants do not depend on other parameters, such as the point $x_0\in\delta_{\Sigma_0\cap \mathbf{E}}(\UU_{R_2})$ chosen below.

Assume $x_0\in\delta_{\Sigma_0\cap \mathbf{E}}(\UU_{R_2})$ is a point on the boundary of $\UU_{R_2}$ in $\Sigma_0\cap \mathbf{E}$. Clearly,
\begin{equation*}
y(x_0)=R_2.
\end{equation*}
Thus $|1-\sigma(x_0)|=(R_2^2+z(x_0)^2)^{-1/2}$. Since $1-\sigma$ is a smooth function on $\widetilde{\mathbf{M}}$ and $z(x_0)^2\leq B<1/4$ (see Lemma \ref{coefficients}), there is $r'_2=r'_2(A_0,R_2)>0$ such that $|1-\sigma(x)|\in(1/(2R_2),2/R_2)$ in $B_{r'_2}(x_0)$. Thus the function
\begin{equation*}
y:B_{r'_2}(x_0)\to\mathbb{R},\quad y(x)=\Re[(1-\sigma(x))^{-1}],
\end{equation*}
is well defined; observe that, with $\partial_j$ defined according to the coordinate charts defined  in section 2.2,   
\begin{equation}\label{alexnew10}
\sup_{x\in B_{r'_2(x_0)}}\big(|y(x)|+|D^1y(x)|+\ldots+|D^4y(x)|\big)\leq\widetilde{C}_{R_2}.
\end{equation}
By choosing $r'_2$ sufficiently small it follows from $y(x_0)=R_2$ and \eqref{alexnew10} that
\begin{equation}\label{alexnew11}
y(x)\in((y_{S_0}+R_1)/2,2R_2)\,\, \text{ for any }\,\,x\in B_{r'_2}(x_0).
\end{equation}
In view of \eqref{nontangq} there is $\delta_2>\widetilde{C}_{R_2}^{-1}$ small\footnote{The constants $r'_2$ and $\delta_2$ are fixed in this paragraph such that $r'_2,\delta_2\ll 1$. We later fix the constants $r_2\ll \min(r'_2,\delta_2)$ (Lemma \ref{constr_2}), $r_3\ll  r_2$ (Proposition \ref{vanishingS3}), and $r'\ll  r_3$ (proof of Proposition \ref{bootstrap2}). All of these constants are bounded from below by some constant $\widetilde{C}_{R_2}^{-1}$} such that the set $(-\delta_2,\delta_2)\times (B_{r'_2}(x_0)\cap\Sigma_0)$ is diffeomorphic to the set $\cup_{|t|<\delta_2}\Phi_t(B_{r'_2}(x_0)\cap\Sigma_0)$. We  let $Q:\cup_{|t|<\delta_2}\Phi_t(B_{r'_2}(x_0)\cap\Sigma_0)\to B_{r'_2}(x_0)\cap\Sigma_0$ denote the induced  smooth projection which takes every  point $\Phi_t(x)$ into $x$.  

We now  define the connected open set of  $\widetilde{\mathbf{M}}$,  which we denote by $\N_{R_2}$,
\begin{equation}
\mathbf{N}_{R_2}=\text{ connected component of }\big[\big(\cup_{t\in\mathbb{R}}\Phi_t(\UU_{R_2})\big)\cup\O_{r_1}\big]\cap\widetilde{\mathbf{M}}\text{ containing }\UU_{R_2}, \label{NR2}
\end{equation}
where $r_1$ is as in Proposition \ref{Lemmaa1}. Since $\T$ is a Killing vector-field, $\LL_\T\SS=0$ in $\widetilde{\mathbf{M}}$ and $\T(1-\sigma)=0$. In view of our induction hypothesis  $\SS=0$ in $\UU_{R_2}$ and $\T$ does not vanish in $\mathbf{E}$;  it follows that
\begin{equation*}
1-\sigma\neq 0\text{ in }\mathbf{N}_{R_2}\,\, \text{ and }\,\, \SS=0\text{ in }\mathbf{N}_{R_2}\cap\mathbf{E}.
\end{equation*}
Thus the computations in section \ref{section2} can be applied in the open set $\mathbf{N}_{R_2}$.

\begin{lemma}\label{constr_2}
With $x_0\in\delta_{\Sigma_0\cap \mathbf{E}}(\UU_{R_2})$ as before, there is $r_2\in(0,r'_2]$ such that 
\begin{equation}\label{bingo}
\{x\in B_{r_2}(x_0):y(x)<R_2\}\subseteq \cup_{|t|<\delta_2}\Phi_t(\UU_{R_2}).
\end{equation}
\end{lemma}

\begin{proof}[Proof of Lemma \ref{constr_2}] In view of  \eqref{se30},
\begin{equation*}
\D^\al y\D_\al y=\frac{y^2-y+B}{4M^2(y^2+z^2)}\text{ in }\mathbf{N}_{R_2}.
\end{equation*}
Thus, if $r_2''\leq\widetilde{C}_{R_2}^{-1}$ is sufficiently small then $\D^\al y\D_\al y\geq\widetilde{C}_{R_2}^{-1}$ in $B_{r''_2}(x_0)$. It follows that there exists  $r_2=r_2(A_0,\widetilde{A}_{\widetilde{C}^{-1}},R_2)>0$ and an open set  $B'$, 
$ B_{r_2}(x_0)\subseteq  B'\subseteq B_{r''_2}(x_0)$,  such that the set $\{x\in B':y(x)<R_2\}$ is connected. Let $Q:B' \to B_{r'_2}(x_0)\cap\Sigma_0$ denote the projection defined above. The set $Q(\{x\in B':y(x)<R_2\})\subseteq B_{r'_2}(x_0)\cap\Sigma_0$ is connected and contains the set $\{x\in B'\cap\Sigma_0:y(x)<R_2\}$. Since $y(Q(x))=y(x)$, it follows from the definition of $\UU_{R_2}$ (as a  connected component of the set $\VV_{R_2}$) that 
\begin{equation*}
Q(\{x\in B':y(x)<R_2\})\subseteq\UU_{R_2}.
\end{equation*}
The claim \eqref{bingo} follows.
\end{proof}

We define now $\mathbf{N}'=\mathbf{N}_{R_2}\cap B_{r_2}(x_0)$. Since $y^2-y+B>\widetilde{C}_{R_2}^{-1}$ in $\mathbf{N}'$, the calculations following \eqref{stuff2} in the previous sections are also applicable in $\mathbf{N}'$. Recall the function $H$ defined in \eqref{stuff50},
\begin{equation*}
H=\frac{y^2-z^2-2y(B-z^2)}{(y^2-y+B)(y^2+z^2)}.
\end{equation*}
Since $B\in[0,1/4)$ (see Lemma \ref{yhorizon}) and $y\geq y_{S_0}+\widetilde{C}_{R_2}^{-1}\geq 1/2+\widetilde{C}_{R_2}^{-1}$, it follows that $H\geq\widetilde{C}_{R_2}^{-1}$ in $\mathbf{N}'$.

\subsection{Vanishing of $\SS$ is a neighborhood of $x_0$}\label{vanishingS2}

Assume $x_0\in\delta_{\Sigma_0\cap \mathbf{E}}(\UU_{R_2})$ is as before, and $r_2>0$ is constructed as in Lemma \ref{constr_2}. We show now that the tensor $\SS$ vanishes in a neighborhood of $x_0$.

\begin{proposition}\label{vanishingS3}
There is $r_3=r_3(A_0,\widetilde{A}_{\widetilde{C}^{-1}},R_2)\in(0,r_2)$ such that
$
\SS=0\text{ in }B_{r_3}(x_0).
$
\end{proposition}

As in section \ref{sectionS=01}, the main ingredient needed to prove Proposition \ref{vanishingS3} is a Carleman inequality. We define the smooth function $N^{x_0}:\Phi^{x_0}(B_1)=B_1(x_0)\to[0,\infty)$
\begin{equation*}
N^{x_0}(x)=|(\Phi^{x_0})^{-1}(x)|^2.
\end{equation*}

\begin{lemma}\label{Carl2}
There is $\ep\in (0,r_2]$ sufficiently small and $\widetilde{C}_\ep$ sufficiently
large such that for any $\lambda\geq\widetilde{C}_\ep$ and any $\phi\in C^\infty_0(B_{\ep^{10}}(x_0))$
\begin{equation}\label{Carb1}
\lambda \|e^{-\lambda \f_\ep}\phi\|_{L^2}+\|e^{-\lambda \f_\ep}|D^1\phi|\,\|_{L^2}\leq \widetilde{C}_\ep\lambda^{-1/2}\|e^{-\lambda \f_\ep}\,\square_{\g}\phi\|_{L^2}+\ep^{-6}\|e^{-\lambda \f_\ep}\T(\phi)\,\|_{L^2},
\end{equation}
where, with $R_2=y(x_0)$,
\begin{equation}\label{Carb2}
\f_\ep=\ln[y-R_2+\ep+\ep^{12}N^{x_0}].
\end{equation}
\end{lemma}

\begin{proof}[Proof of Lemma \ref{Carl2}] We will use the notation $\widetilde{C}_{R_2}$ to denote various constants in $[1,\infty)$ that may depend only on the constants $A_0$, $\widetilde{A}_{\widetilde{C}^{-1}}$, and $R_2$. We would like to apply Proposition \ref{Cargen} with $V=\T$, $h_\ep=y-R_2+\ep$ and $e_\eps=\ep^{12}N^{x_0}$. The condition \eqref{smallweight} for the negligible perturbation $e_\ep$ is clearly satisfied if $\ep$ is sufficiently small. It remains to show that there is $\ep_1$ sufficiently small such that the family of weights $\{h_\ep\}_{\ep\in(0,\ep_1)}$ satisfies the pseudo-convexity conditions \eqref{po5}, \eqref{po3.2}, and \eqref{po3}. 

Clearly, $h_\ep(x_0)=\ep$ and $\T(h_\ep)(x_0)=0$ since $\T(\sigma)=0$. Also $|D^j y|\leq \widetilde{C}_{R_2}$ for $j=1,2,3,4$ in $B_{r_2}(x_0)$, see \eqref{alexnew10}, thus condition \eqref{po5} is satisfied if $\eps_1$ is sufficiently small.  

To prove \eqref{po3.2} and \eqref{po3} we use the complex null tetrad $l=e_{4}$, $\Ll=e_{3}$, $\Mm=e_{1}$, $\Omm=e_{2}$, normalized as in \eqref{stuff2}. With $\D_{(\a)}=\D_{e_\al}$, using Lemma \ref{coefficients} and the definition $h_\eps=y-R_2+\ep$ we have
\begin{equation}\label{grady}
\D_{(1)}h_\eps=\D_{(2)}h_\eps=0,\qquad \D_{(3)}h_\eps=-W,\qquad \D_{(4)}h_\eps=1,
\end{equation}
and, using also $\eta=\ze\frac{\overline{P}}{P}$,
\begin{equation}\label{hessy}
\begin{cases}
\D_{(4)}\D_{(4)}h_\eps=\D_{(3)}\D_{(3)}h_\eps=0, \qquad\qquad &\D_{(4)}\D_{(3)}h_\eps=\D_{(3)}\D_{(4)}h_\eps=-WH\\
\D_{(4)}\D_{(1)}h_\eps=\D_{(1)}\D_{(4)}h_\eps=\ze,\qquad \qquad &\D_{(4)}\D_{(2)}h_\eps=\D_{(2)}\D_{(4)}h_\eps=\overline{\ze} \\
\D_{(3)}\D_{(1)}h_\eps=\D_{(1)}\D_{(3)}h_\eps=W\ze\frac{\overline{P}}{P}, \qquad\quad  &\D_{(3)}\D_{(2)}h_\eps=\D_{(2)}\D_{(3)}h_\eps=W\overline{\ze}\frac{P}{\overline{P}}\\
\D_{(1)}\D_{(2)}h_\eps=\D_{(2)}\D_{(1)}h_\eps=W\frac{2R_2}{R_2^2+z^2},\quad\, &\D_{(1)}\D_{(1)}h_\eps=\D_{(2)}\D_{(2)}h_\eps=0.
\end{cases}
\end{equation}
where all the functions are evaluated at $x_0$. Thus
\begin{equation*}
\D^\al h_\eps\D^\be h_\eps(\D_\al h_\ep\D_\be h_\ep-\ep\D_\al\D_\be h_\ep)=4W^2-2\ep W^2H,
\end{equation*}
which is  bounded from below by $\eps_1^2$ if $\eps_1$ is sufficiently small, since $W(x_0)\geq \widetilde{C}_{R_2}^{-1}$ and $|H(x_0)|\leq \widetilde{C}_{R_2}$. The condition \eqref{po3.2} is therefore satisfied.

We prove now condition \eqref{po3} for a vector $X=X^{(1)}e_1+\overline{X^{(1)}}e_2+Ye_3+Ze_4$, $Y,Z\in\mathbb{R}$, $X^{(1)}\in\mathbb{C}$. Recall, see \eqref{stuff8},
\begin{equation*}
\T/(2M)=\overline{\ze}Pe_1+\ze\overline{P}e_2-e_3-We_4.
\end{equation*}
Thus, using also \eqref{grady}
\begin{equation}\label{ro1}
\begin{split}
\ep^{-2}(&|X^\al\T_\al|^2+|X^\al\D_\al h_\ep|^2)\\
&=\ep^{-2}(Z-WY)^2+\ep^{-2}4M^2(\ze\overline{P}X^{(1)}+\overline{\ze}P\overline{X^{(1)}}+YW+Z)^2\\
&\geq(\ep^{-2}/2)(Z-WY)^2+(\ep^{-1}/2)(\ze\overline{P}X^{(1)}+\overline{\ze}P\overline{X^{(1)}}+2YW)^2 
\end{split}
\end{equation}
for $\ep$ sufficiently small. Using \eqref{hessy}
\begin{equation*}
\begin{split}
X^\al X^\be&(\mu\g_{\al\be}-\D_\al\D_\be h_\ep)=2X^{(1)}\overline{X^{(1)}}\Big(\mu-\frac{2R_2W}{R_2^2+z^2}\Big)+2YZ(-\mu+WH)\\
&-2\ze X^{(1)}[Z+WY(\overline{P}/P)]-2\overline{\ze}\overline{X^{(1)}}[Z+WY(P/\overline{P})].
\end{split}
\end{equation*}
Let $L=\ze\overline{P}X^{(1)}+\overline{\ze}P\overline{X^{(1)}}$. We write $Z=WY+Z-WY$, and then $L=-2WY+L+2WY$, and use
\begin{equation*}
1+(\overline{P}/P)=\overline{P}\frac{2R_2}{R_2^2+z^2},\qquad 1+(P/\overline{P})=P\frac{2R_2}{R_2^2+z^2}, 
\end{equation*}
to rewrite
\begin{equation}\label{ro2}
\begin{split}
&X^\al X^\be(\mu\g_{\al\be}-\D_\al\D_\be h_\ep)=2X^{(1)}\overline{X^{(1)}}\Big(\mu-\frac{2R_2W}{R_2^2+z^2}\Big)+2Y^2(-W\mu+W^2H)\\
&-\frac{4R_2}{R_2^2+z^2}WY\cdot L+(Z-WY)[2Y(-\mu+WH)-2\ze X^{(1)}-2\overline{\ze}\overline{X^{(1)}}]\\
&=2X^{(1)}\overline{X^{(1)}}\Big(\mu-\frac{2R_2W}{R_2^2+z^2}\Big)+2Y^2\Big(-W\mu+W^2H+\frac{4R_2W^2}{R_2^2+z^2}\Big)\\
&-\frac{4R_2}{R_2^2+z^2}WY\cdot (L+2WY)+(Z-WY)[2Y(-\mu+WH)-2\ze X^{(1)}-2\overline{\ze}\overline{X^{(1)}}].
\end{split}
\end{equation}
We set now $\mu=3R_2W/(R_2^2+z^2)$ and combine \eqref{ro1} and \eqref{ro2}. Since $H(x_0)\geq 0$ it follows that
\begin{equation*}
\begin{split}
&X^\al X^\be(\mu\g_{\al\be}-\D_\al\D_\be h_\ep)+\ep^{-2}(|X^\al\T_\al|^2+|X^\al\D_\al h_\ep|^2)\\
&\geq (\ep^{-2}/2)(Z-WY)^2+(\ep^{-1}/2)(L+2YW)^2+2|X^{(1)}|^2\frac{R_2W}{R_2^2+z^2}+2Y^2\frac{R_2W^2}{R_2^2+z^2}\\
&-\widetilde{C}_{R_2}(|Z-WY|+|L+WY|)(|Y|+|X^{(1)}|)\\
&\geq (\ep^{-2}/4)(Z-WY)^2+(\ep^{-1}/4)(L+2YW)^2+|X^{(1)}|^2\frac{R_2W}{R_2^2+z^2}+Y^2\frac{R_2W^2}{R_2^2+z^2}
\end{split}
\end{equation*}
if $\eps$ is sufficiently small, since $W\geq\widetilde{C}_{R_2}^{-1}$. It follows that
\begin{equation*}
X^\al X^\be(\mu\g_{\al\be}-\D_\al\D_\be h_\ep)+\ep^{-2}(|X^\al\T_\al|^2+|X^\al\D_\al h_\ep|^2)\geq \widetilde{C}_{R_2}^{-1}(Z^2+|X^{(1)}|^2+Y^2),
\end{equation*}
thus the condition \eqref{po3} is satisfied for $\eps_1$ sufficiently small. This completes the  proof of the lemma.
\end{proof}

We prove now Proposition \ref{vanishingS2}.

\begin{proof}[Proof of Proposition \ref{vanishingS2}] We use the Carleman estimate in Lemma \ref{Carl2} and Lemma \ref{constr_2}. In view of Lemma \ref{Carl2}, there are constants $\ep\in(0,r_2]$ and $\widetilde{C}_\ep\geq 1$ such that for any $\lambda\geq \widetilde{C}_\ep$ and any $\phi\in C^\infty_0(B_{\ep^{10}}(x_0))$,

\begin{equation}\label{Carbre1}
\lambda \|e^{-\lambda \f_\ep}\phi\|_{L^2}+\|e^{-\lambda \f_\ep}|D^1\phi|\,\|_{L^2}\leq \widetilde{C}_\ep\lambda^{-1/2}\|e^{-\lambda \f_\ep}\,\square_{\g}\phi\|_{L^2}+\ep^{-6}\|e^{-\lambda \f_\ep}\T(\phi)\,\|_{L^2},
\end{equation}
where
\begin{equation}\label{Carbre2}
\f_\ep=\ln[y-R_2+\ep+\ep^{12}N^{x_0}].
\end{equation}
The constant $\ep$ will remain fixed in this proof. For simplicity of notation, we replace the constants $\widetilde{C}_\ep$ with $\widetilde{C}_{R_2}$; since $\ep$ is fixed, these constants may depend only on the constants $A_0$, $\widetilde{A}_{\widetilde{C}^{-1}}$, and $R_2$. We will show that $\Ss\equiv 0$ in the set $B_{\ep^{100}}=B_{\ep^{100}}(x_0)$.

In view of Theorem \ref{wave} and the fact that $\T$ is a Killing vector-field
\begin{equation}\label{maineq}
\begin{cases}
&\square_\g\SS_{\al_1\ldots\al_4}=\SS_{\be_1\ldots \be_4}{\AA^{\be_1\ldots\be_4}}_{\alpha_1\ldots\alpha_4}+\D_{\be_5}\SS_{\be_1\ldots \be_4}{\BB^{\be_1\ldots\be_5}}_{\alpha_1\ldots\alpha_4};\\
&\mathcal{L}_\T\SS=0,
\end{cases}
\end{equation}
in $B_{\ep^{10}}(x_0)$, for some smooth tensor-fields $\AA$ and $\BB$. Also, using Lemma \ref{constr_2} and the fact that $\SS$ vanishes in $\UU_{R_2}$ (the bootstrap assumption),
\begin{equation}\label{maineq2}
\SS=0\text{ in }\{x\in B_{\ep^{10}}(x_0):y(x)<R_2\}.
\end{equation}

As in the proof of Proposition \ref{Lemmaa1}, for $(j_1,\ldots,j_4)\in\{1,2,3,4\}^4$ we define, using the coordinate chart $\Phi$,
\begin{equation*}
\phi_{(j_1\ldots j_4)}=\mathcal{S}(\partial_{j_1},\ldots ,\partial_{j_k}).
\end{equation*}
The functions $\phi_{(j_1\ldots j_4)}:B_{\ep^{10}}(x_0)\to\mathbb{C}$ are smooth. Let $\eta:\mathbb{R}\to[0,1]$ denote a smooth function supported
in $[1/2,\infty)$ and equal to $1$ in $[3/4,\infty)$. We define
\begin{equation*}
\begin{split}
\phi^{\ep}_{(j_1\ldots j_4)}=\phi_{(j_1\ldots j_4)}\cdot \big(1-\eta(N(x)/ \ep^{40})\big)=\phi_{(j_1\ldots j_4)}\cdot \widetilde{\eta}_\ep.
\end{split}
\end{equation*}
Clearly, $\phi^{\ep}_{(j_1\ldots j_4)}\in C^\infty _0(B_{\ep^{10}}(x_0))$ and
\begin{equation*}
\begin{cases}
&\square_\g\phi^{\ep}_{(j_1\ldots j_4)}=\widetilde{\eta}_{\ep}\cdot\square_\g\phi_{(j_1\ldots j_4)}+2\D_\al\phi_{(j_1\ldots j_4)}\cdot \D^\al \widetilde{\eta}_{\ep}+\phi_{(j_1\ldots j_4)}\cdot \square_\g\widetilde{\eta}_{\ep}\\
&\T(\phi^{\ep}_{(j_1\ldots j_4)})=\widetilde{\eta}_{\ep}\cdot \T(\phi_{(j_1\ldots j_4)})+\phi_{(j_1\ldots j_4)}\cdot \T(\widetilde{\eta}_{\ep}).
\end{cases}
\end{equation*}
Using the Carleman inequality \eqref{Carbre1}, for any $(j_1,\ldots j_4)\in\{1,2,3,4\}^4$ we have
\begin{equation}\label{tr30}
\begin{split}
&\lambda\cdot \|e^{-\lambda \f_{\ep}}\cdot \widetilde{\eta}_{\ep}\phi_{(j_1\ldots j_4)} \|_{L^2}+\|e^{-\lambda \f_{\ep}}\cdot \widetilde{\eta}_{\ep}|D^{1} \phi_{(j_1\ldots j_4)}|\, \|_{L^2}\\
&\leq \widetilde{C}_{R_2}\lambda ^{-1/2}\cdot \|e^{-\lambda \f_{\ep}}\cdot \widetilde{\eta}_{\ep}\square_\g\phi_{(j_1\ldots j_4)}\|_{L^2}+\widetilde{C}_{R_2}\|e^{-\lambda \f_\ep}\cdot \widetilde{\eta}_\ep\T(\phi_{(j_1\ldots j_4)})\|_{L^2}\\
&+\widetilde{C}_{R_2}\Big[\|e^{-\lambda \f_{\ep}}\cdot \D_\al\phi_{(j_1\ldots j_4)}\D^\al \widetilde{\eta}_{\ep} \|_{L^2}+\|e^{-\lambda \f_{\ep}}\cdot \phi_{(j_1\ldots j_4)} ( |\square_\g\widetilde{\eta}_{\ep}|+|D^1\widetilde{\eta}_{\ep}| )\|_{L^2}\Big],
\end{split}
\end{equation}
for any $\lambda\geq\widetilde{C}_{R_2}$. Using the identities in \eqref{maineq}, in $ B_{\ep^{10}}(x_0)$ we estimate pointwise
\begin{equation}\label{tr31}
\begin{cases}
&|\square_\g\phi_{(j_1\ldots j_4)}|\leq \widetilde{C}_{R_2}\sum_{l_1,\ldots,l_4}\big(|D^1\phi_{(l_1\ldots l_4)}|+|\phi_{(l_1\ldots l_4)}|\big);\\
&|\T(\phi_{(j_1\ldots j_4)})|\leq \widetilde{C}_{R_2}\sum_{l_1,\ldots,l_4}|\phi_{(l_1\ldots l_4)}|.
\end{cases}
\end{equation}
We add up the inequalities \eqref{tr30} over $(j_1,\ldots,j_4)\in\{1,2,3,4\}^4$. The key observation is that, in view of \eqref{tr31}, the first two terms in the right-hand side can be absorbed into the left-hand side for $\lambda$ sufficiently large. Thus, for any $\lambda\geq \widetilde{C}_{R_2}$
\begin{equation}\label{tr32}
\begin{split}
&\lambda\sum_{j_1,\ldots ,j_4}\|e^{-\lambda \f_{\ep}}\cdot \widetilde{\eta}_{\ep}\phi_{(j_1\ldots j_4)} \|_{L^2}\\
&\leq \widetilde{C}_{R_2}\sum_{j_1,\ldots ,j_4}\Big[\|e^{-\lambda \f_{\ep}}\cdot \D_\al\phi_{(j_1\ldots j_4)}\D^\al \widetilde{\eta}_{\ep} \|_{L^2}+\|e^{-\lambda \f_{\ep}}\cdot \phi_{(j_1\ldots j_4)} ( |\square_\g\widetilde{\eta}_{\ep}|+|D^1\widetilde{\eta}_{\ep}| )\|_{L^2}\Big].
\end{split}
\end{equation}

Using the hypothesis \eqref{maineq2} and the definition of the function $\widetilde{\eta}_\ep$, we have
\begin{equation*}
|\D_\al\phi_{(j_1\ldots j_4)}\D^\al \widetilde{\eta}_{\ep}|+\phi_{(j_1\ldots j_4)} ( |\square_\g\widetilde{\eta}_{\ep}|+|D^1\widetilde{\eta}_{\ep}| )\leq\widetilde{C}_{R_2}\cdot \mathbf{1}_{\{x\in B_{\ep^{10}}(x_0):\,y(x)\geq R_2\text{ and }N(x)\geq\ep^{50}\}}.
\end{equation*}
Using the definition \eqref{Carbre2}, we observe also that
\begin{equation*}
\inf_{B_{\ep^{100}}}e^{-\lambda \f_\ep}\geq e^{-\lambda \ln(\ep+\ep^{70})}\geq\sup_{\{x\in B_{\ep^{10}}(x_0):\,y(x)\geq R_2\text{ and }N(x)\geq\ep^{50}\}}e^{-\lambda \f_\ep}.  
\end{equation*}
It follows from these last two inequalities and \eqref{tr32} that
\begin{equation*}
\lambda\sum_{j_1,\ldots ,j_4}\|\mathbf{1}_{B_{\ep^{100}}}\cdot \phi_{(j_1\ldots j_4)} \|_{L^2}\leq \widetilde{C}_{R_2}\sum_{j_1,\ldots ,j_4}\|\mathbf{1}_{\{x\in B_{\ep^{10}}(x_0):\,y(x)\geq R_2\text{ and }N(x)\geq\ep^{50}\}}\|_{L^2},
\end{equation*}
for any $\lambda\geq \widetilde{C}_{R_2}$. The proposition follows by letting $\lambda\to\infty$.
\end{proof}

\subsection{Proof of Proposition \ref{bootstrap2} and the Main Theorem}\label{subsectionlast} In this subsection we complete the proof of the Main Theorem.

\begin{proof}[Proof of Proposition \ref{bootstrap2}] In view of Proposition \ref{vanishingS3}, the tensor $\SS$ vanishes in the connected open set $\mathbf{N'}=\mathbf{N}_{R_2}\bigcup\big(\cup_{x_0\in \delta_{\Sigma_0\cap\mathbf{E}}(\UU_{R_2})}B_{r_3}(x_0)\big)$. It remains to show that for  some  $r'\ll r_3$ we have
\bea
\label{endp1}
\UU_{R_2+r'} &\subseteq & \UU_{R_2}\cup\big(\cup_{x_0\in \delta_{\Sigma_0\cap\mathbf{E}}(\UU_{R_2})}G_{r_3/\tilde{C}}(x_0)\big),\\ 
 G_{r}(x_0) &= &   \{x\in B_{r}(x_0)\cap\Sigma_0:y(x)<R_2+r'\}\big),\nn
\eea
where $\widetilde{C}$ is sufficiently large so that,
\bea
\overline{\big(\cup_{x_0\in \delta_{\Sigma_0\cap\mathbf{E}}(\UU_{R_2})}G_{r_3/\widetilde{C}}(x_0)\big)}\subseteq \big(\cup_{x_0\in \delta_{\Sigma_0\cap\mathbf{E}}(\UU_{R_2})}\overline{G_{r_3/4}(x_0) }\big),\label{choose-tildeC}
\eea
with the bars denoting  the closures in $\Si_0$.  We observe that such 
a constant exists in view of the fact that $\delta_{\Sigma_0\cap\mathbf{E}}(\UU_{R_2})$ is compact and the function $y$ tends to infinity in the asymptotic region of $\Si_0$ (in view of our assumption {\bf AF}).

Assume, by contradiction,  that \eqref{endp1} does not hold, thus there exists  $p\in\UU_{R_2+r'}$ which does not belong in  the open set (in $\Sigma_0$) in the right-hand side of \eqref{endp1}. Let $\gamma:[0,1]\to\UU_{R_2+r'}\cup S_0$ denote a smooth curve such that $\gamma(0)\in S_0$ and $\gamma(1)=p$. Let $p'=\gamma(t')$ denote the first point on this curve which is not in the open set in the right-hand side of \eqref{endp1}. Clearly, $p'$ does not belong to the closure of $\UU_{R_2}$, thus
\begin{equation*}
p'\in \overline{\cup_{x_0\in \delta_{\Sigma_0\cap\mathbf{E}}(\UU_{R_2})}G_{r_3/\tilde{C}} (x_0)  }.
\end{equation*}
In view of \eqref{choose-tildeC} we infer that, for some  
$x_0\in \delta_{\Sigma_0\cap\mathbf{E}}   (\UU_{R_2})$,
\begin{equation}\label{endp2}
p'\in\{x\in B_{r_3/2}(x_0)\cap\Sigma_0:y(x)<R_2+r'\}. 
\end{equation}  

Recall our  smooth vector-field $Y=\g^{\al\be}\partial_\al y\partial_\be$,
 see \eqref{vect-Y} and discussion following it,     with the property that $\g(Y,Y)\geq \widetilde{C}_{R_2}^{-1}$ in $B_{r_3}(x_0)$. We consider the integral curve starting from the point $p'$ and flowing (backwards) a short distance $\widetilde{C}_{R_2}^{-1}$ (much smaller than $r_3$) along $Y$, and project this integral curve to $\Sigma_0$ using the smooth projection $Q: \cup_{|t|<\delta_2}\Phi_t(B_{r_3}(x_0)\cap\Sigma_0)\to B_{r_3}(x_0)\cap\Sigma_0$.
  The resulting curve is a smooth curve in  $B_{r_3}(x_0)\cap\Sigma_0$; if $r'$ sufficiently small then this curve contains a point $p''$ such that $y(p'')<R_2$. In view  of Lemma \ref{constr_2}, $p''\in\UU_{R_2}$, thus there is a point $p'''\in\delta_{\Sigma_0\cap\mathbf{E}}(\UU_{R_2})$ on the curve joining $p'$ and $p''$. Then $p'\in B_{r_3/\widetilde{C}}(p''')$, which gives a contradiction.
\end{proof}

To complete the proof of the Main Theorem we use Proposition \ref{bootstrap2} and Proposition \ref{uniformy}. Using Proposition \ref{bootstrap2}, it follows that the tensor $\SS$ vanishes in the connected component of the set $\Sigma'_0$ whose closure in $\Sigma_0$ contains $S_0$. Assume $(\Sigma_0\cap\E)\setminus\Sigma'_0\neq\emptyset$ and let $p\in(\Sigma_0\cap\E)\setminus\Sigma'_0$. Assume $\gamma:[0,1]\to\Sigma_0\cap\overline{\E}$ is a smooth curve such that $\gamma(0)\in S_0$ and $\gamma(1)=p$. Let $p'=\gamma(t')$ denote the first point on this curve which is not in $\Sigma'_0\cup S_0$. Thus $\gamma(t'')$ belongs to the connected component of the set $\Sigma'_0$ whose closure in $\Sigma_0$ contains $S_0$ for any $t''<t'$. Since $\SS$ vanishes in this connected component, it follows from Lemma \ref{uniformy} that the function $y$ is bounded by a constant at all points $\gamma(t'')$, $t''<t'$. Thus $p'\in\Sigma'_0$, contradiction.

It follows that $\Sigma'_0=\Sigma\cap\E$ and $\SS=0$ in $\Sigma\cap\E$, which establishes the claim \eqref{mainclaim}. 

\appendix 

\section{The main formalism}\label{NP}

\subsection{Horizontal structures}
Assume $(\mathbf{N},\g)$ is a smooth\footnote{As before, $\mathbf{N}$ is assumed to be a connected, orientable, paracompact $C^\infty$ manifold without boundary.} vacuum Einstein space-time of dimension $4$. Assume $(l,\ul)$ is a null pair on $\mathbf{N}$, i.e.
\begin{equation*}
\g(l,l)=\g(\ul,\ul)=0\text{ and }\g(l,\ul)=-1.
\end{equation*}
We say that a vector-field $X$ is {\it{horizontal}} if
\begin{equation*}
\g(l,X)=\g(\ul,X)=0.
\end{equation*}
Let $\O(\N)$ denote the vector space of horizontal vector-fields on $\N$. We define the induced metric,
and induced  volume form,
\begin{equation}\label{induced}
\begin{cases}
&\gamma(X,Y)=\g(X,Y)\,\qquad \qquad \forall\, X, Y\in \O(\N),\\
&\in(X,Y)=\in(X,Y,\lb,l)\qquad \forall\,  X, Y\in \O(\N).
\end{cases}
\end{equation}
where $\in$ denotes the standard volume form
on $\N$.  If $(e_a)_{a=1,2}$ is an orthonormal
basis of horizontal vector-fields, i.e. $\ga(e_a, e_b)
=\de_{ab}$, we write $\in_{ab}=\in(e_a, e_b)$ and 
without loss of generality we assume that $\in_{12}=1$.  

In general  the 
commutator $[X,Y]$ of two horizontal vector-fields
may fail  to be horizontal. We say that the pair $(l,\lb)$ is \emph{integrable} if the set of 
horizontal vector-fields  forms an integrable distribution,
i.e. $X, Y\in\O(\mathbf{N}) $ implies that $[X,Y]\in\O(\N)$.
For any vector-field $X\in \mathbf{T}(\mathbf{N})$ we define its horizontal projection
\begin{equation*}
{}^{(h)}X=X+\g(X,\ul)l+\g(X,l)\ul.
\end{equation*}
Using this projection we define the horizontal covariant derivative $\nabla_XY$, $X\in\mathbf{T}(\mathbf{N})$, $Y\in \mathbf{O}(\mathbf{N})$,
\begin{equation*}
\nabla_XY={}^{(h)}(\D_XY)=\D_XY-g(\D_X\ul,Y)l-g(\D_Xl,Y)\ul.
\end{equation*}
The definition shows easily that,
\begin{equation}\label{cov1}
\begin{cases}
&\nabla_{fX+f'X'}Y=f\nabla_XY+f'\nabla_{X'}Y;\\
&\nabla_X(fY+f'Y')=f\nabla_XY+X(f)Y+f'\nabla_XY'+X(f')Y';\\
&X \ga(Y,Y')=\ga(\nabla_XY,Y')+\ga(Y,\nabla_XY'),
\end{cases}
\end{equation}
for any $X,X'\in\mathbf{T}(\mathbf{N})$, $Y,Y'\in\mathbf{O}(\mathbf{N})$, $f,f'\in C^\infty(\mathbf{N})$.  In particular we see that $\nab$ is compatible with the horizontal  metric $\ga$.

In what follows  we  identify covariant and contravariant  horizontal  tensor-fields  using  the induced metric ${}^{(h)}\ga$.
For any $k\in\mathbb{Z}_+$ let
  $\mathbf{O}_k(\mathbf{N})$ denote the vector space of {\it{$k$ horizontal tensor-fields}}
\begin{equation*}
U:\mathbf{O}(\mathbf{N})\times\ldots\times\mathbf{O}(\mathbf{N})\to\mathbb{C}.
\end{equation*}
Given a horizontal tensor-field $U\in\mathbf{O}_k(\mathbf{N})$ and $X\in\mathbf{T}(\mathbf{N})$ we define the covariant derivative $\nabla_XU\in\mathbf{O}_k(\mathbf{N})$ by the formula
\begin{equation}\label{cov2}
\nabla_XU(Y_1,\ldots,Y_k)=X(U(Y_1,\ldots,Y_k))-U(\nabla_XY_1,\ldots,Y_k)-\ldots-U(Y_1,\ldots,\nabla_XY_k).
\end{equation}
According to the  definition  the mapping $(X,Y_1,\ldots,Y_k)\to\nabla_XU(Y_1,\ldots,Y_k)$ is a multilinear mapping on $\mathbf{T}(\mathbf{N})\times\mathbf{O}(\mathbf{N})\times\ldots\times\mathbf{O}(\mathbf{N})$.

We define  the null second fundamental forms ${}^{(h)}\chi,{}^{(h)}\underline\chi\in\mathbf{O}_2(\mathbf{N})$ by
\begin{equation}\label{fo1}
\begin{cases}
&{}^{(h)}\chi(X,Y)=g(\D_Xl,Y),\\
&{}^{(h)}\underline\chi(X,Y)=g(\D_X\ul,Y).
\end{cases}
\end{equation}
Observe that    ${}^{(h)}\chi$
 and ${}^{(h)}\chib$  are  symmetric if and 
 only if   the horizontal structure is 
 integrable. Indeed this follows easily from
 the formulas,
 \beaa
 {}^{(h)}\chi(X,Y)-{}^{(h)}\chi(Y,X)&=&\g(\D_X l, Y)-\g(\D_Yl,X)=-\g(l, [X,Y])\\
  {}^{(h)}\chib(X,Y)-{}^{(h)}\chib(Y,X)&=&\g(\D_X \lb, Y)-\g(\D_Y\lb,X)=-\g(\lb, [X,Y]).
\eeaa
The trace of an horizontal 2-tensor $U$ is defined
according to 
\beaa
\mbox{tr} (U):=\de^{ab}U_{ab}
\eeaa
where $(e_a)_{a=1,2}$ is an arbitrary orthonormal
frame of horizontal vector-fields. Observe that the definition does not depend on the particular 
frame. We denote by  $\trch$  and $\trchb$ the traces of ${}^{(h)}\chi$ and ${}^{(h)}\chib$. 
If  $U\in \O_k(\N) $ with $k=1,2$   we define its dual, expressed relative to an arbitrary orthonormal 
frame  $(e_a)_{a=1,2}\in \O(\N)$,
\beaa
\dual U_a=\in_{ab}U_b,\qquad \dual U_{ab}=
\in_{ac}U_{cb}
\eeaa
  Clearly    $\dual (\dual \om)=-\om$.     If 
   $\om\in\O(\N)_2$  is symmetric traceless then so is its dual $\dual\om$.

We define also the horizontal $1$-forms ${}^{(h)}\xi,{}^{(h)}\xib,{}^{(h)}\eta,{}^{(h)}\etab,{}^{(h)}\ze\in\mathbf{O}_1(\mathbf{N})$ by
\begin{equation}\label{fo2}
\begin{cases}
&{}^{(h)}\xi(X)=\g(\D_ll,X),\quad {}^{(h)}\xib(X)=\g(\D_{\ul}\ul,X),\\
&{}^{(h)}\eta(X)=\g(\D_{\ul}l,X),\quad {}^{(h)}\etab(X)=\g(\D_l\ul,X),\\
&{}^{(h)}\ze(X)=\g(\D_X l,\ul),
\end{cases}
\end{equation}
and the real scalars
\begin{equation}\label{fo3}
\om=\g(\D_ll,\ul),\quad\omb=\g(\D_{\ul}\ul,l).
\end{equation}

Assume that $W\in\mathbf{T}_4^0(\mathbf{N})$ is a Weyl field, i.e.
\begin{equation}\label{fo4}
\begin{cases}
&W_{\al\be\mu\nu}=-W_{\be\al\mu\nu}=-W_{\al\be\nu\mu}=W_{\mu\nu\al\be};\\
&W_{\al\be\mu\nu}+W_{\al\mu\nu\be}+W_{\al\nu\be\mu}=0;\\
&\g^{\be\nu}W_{\al\be\mu\nu}=0.
\end{cases}
\end{equation}
We define the null components of the Weyl field $W$, $\al(W),\aa(W),\varrho(W)\in \mathbf{O}_2(\mathbf{N})$ and $\be(W),\bb(W)\in\mathbf{O}_1(\mathbf{N})$ by the formulas
\begin{equation}\label{fo5}
\begin{cases}
\al(W)(X,Y)=W(l,X,l,Y),\\
\aa(W)(X,Y)=W(\ul,X,\ul,Y),\\
\be(W)(X)=W(X,l,\ul,l),\\
\bb(W)(X)=W(X,\ul,l,\ul),\\
\varrho(W)(X,Y)=W(X,\ul,Y,l).
\end{cases}
\end{equation}
Recall that if $W$ is a Weyl field its 
Hodge dual $\dual W$, defined by 
${}^{\ast}W_{\al\be\mu\nu}=\frac{1}{2}{\in_{\mu\nu}}^{\rho\si}W_{\al\be\rho\si}$,  is also a Weyl field. We easily
check the formulas,
\begin{equation}
\begin{cases}
&\aa(\dual W)=\dual \aa(W),\qquad \a(\dual W)=-
\dual \a(W) \\
&\bb(\dual W)=\dual\bb(W),\qquad \b(\dual W)=-\dual \b(W)\\
&\varrho(\dual W)=\dual \varrho(W) 
\end{cases}
\end{equation}
It is easy to check that $\a,\aa$ are symmetric traceless  horizontal tensor-fields   in $\O_2(\N)$.
 On the other hand    $\varrho\in \O_2(\N)$ is however  neither symmetric nor 
 traceless.  It is convenient to express it in terms
 of  the following  two scalar quantities,
 \bea
 \label{fo5'}
 \rho(W)=W(l,\lb,l,\lb),\qquad \dual\rho(W)=
 \dual W(l,\lb,l,\lb)\label{rho-dualrho}.
 \eea
 Observe also that,
 \beaa
\rho(\dual W)=\rhod(W), \qquad \rhod(\dual W)=-\rho.
\eeaa  
Thus,
\bea
\varrho(X,Y)=\frac{1}{2}\big(-\rho\,\ga(X,Y)+\rhod\, \in(X,Y)\big),\qquad 
\forall\, X,Y\in \O(\N).
\eea
We have,
 $
W(X, Y, \lb,l)=\varrho(W)(X,Y)-\varrho(W)(Y,X)=\rhod(W)\in(X,Y)
$. 
Also, since  $\dual(\dual W)=-W$, we deduce  that
$W(X, Y, X', Y')= \in(X,Y)\dual 
W(X', Y' , \lb, l)=\in(X,Y)\in(X', Y')
 \rhod(\dual W)$.  Therefore,
\beaa
\begin{cases}
&W(X, Y, \lb,l)=\in(X,Y)\rhod(W)\\
&W(X, Y, X' , Y')=- \in(X,Y)\in(X', Y')\rho(W)\\
&W(X,Y, Z ,\lb)=\in(X,Y)\,\, \bb(W)(Z).\\
\end{cases}
\eeaa
We also consider the case of a self-dual 
Weyl field $\WW=W+i\dual W$, i.e.
$\dual \WW=-i \WW$. Defining the null 
decomposition  $\aa(\WW),\, \bb(\WW),\, \rho(\WW), \, \rhod(\WW),
\b(\WW),\,  \a(\WW)$ as in \eqref{fo5}, \eqref{fo5'} and  setting  $\rhod(\WW):=\rho(\dual\WW)$  as in \eqref{fo5'},
we find, 
\beaa
\rhod(\WW)=-i \rho(\WW)
\eeaa
Relative to a null frame $e_1, e_2, e_3=\lb, l= e_4$ we have,
\bea
\WW_{ab34}=-i \in_{ab}\rho(\WW),\quad \WW_{abcd}=
-\in_{ab}\in_{cd}\rho(\WW),\quad \WW_{abc3}=\in_{ab}\bb_c
(\WW)\label{dualWW-form}
\eea

\subsection{Complex null tetrads}\label{complexnull}
We extend  by linearity the definition of 
horizontal vector-fields to complex ones.
We say that a complex vector-field $m$ on $\mathbf{N}$ is  {\it{compatible}} with the null pair $(\lb, l)$
if, , i.e. 
\begin{equation*}
\g(l,m)=\g(\ul,m)=\g(m,m)=0,\quad \g(m,\overline{m})=1.
\end{equation*}
In that case we say that
$(m,\overline{m},\ul,l)$ forms \emph{ a complex null tetrad}.
Clearly $m$ is compatible if and only if $m=\frac{1}{\sqrt{2}}(X+iY)$ for some  real vectors $X,Y\in\mathbf{O}(\mathbf{N})$ with $g(X,Y)=0$, $g(X,X)=g(Y,Y)=1$.
Given a compatible vector-field $m$ and ${}^{(h)}U\in \mathbf{O}_1(\mathbf{N})$ we can define the complex scalar $U_1:\mathbf{N}\to\mathbb{C}$,
\begin{equation*}
U_1={}^{(h)}U(m).
\end{equation*}
Similarly, given ${}^{(h)}V\in \mathbf{O}_2(\mathbf{N})$ we can define the complex scalars $V_{21},V_{11}:\mathbf{N}\to\mathbb{C}$,
\begin{equation*}
V_{21}={}^{(h)}V(\overline{m},m),\quad V_{11}={}^{(h)}V(m,m).
\end{equation*}
The complex scalars $U_1$, respectively $V_{21}$ and $V_{11}$, determine uniquely the real  horizontal tensors fields ${}^{(h)}U$ and ${}^{(h)}V$ respectively.

Given a compatible vector-field $m$ we define (compare with \eqref{fo1}, \eqref{fo2}, and \eqref{fo3})
\begin{equation}\label{fo10}
\begin{split}
\th={}^{(h)}\chi(\overline{m},m)=\g(\D_{\overline{m}}l,m)&,\qquad \thb={}^{(h)}\underline\chi(\overline{m},m)=\g(\D_{\overline{m}}\ul,m),\\
\va={}^{(h)}\chi(m,m)=\g(\D_ml,m)&,\qquad \vab={}^{(h)}\underline\chi(m,m)=\g(\D_m\ul,m),\\
\xi={}^{(h)}\xi(m)=\g(\D_ll,m)&,\qquad \xib={}^{(h)}\xib(m)=\g(\D_{\ul}\ul,m),\\
\eta={}^{(h)}\eta(m)=\g(\D_{\ul}l,m)&,\qquad \etab={}^{(h)}\etab(m)=\g(\D_l\ul,m),\\
\om=g(\D_ll,\ul)&,\qquad \omb=\g(\D_{\ul}\ul,l),\\
\ze={}^{(h)}\ze(m)=\g(\D_m l,\ul)&.
\end{split}
\end{equation}
The complex scalars $\th,\thb,\va,\vab,\xi,\xib,\eta,\etab,\ze$ and the real scalars $\om,\omb$ are the main connection coefficients of the null tetrad. 

Similarly, given a real-valued Weyl field $W$ we define (compare with \eqref{fo5})
\begin{equation}\label{fo11}
\begin{cases}
\Psi_{(2)}=\Psi_{(2)}(W)=\al(W)(m,m)=W(l,m,l,m),\\
{\underline\Psi}_{(2)}={\underline\Psi}_{(2)}(W)=\aa(W)(m,m)=W(\ul,m,\ul,m),\\
\Psi_{(1)}=\Psi_{(1)}(W)=\be(W)(m)=W(m,l,\ul,l),\\
{\underline\Psi}_{(1)}={\underline\Psi}_{(1)}(W)=\bb(W)(m)=W(m,\ul,l,\ul),\\
\Psi_{(0)}=\Psi_{(0)}(W)=\varrho(W)(\overline{m},m)=W(\overline{m},\ul,m,l).
\end{cases}
\end{equation}
Notice that, in view of \eqref{fo4}, $\al(W)(\overline{m},m)=\aa(W)(\overline{m},m)=\varrho(W)(m,m)=0$, so the scalars $\Psi_{(2)},{\underline\Psi}_{(2)}, \Psi_{(1)}, {\underline\Psi}_{(1)}, \Psi_{(0)}$ uniquely determine the real-valued Weyl field $W$. In addition, if 
\begin{equation*}
{}^{\ast}W_{\al\be\mu\nu}=\frac{1}{2}{\in_{\mu\nu}}^{\rho\si}W_{\al\be\rho\si}
\end{equation*}
is the dual dual of $W$, and the null tetrad $(m,\overline{m},\ul,l)$ has positive orientation (i.e. $\in_{\al\be\mu\nu}m^\al\overline{m}^\be\ul^\mu l^\nu=i$) then
\begin{equation}\label{fo11.11}
\begin{split}
&\Psi_2({}^{\ast}W)=(-i)\Psi_2(W),\quad \Psi_1({}^{\ast}W)=(-i)\Psi_1(W), \quad\Psi_0({}^{\ast}W)=(-i)\Psi_0(W),\\
&{\underline\Psi}_{(2)}({}^{\ast}W)=i{\underline\Psi}_2(W),\quad\quad {\underline\Psi}_{(1)}({}^{\ast}W)=i{\underline\Psi}_1(W).
\end{split}
\end{equation} 
In what follows we denote,
\begin{equation*}
e_{1}=m,\,\,e_{2}=\overline{m}\,\,e_{3}=\ul,\,\,e_{4}=l.
\end{equation*}
We define the connection coefficients ${\Gamma^{\mu}}_{\a\b},\Gamma_{\mu\a\b}$ by the formulas
\begin{equation*}
\D_{e_{\b}}e_{\a}={\Gamma^{\mu}}_{\a\b}e_{\mu}.
\end{equation*}
and
\begin{equation*}
\Gamma_{\mu\a\b}=\g_{\mu\nu}{\Gamma^{\nu}}_{\a\b}=\g(e_{\mu},\D_{e_{\b}}e_{\a}).
\end{equation*}
Clearly
\begin{equation*}
\Gamma_{\mu\a\b}+\Gamma_{\a\mu\b}=0.
\end{equation*}
We easily check  the formulas,
\begin{equation}\label{def1}
\begin{split}
&\Gamma_{144}=\xi,\quad\Gamma_{244}=\overline{\xi},\quad\Gamma_{133}=\xib,\quad\Gamma_{233}=\overline{\xib},\\
&\Gamma_{143}=\eta,\quad\Gamma_{243}=\overline{\eta},\quad\Gamma_{134}=\etab,\quad\Gamma_{234}=\overline{\etab},\\
&\Gamma_{142}=\th,\quad\Gamma_{241}=\overline{\th},\quad\Gamma_{132}=\thb,\quad\Gamma_{231}=\overline{\thb},\\
&\Gamma_{141}=\va,\quad\Gamma_{242}=\overline{\va},\quad\Gamma_{131}=\vab,\quad\Gamma_{232}=\overline{\vab},\\
&\Gamma_{344}=\om,\quad\Gamma_{433}=\omb,\quad\Gamma_{341}=\ze,\quad\Gamma_{342}=\overline{\ze}.\\
&
\end{split}
\end{equation}
Using the definition \eqref{cov2} we see easily that if ${}^{(h)}U\in\mathbf{O}_1(\mathbf{N})$, ${}^{(h)}V\in\mathbf{O}_2(\mathbf{N})$, and $\a\in\{1,2,3,4\}$ then
\begin{equation}\label{formula1}
\nabla_{\a}{}^{(h)}U_{1}=(e_{\a}+\Gamma_{12\a})({}^{(h)}U_{1}),
\end{equation}
and
\begin{equation}\label{formula2}
\begin{split}
\nabla_{\a}{}^{(h)}V_{11}=(e_{\a}+2\Gamma_{12\a})({}^{(h)}V_{11}),\quad\nabla_{\a}{}^{(h)}V_{21}=e_{\a}({}^{(h)}V_{21}).
\end{split}
\end{equation}

\subsection{The null structure equations and the Bianchi identities}

We define
\begin{equation*}
D=l=e_{4},\,\,\underline{D}=\ul=e_{3},\,\,\delta=m=e_{1},\,\,\overline{\delta}=\overline{m}=e_{2}.
\end{equation*}
Let $R$ denote the Riemann curvature tensor on $\mathbf{M}$. We compute
\begin{equation*}
\begin{split}
R_{\a\b\mu\nu}&=\g(e_{\a},[\D_{e_{\mu}}(\D_{e_{\nu}}e_{\b})-\D_{e_{\nu}}(\D_{e_{\mu}}e_{\b})-\D_{[e_{\mu},e_{\nu}]}e_{\b}])\\
&=\g(e_{\a},[\D_{e_{\mu}}({\Gamma^{\rho}}_{\b\nu}e_{\rho})-\D_{e_{\nu}}({\Gamma^{\rho}}_{\b\mu}e_{\rho})-({\Gamma^{\rho}}_{\nu\mu}-{\Gamma^{\rho}}_{\mu\nu})\D_{e_{\rho}}e_{\b}])\\
&=e_{\mu}(\Gamma_{\a\b\nu})-e_{\nu}(\Gamma_{\a\b\mu})+{\Gamma^{\rho}}_{\b\nu}\Gamma_{\a\rho\mu}-{\Gamma^{\rho}}_{\b\mu}\Gamma_{\a\rho\nu}+({\Gamma^{\rho}}_{\mu\nu}-{\Gamma^{\rho}}_{\nu\mu})\Gamma_{\a\b\rho}.
\end{split}
\end{equation*}
Using this formula and the table \eqref{def1} we derive the null structure equations. Using $R_{1441}=-\Psi_{(2)}(R)$ we derive
\begin{equation}\label{fo21}
(D+2\Gamma_{124})\va-(\delta+\Gamma_{121})\xi=\xi(2\ze+\eta+\etab)-\va(\om+\th+\overline{\th})-\Psi_{(2)}(R).
\end{equation}
Using $R_{1331}=-{\underline\Psi}_2(R)$ we derive
\begin{equation}\label{fo21'}
(\underline{D}+2\Gamma_{123})\vab-(\delta+\Gamma_{121})\xib=\xib(-2\ze+\etab+\eta)-\vab(\omb+\thb+\overline{\thb})-{\underline\Psi}_{(2)}(R).
\end{equation}
Using $R_{1442}=0$ we derive
\begin{equation}\label{fo22}
D\th-(\overline{\delta}+\Gamma_{122})\xi=-\th^2-\om\th-\va\overline{\va}+\overline{\xi}\eta+\xi(2\overline{\ze}+\overline{\etab}).
\end{equation}
Using $R_{1332}=0$ we derive
\begin{equation}\label{fo22'}
\underline{D}\,\thb-(\overline{\delta}+\Gamma_{122})\xib=-\thb^2-\omb\,\thb-\vab\,\overline{\vab}+\overline{\xib}\,\etab+\xib(-2\overline{\ze}+\overline{\eta}).
\end{equation}
Using $R_{1443}=-\Psi_{(1)}(R)$ we derive
\begin{equation}\label{fo23}
(D+\Gamma_{124})\eta-(\underline{D}+\Gamma_{123})\xi=-2\omb\xi+\th(\etab-\eta)+\va(\overline{\etab}-\overline\eta)-\Psi_{(1)}(R).
\end{equation}
Using $R_{1334}=-\underline{\Psi}_{(1)}(R)$ we derive
\begin{equation}\label{fo23'}
(\underline{D}+\Gamma_{123})\etab-(D+\Gamma_{124})\xib=-2\om\xib+\thb(\eta-\etab)+\vab(\overline{\eta}-\overline\etab)-\underline{\Psi}_{(1)}(R).
\end{equation}
Using $R_{1431}=0$ we derive
\begin{equation}\label{fo24}
(\underline{D}+2\Gamma_{123})\va-(\delta+\Gamma_{121})\eta=\eta^2+\xib\xi-\vab\th+\va(\omb-\overline{\thb}).
\end{equation}
Using $R_{1341}=0$ we derive
\begin{equation}\label{fo24'}
(D+2\Gamma_{124})\vab-(\delta+\Gamma_{121})\etab=\etab^2+\xi\xib-\va\thb+\vab(\om-\overline{\th}).
\end{equation}
Using $R_{1432}=-\Psi_{(0)}(R)$ we derive
\begin{equation}\label{fo25}
\underline{D}\th-(\overline{\delta}+\Gamma_{122})\eta=\xi\overline{\xib}+\eta\overline{\eta}-\va\overline{\vab}+\theta(\omb-\thb)-\Psi_{(0)}(R).
\end{equation}
Using $R_{1342}=-\overline{\,\Psi_{(0)}(R)}$ we derive
\begin{equation}\label{fo25'}
D\thb-(\overline{\delta}+\Gamma_{122})\etab=\xib\overline{\xi}+\etab\,\overline{\etab}-\vab\overline{\va}+\thb(\om-\th)-\overline{\,\Psi_{(0)}(R)}.
\end{equation}
Using $R_{1421}=-\Psi_{(1)}(R)$ we derive
\begin{equation}\label{fo26}
(\overline{\delta}+2\Gamma_{122})\va-\delta\theta=\ze\theta-\overline{\ze}\va+\eta(\th-\overline{\th})+\xi(\thb-\overline{\thb})-\Psi_{(1)}(R).
\end{equation}
Using $R_{1321}=-\underline{\Psi}_{(1)}(R)$ we derive
\begin{equation}\label{fo26'}
(\overline{\delta}+2\Gamma_{122})\vab-\delta\thb=-\ze\thb+\overline{\ze}\vab+\etab(\thb-\overline{\thb})+\xib(\th-\overline{\th})-\underline{\Psi}_{(1)}(R).
\end{equation}
Using $R_{3441}=-\Psi_{(1)}(R)$ we derive
\begin{equation}\label{fo27}
(D+\Gamma_{124})\ze-\delta\om=\om(\ze+\etab)+\overline{\th}(\etab-\ze)+\va(\overline{\etab}-\overline{\ze})-\xi(\overline{\thb}+\omb)-\overline{\xi}\vab-\Psi_{(1)}(R).
\end{equation}
Using $R_{4331}=-\underline{\Psi}_{(1)}(R)$ we derive
\begin{equation}\label{fo27'}
(\underline{D}+\Gamma_{123})(-\ze)-\delta\omb=\omb(-\ze+\eta)+\overline{\thb}(\eta+\ze)+\vab(\overline{\eta}+\overline{\ze})-\xib(\overline{\th}+\om)-\overline{\xib}\va-\underline{\Psi}_{(1)}(R).
\end{equation}
Using $R_{3443}=\Psi_{(0)}(R)+\overline{\,\Psi_{(0)}(R)}$ we derive
\begin{equation}\label{fo28}
D\omb+\underline{D}\om=\overline{\xi}\xib+\xi\overline{\xib}-\overline{\eta}\etab-\eta\overline{\etab}+\ze(\overline{\eta}-\overline{\etab})+\overline{\ze}(\eta-\etab)-(\Psi_{(0)}(R)+\overline{\,\Psi_{(0)}(R)}).
\end{equation}
Using $R_{3421}=\Psi_{(0)}(R)-\overline{\,\Psi_{(0)}(R)}$ we derive
\begin{equation}\label{fo28'}
(\delta-\Gamma_{121})\overline{\ze}-(\overline{\delta}+\Gamma_{122})\ze=(\overline{\va}\vab-\va\overline{\vab})+(\theta\overline{\thb}-\overline{\th}\thb)+\omb(\theta-\overline{\theta})-\om(\thb-\overline{\thb})-(\Psi_{(0)}(R)-\overline{\,\Psi_{(0)}(R)}).
\end{equation}

We derive now the Bianchi identities. Assume $W$ is a real-valued Weyl field, see \eqref{fo4}, and
\begin{equation*}
\D^{\al}W_{\al\be\mu\nu}=J_{\be\mu\nu},
\end{equation*}
for some Weyl current $J\in\mathbf{T}_3^0(\mathbf{M})$. Then, using Proposition \ref{prop:Bianchi-dual},
\begin{equation}\label{bianchi-use}
\D_{[\rho}W_{\al\be]\mu\nu}=\D_\rho W_{\al\be\mu\nu}+\D_{\al}W_{\be\rho\mu\nu}+\D_{\be}W_{\rho\al\mu\nu}=\in_{\sigma\rho\al\be}{{}^\ast J^{\sigma}}_{\mu\nu},
\end{equation}
where
\begin{equation*}
{{}^\ast J^{\sigma}}_{\mu\nu}=\frac{1}{2}{\in_{\mu\nu}}^{\gamma\delta}{J^\sigma}_{\gamma\delta}.
\end{equation*}
Using \eqref{fo4}, we derive the following
\begin{equation}\label{Weylcomp}
\begin{split}
&W_{3141}=W_{3242}=W_{4241}=W_{3231}=0,\\
&W_{4141}=\Psi_{(2)},\quad W_{4242}=\overline{\,\Psi_{(2)}},\quad W_{3131}=\underline{\Psi}_{(2)},\quad W_{3232}=\overline{\,\underline{\Psi}_{(2)}},\\
&W_{2314}=\Psi_{(0)},\quad W_{1324}=\overline{\,\Psi_{(0)}},\quad W_{4343}=W_{1212}=-\Psi_{(0)}-\overline{\,\Psi_{(0)}},\quad W_{1234}=\overline{\,\Psi_{(0)}}-\Psi_{(0)},\\
&W_{1434}=W_{2141}=\Psi_{(1)},\quad W_{2434}=W_{1242}=\overline{\,\Psi_{(1)}},\\
&W_{1343}=W_{2131}={\underline\Psi}_{(1)},\quad W_{2343}=W_{1232}=\overline{\,{\underline\Psi}_{(1)}}.
\end{split}
\end{equation}

We use the table \eqref{Weylcomp} and the formula \eqref{bianchi-use} to derive the Bianchi identities. Using $\D_{[2}W_{41]41}=-J_{414}$ we derive
\begin{equation}\label{Bi1}
(\overline{\delta}+2\Gamma_{122})\Psi_{(2)}-(D+\Gamma_{124})\Psi_{(1)}=-(2\overline{\ze}+\overline{\etab})\Psi_{(2)}+(4\th+\om)\Psi_{(1)}+3\xi\Psi_{(0)}-J_{414}.
\end{equation}
Using $\D_{[2}W_{31]31}=-J_{313}$ we derive
\begin{equation}\label{Bi2}
(\overline{\delta}+2\Gamma_{122})\underline{\Psi}_{(2)}-(\underline{D}+\Gamma_{123})\underline{\Psi}_{(1)}=-(-2\overline{\ze}+\overline{\eta})\underline{\Psi}_{(2)}+(4\thb+\omb)\underline{\Psi}_{(1)}+3\xib\overline{\,\Psi_{(0)}}-J_{313}.
\end{equation}
Using $\D_{[3}W_{41]41}=J_{114}$ we derive
\begin{equation}\label{Bi3}
(\underline{D}+2\Gamma_{123})\Psi_{(2)}-(\delta+\Gamma_{121})\Psi_{(1)}=(2\omb-\overline{\thb})\Psi_{(2)}+(\ze+4\eta)\Psi_{(1)}+3\va\Psi_{(0)}+J_{114}
\end{equation}
Using $\D_{[4}W_{31]31}=J_{113}$ we derive
\begin{equation}\label{Bi4}
(D+2\Gamma_{124})\underline{\Psi}_{(2)}-(\delta+\Gamma_{121})\underline{\Psi}_{(1)}=(2\om-\overline{\th})\underline{\Psi}_{(2)}+(-\ze+4\etab)\underline{\Psi}_{(1)}+3\vab\overline{\,\Psi_{(0)}}+J_{113}.
\end{equation}
Using $\D_{[2}W_{34]41}=-J_{214}$ we derive
\begin{equation}\label{Bi5}
-D\Psi_{(0)}-(\overline\delta+\Gamma_{122})\Psi_{(1)}=-\overline{\vab}\Psi_{(2)}+(2\overline{\etab}+\overline{\ze})\Psi_{(1)}+3\th\Psi_{(0)}+2\xi\overline{\,\underline{\Psi}_{(1)}}-J_{214}.
\end{equation}
Using $\D_{[2}W_{43]31}=-J_{213}$ we derive
\begin{equation}\label{Bi6}
-\underline{D}\overline{\,\Psi_{(0)}}-(\overline\delta+\Gamma_{122})\underline{\Psi}_{(1)}=-\overline{\va}\underline{\Psi}_{(2)}+(2\overline{\eta}-\overline{\ze})\underline{\Psi}_{(1)}+3\thb\overline{\,\Psi_{(0)}}+2\xib\overline{\,\Psi_{(1)}}-J_{213}.
\end{equation}
Using $\D_{[1}W_{42]31}=J_{413}$ we derive
\begin{equation}\label{Bi7}
\delta\overline{\,\Psi_{(0)}}+(D+\Gamma_{124})\underline{\Psi}_{(1)}=-2\vab\overline{\,\Psi_{(1)}}-3\etab\overline{\,\Psi_{(0)}}+(\omega-2\overline{\theta})\underline{\Psi}_{(1)}+\overline{\xi}\underline{\Psi}_{(2)}+J_{413}.
\end{equation}
Using $\D_{[1}W_{32]41}=J_{314}$ we derive
\begin{equation}\label{Bi8}
\delta\Psi_{(0)}+(\underline{D}+\Gamma_{123})\Psi_{(1)}=-2\va\overline{\underline{\Psi}}_1-3\eta\Psi_{(0)}+(\omb-2\overline{\thb})\Psi_{(1)}+\overline{\xib}\Psi_{(2)}+J_{314}.
\end{equation}

\subsection{Symmetries of the formalism}

We discuss now the main symmetries of the formalism introduced in this section. 

{\bf{1. Interchange of the vectors $l$ and $\ul$}}. We define the complex tetrad $(m',\overline{m'},\ul',l')$,
\begin{equation}\label{trans1}
e'_1=m'=m,\quad e'_2=\overline{m'}=\overline{m},\quad e'_3=\ul'=l,\quad e'_4=l'=\ul.
\end{equation}
Using this new complex tetrad we define the scalars $\th',\thb',\va',\vab',\xi',\xib',\eta',\etab',\om',\omb',\ze'$ as in \eqref{fo10}. Given a real-valued Weyl field $W$, we define the scalars $\Psi'_{(2)},\underline{\Psi}'_{(2)},\Psi'_{(1)},\underline{\Psi}'_{(1)},\Psi'_{(0)}$ as in \eqref{fo11}. We define the connection coefficients $\Gamma'_{\mu\al\be}=g(e'_\mu,\D_{e'_\be}e'_\al)$. The definitions show easily that
\begin{equation}\label{tab1}
\begin{split}
&\th'=\thb,\,\thb'=\th,\,\va'=\vab,\,\vab'=\va,\,\xi'=\xib,\,\xib'=\xi,\,\eta'=\etab,\,\etab'=\eta,\,\om'=\omb,\,\omb'=\om,\,\ze'=-\ze,\\
&\Psi'_{(2)}=\underline{\Psi}_{(2)},\,\underline{\Psi}'_{(2)}=\Psi_{(2)},\,\Psi'_{(1)}=\underline{\Psi}_{(1)},\,\underline{\Psi}'_{(1)}=\Psi_{(1)},\,\Psi'_{(0)}=\overline{\,\Psi_{(0)}},\\
&\delta'=\delta,\,\overline{\delta'}=\overline{\delta},\,\underline{D}'=D,\,D'=\underline{D},\,\Gamma'_{121}=\Gamma_{121},\,\Gamma'_{122}=\Gamma_{122},\,\Gamma'_{123}=\Gamma_{124},\,\Gamma'_{124}=\Gamma_{123}.
\end{split}
\end{equation}
The Ricci equations \eqref{fo21}--\eqref{fo28'} and the Bianchi identities \eqref{Bi1}--\eqref{Bi8} are invariant with respect to the transformation \eqref{trans1}. For example, the equation corresponding to \eqref{fo21} in the complex tetrad $(m',\overline{m'},\ul',l')$ reads
\begin{equation*}
(D'+2\Gamma'_{124})\va'-(\delta'+\Gamma'_{121})\xi'=\xi'(2\ze'+\eta'+\etab')-\va'(\om'+\th'+\overline{\th'})-\Psi'_{(2)}(R).
\end{equation*}
After using the table \eqref{tab1}, this is equivalent to
\begin{equation*}
(\underline{D}+2\Gamma_{123})\underline{\va}-(\delta+\Gamma_{121})\xib=\xib(-2\ze+\etab+\eta)-\vab(\omb+\thb+\overline{\thb})-\underline{\Psi}_{(2)}(R),
\end{equation*}
which is \eqref{fo21'}.

{\bf{2. Interchange of the vectors $m$ and $\overline{m}$}}. We define the complex tetrad $(m',\overline{m'},\ul',l')$,
\begin{equation}\label{trans2}
e'_1=m'=\overline{m},\quad e'_2=\overline{m'}=m,\quad e'_3=\ul'=\ul,\quad e'_4=l'=l.
\end{equation}
Using this new complex tetrad we define the scalars $\th',\thb',\va',\vab',\xi',\xib',\eta',\etab',\om',\omb',\ze'$ as in \eqref{fo10}. Given a real-valued Weyl field $W$, we define the scalars $\Psi'_{(2)},\underline{\Psi}'_{(2)},\Psi'_{(1)},\underline{\Psi}'_{(1)},\Psi'_{(0)}$ as in \eqref{fo11}. We define the connection coefficients $\Gamma'_{\mu\al\be}=g(e'_\mu,\D_{e'_\be}e'_\al)$. The definitions show easily that
\begin{equation}\label{tab2}
\begin{split}
&\th'=\overline{\th},\,\thb'=\overline{\thb},\,\va'=\overline{\va},\,\vab'=\overline{\vab},\,\xi'=\overline{\xi},\,\xib'=\overline{\xib},\,\eta'=\overline{\eta},\,\etab'=\overline{\etab},\,\om'=\om,\,\omb'=\omb,\,\ze'=\overline{\ze},\\
&\Psi'_{(2)}=\overline{\,\Psi_{(2)}},\,\underline{\Psi}'_{(2)}=\overline{\,\underline{\Psi}_{(2)}},\,\Psi'_{(1)}=\overline{\,\Psi_{(1)}},\,\underline{\Psi}'_{(1)}=\overline{\,\underline{\Psi}_{(1)}},\,\Psi'_{(0)}=\overline{\,\Psi_{(0)}},\\
&\delta'=\overline{\delta},\,\overline{\delta'}=\delta,\,\underline{D}'=\underline{D},\,D'=D,\,\Gamma'_{121}=\overline{\Gamma_{121}},\,\Gamma'_{122}=\overline{\Gamma_{122}},\,\Gamma'_{123}=\overline{\Gamma_{123}},\,\Gamma'_{124}=\overline{\Gamma_{124}}.
\end{split}
\end{equation}
The Ricci equations \eqref{fo21}--\eqref{fo28'} and the Bianchi identities \eqref{Bi1}--\eqref{Bi8} are invariant with respect to the transformation \eqref{trans2}. For example, the equation corresponding to \eqref{fo21} in the complex tetrad $(m',\overline{m'},\ul',l')$ reads
\begin{equation*}
(D'+2\Gamma'_{124})\va'-(\delta'+\Gamma'_{121})\xi'=\xi'(2\ze'+\eta'+\etab')-\va'(\om'+\th'+\overline{\th'})-\Psi'_{(2)}(R).
\end{equation*}
After using the table \eqref{tab2}, this is equivalent to
\begin{equation*}
(D+2\overline{\Gamma_{124}})\overline{\va}-(\overline{\delta}+\overline{\Gamma_{121}})\overline{\xi}=\overline{\xi}(2\overline{\ze}+\overline{\eta}+\overline{\etab})-\overline{\va}(\om+\overline{\th}+\th)-\overline{\,{\Psi}_{(2)}(R)},
\end{equation*}
which is equivalent to \eqref{fo21} after complex conjugation.

{\bf{3. Rescaling of the null pair $l,\ul$}}. We define the complex tetrad $(m',\overline{m'},\ul',l')$,
\begin{equation}\label{trans3}
e'_1=m'=m,\quad e'_2=\overline{m'}=\overline{m},\quad e'_3=\ul'=A^{-1}\ul,\quad e'_4=l'=A\cdot l,
\end{equation}
for some smooth function $A:\mathbf{N}\to\mathbb{R}\setminus\{0\}$. Using this new complex tetrad we define the scalars $\th',\thb',\va',\vab',\xi',\xib',\eta',\etab',\om',\omb',\ze'$ as in \eqref{fo10}. Given a real-valued Weyl field $W$, we define the scalars $\Psi'_{(2)},\underline{\Psi}'_{(2)},\Psi'_{(1)},\underline{\Psi}'_{(1)},\Psi'_{(0)}$ as in \eqref{fo11}. We define the connection coefficients $\Gamma'_{\mu\al\be}=g(e'_\mu,\D_{e'_\be}e'_\al)$. The definitions show easily that
\begin{equation}\label{tab3}
\begin{split}
&\th'=A\th,\,\thb'=A^{-1}\thb,\,\va'=A\va,\,\vab'=A^{-1}\vab,\,\xi'=A^2\xi,\,\xib'=A^{-2}\xib,\,\eta'=\eta,\,\etab'=\etab,\\
&\Psi'_{(2)}=A^2\Psi_{(2)},\,\underline{\Psi}'_{(2)}=A^{-2}\underline{\Psi}_{(2)},\,\Psi'_{(1)}=A\Psi_{(1)},\,\underline{\Psi}'_{(1)}=A^{-1}\underline{\Psi}_{(1)},\,\Psi'_{(0)}=\Psi_{(0)},\\
&\delta'=\delta,\,\overline{\delta'}=\overline{\delta},\,\underline{D}'=A^{-1}\underline{D},\,D'=AD,\,\Gamma'_{121}=\Gamma_{121},\,\Gamma'_{122}=\Gamma_{122},\\
&\om'=A\om-D(A),\,\omb'=A^{-1}\omb-\underline{D}(A^{-1}),\,\ze'=\ze-\delta(A)/A,\,\Gamma'_{123}=A^{-1}\Gamma_{123},\,\Gamma'_{124}=A\Gamma_{124}.
\end{split}
\end{equation}
The Ricci equations \eqref{fo21}--\eqref{fo28'} and the Bianchi identities \eqref{Bi1}--\eqref{Bi8} are invariant with respect to the transformation \eqref{trans3}. For example, the equation corresponding to \eqref{fo21} in the complex tetrad $(m',\overline{m'},\ul',l')$ reads
\begin{equation*}
(D'+2\Gamma'_{124})\va'-(\delta'+\Gamma'_{121})\xi'=\xi'(2\ze'+\eta'+\etab')-\va'(\om'+\th'+\overline{\th'})-\Psi'_{(2)}(R).
\end{equation*}
After using the table \eqref{tab3}, this is equivalent to
\begin{equation*}
\begin{split}
&(AD+2A\Gamma_{124})(A\va)-(\delta+\Gamma_{121})(A^2\xi)\\
&=A^2\xi(2\ze-2\delta(A)/A+\eta+\etab)-A\va(A\om-D(A)+A\theta+A\overline{\theta})-A^2\Psi_{(2)}(R).
\end{split}
\end{equation*}
This is equivalent to \eqref{fo21}, after simplifying the term $AD(A)\va-2A\delta(A)\xi$ and multiplying by $A^{-2}$.

{\bf{4. Rotation of the vector $m$}}. We define the complex tetrad $(m',\overline{m'},\ul',l')$,
\begin{equation}\label{trans4}
e'_1=m'=Bm,\quad e'_2=\overline{m'}=B^{-1}\overline{m},\quad e'_3=\ul'=\ul,\quad e'_4=l'=l,
\end{equation}
for some smooth function $B:\mathbf{N}\to\mathbb{C}$, $|B|\equiv 1$. Using this new complex tetrad we define the scalars $\th',\thb',\va',\vab',\xi',\xib',\eta',\etab',\om',\omb',\ze'$ as in \eqref{fo10}. Given a real-valued Weyl field $W$, we define the scalars $\Psi'_{(2)},\underline{\Psi}'_{(2)},\Psi'_{(1)},\underline{\Psi}'_{(1)},\Psi'_{(0)}$ as in \eqref{fo11}. We define the connection coefficients $\Gamma'_{\mu\al\be}=g(e'_\mu,\D_{e'_\be}e'_\al)$. The definitions show easily that
\begin{equation}\label{tab4}
\begin{split}
&\th'=\th,\,\thb'=\thb,\,\va'=B^2\va,\,\vab'=B^2\vab,\,\xi'=B\xi,\,\xib'=B\xib,\,\eta'=B\eta,\,\etab'=B\etab,\,\om'=\om,\,\omb'=\omb,\\
&\ze'=B\ze,\,\Psi'_{(2)}=B^2\Psi_{(2)},\,\underline{\Psi}'_{(2)}=B^2\underline{\Psi}_{(2)},\,\Psi'_{(1)}=B\Psi_{(1)},\,\underline{\Psi}'_{(1)}=B\underline{\Psi}_{(1)},\,\Psi'_{(0)}=\Psi_{(0)},\\
&\delta'=B\delta,\,\overline{\delta'}=B^{-1}\overline{\delta},\,\Gamma'_{121}=B\Gamma_{121}-\delta(B),\,\Gamma'_{122}=B^{-1}\Gamma_{122}+\overline{\delta}(B^{-1}),\\
&\,\underline{D}'=\underline{D},\,D'=D,\,\Gamma'_{123}=\Gamma_{123}-\underline{D}(B)/B,\,\Gamma'_{124}=\Gamma_{124}-D(B)/B.
\end{split}
\end{equation}
The Ricci equations \eqref{fo21}--\eqref{fo28'} and the Bianchi identities \eqref{Bi1}--\eqref{Bi8} are invariant with respect to the transformation \eqref{trans4}. For example, the equation corresponding to \eqref{fo21} in the complex tetrad $(m',\overline{m'},\ul',l')$ reads
\begin{equation*}
(D'+2\Gamma'_{124})\va'-(\delta'+\Gamma'_{121})\xi'=\xi'(2\ze'+\eta'+\etab')-\va'(\om'+\th'+\overline{\th'})-\Psi'_{(2)}(R).
\end{equation*}
After using the table \eqref{tab4}, this is equivalent to
\begin{equation*}
\begin{split}
&(D+2\Gamma_{124}-2D(B)/B)(B^2\va)-(B\delta+B\Gamma_{121}-\delta(B))(B\xi)\\
&=B\xi(2B\ze+B\eta+B\etab)-B^2(\om+\th+\overline{\theta})-B^2\Psi_{(2)}(R).
\end{split}
\end{equation*}
This is equivalent to \eqref{fo21}, after simplifying the left-hand side and multiplying by $B^{-2}$.


\begin{thebibliography}{99}
\bibitem{Be-Si} R. Beig and W. Simon,
\textit{On the multipole  expansion for stationary
space-times}, Proceeding of the Royal Society,  series A, {\bf 376}, 
No 1765, (1981), 333-341.
\bibitem{Be-Si2} R. Beig and W. Simon,
\textit{The stationary gravitational field near
spatial infinity}, Gen. Rel. and Grav., {\bf 12}, No 12, (1980), 1003--1013.

\bibitem{Bu-M} G. Buntingand A.K.M. Massood 
ul Alam, \textit{Non-existence of multiple black
holes in asymptotically euclidean static vacuum space-time}, Gen. Rel. Grav. {\bf 19} (1987), 147-154.
 
 \bibitem{Ca-R} B. Carter, \textit{Has the Black Hole Equilibrium Problem 
 Been Solved}, gr-qc/9712038.
 
\bibitem{Ca1} B. Carter, \textit{An axy-symmetric 
black hole has only two degrees of freedom}, Phys. Rev. Letters,
{\bf 26},  (1971) 331-333.


\bibitem{Ch} S. Chandrasekhar, \textit{The mathematical theory of black holes.} International Series of Monographs on Physics, 69, Oxford Science Publications, The Clarendon Press, Oxford University Press, New York (1983).

\bibitem{CKl}  D. Christodoulou and  S. Klainerman, \textit{The global nonlinear stability of the Minkowski space}, Princeton Math. Series 41, Princeton University Press (1993)

\bibitem{Chrusc-Rev} P.T. Chrusciel,\textit{ ``No Hair''  Theorems-Folclore, Conjecture, Results}, Diff. Geom. and Math. Phys.( J. Beem and 
K.L. Duggal) Cont. Math. , {\bf 170}, AMS, Providence, (1994), 23-49,
gr-qc9402032, (1994).

\bibitem{Chrusc-Rev2} P.T. Chrusciel,\textit{Uniqueness of stationary, electro-vaccum black holes revisited}, gr-qc/9610010, v1, (1996).

\bibitem{Chrusc0} P.T. Chrusciel,\textit{
On completeness of orbits of Killing vector fields}, Class. Quantum
Grav. 10 (1993), 2091Ð2101, gr-qc/9304029.
\bibitem{Chrusc} P.T. Chrusciel,\textit{The classification of static vacuum space-times containing an asymptotically flat space-like hypersurface with compact interior} , Class. Quant. Grav. {\bf 16} (1999), 661-687, gr-qc/9809088.
\bibitem{Chrusc2} P.T. Chrusciel, \textit{Uniqueness of black Holes revisited}, Helv. Phys. Acta {\bf 69} (1996),
529--552, Proceedings of Journee
relativistes 1996, Ascona May 1996, gr-qc/ 9610010.

\bibitem{CW} P.T. Chrusciel and R.M. Wald   \textit{ Maximal hypersurfaces in
 A.F. space-times}, Comm. Math. Phys. {\bf 163},   561-164 (1994).
\bibitem{CW2} P.T. Chrusciel and R.M. Wald \textit{On the Topology of Stationary Black Holes}
Class. Quant. Gr. {\bf 10} 1993, 2091-2101.
\bibitem{CGD} P.T. Chrusciel, T. Delay,  G. Galloway
 and R. Howard,  \textit{Regularity of horizon and the area theorem}
Annales H. Poincar\'e, {\bf 2}, (2001), 109-178,
gr-qc000103.


\bibitem{FSW} J.L. Friedman, K. Schleich, D.M. Witt, \textit{Topological censorship},
Phys, Rev Letters, {\bf 71}, 1846-1849  (1993).



\bibitem{FrRaWa} H. Friedrich, I. Racz, R. Wald,
\textit{On the rigidity theorem for spacetimes with a stationary Event or Cauchy horizon}, arXiv:gr/qc9811021, 2002



\bibitem{H-E} S.W. Hawking and G.F.R. Ellis, \textit{The large scale structure of space-time},
 Cambridge Univ. Press, 1973.
\bibitem{Ho} L. H\"{o}rmander, \textit{The analysis of linear partial differential operators IV. Fourier integral operators}, Grundlehren der Mathematischen Wissenschaften [Fundamental Principles of Mathematical Sciences], 275. Springer-Verlag, Berlin (1985).
\bibitem{Ion-K1} A. Ionescu and S. Klainerman,
\textit{Uniqueness results for ill posed 
characteristic problems in curved space-times}, Preprint 2007.
\bibitem{IsMon} J. Isenberg and V. Moncrief, \textit{Symmetries of Cosmological Cauchy Horizons}, Comm. Math. Phys.  vol 89, 387-413, 
(1983).

 
\bibitem{I}W. Israel, \textit{Event horizons in static vacuum space-times}, Phys. Rev. Letters {\bf 164}
(1967), 1776-1779.

\bibitem{KlNi} S. Klainerman and F. Nicol\`o,\textit{ The evolution problem in general relativity.}  Progress in Mathematical Physics, 25. Birkh\"auser Boston, Inc., Boston, MA, (2003).


\bibitem{FS} S. Kobayashi \textit{Transformations groups in differential geometry}, Springer, 1972.



\bibitem{Ma1} M. Mars, A spacetime characterization of the Kerr metric, Classical Quantum Gravity {\bf{16}} (1999), 2507--2523.
\bibitem{Ma2} M. Mars, Uniqueness properties of the Kerr metric, Classical Quantum Gravity {\bf{17}} (2000), 3353--3373.

\bibitem{N-P} E.T. Newman and R. Penrose,
\textit{An approach to gravitational radiation
by a method of spin coefficients}, J. Math. Phys.
{\bf 3} (1962), 566-578.

\bibitem{Rob} D.C. Robinson, \textit{Uniqueness of
the Kerr black hole}, Phys. Rev. Lett. {\bf 34}
(1975), 905-906.
\bibitem{Ra-Wa}I. Racz and R. Wald \textit{Extensions of space-times with Killing horizons}, Class. Quant. Gr.,
{\bf 9} (1992), 2463-2656.
\bibitem{Sim} W. Simon,  \textit{Characterization of the Kerr metric}, Gen. Rel. Grav. {\bf 16} (1984),
465-476.
\bibitem{Su-Wa} D. Sudarski and R.M. Wald,
\textit{Mass formulas for stationary Einstein Yang-Mills black holes and a simple proof of two staticity
theorems}, Phys. Rev {\bf D47} (1993), 5209-5213,
gr-qc /9305023.



\end{thebibliography}
\end{document}